\documentclass{article}

\usepackage{microtype}
\usepackage{graphicx}
\usepackage{booktabs} 

\usepackage{hyperref}
\usepackage{amsmath}

\usepackage{tikzit}

\tikzstyle{basic}=[fill=white, draw=black, shape=circle]
\tikzstyle{square}=[fill=white, draw=black, shape=rectangle]
\tikzstyle{big dashed}=[fill=white, draw=black, shape=circle, minimum width=1cm, dashed]
\tikzstyle{vertical ellipse dashed}=[fill=none, draw=blue, minimum width=0.75cm, minimum height=3cm, ellipse, dashed, tikzit shape=rectangle, tikzit draw=blue, tikzit fill=white]
\tikzstyle{small vertical ellipse dashed}=[fill=none, draw=blue, shape=circle, tikzit fill=white, tikzit draw=blue, dashed, minimum width=0.75cm, minimum height=1.5cm, tikzit shape=rectangle, ellipse]
\tikzstyle{tiny vertical ellipse dashed}=[fill=none, draw=blue, shape=circle, tikzit fill=white, ellipse, dashed, minimum width=0.75cm, minimum height=1cm, tikzit shape=rectangle]
\tikzstyle{red}=[fill=red, draw=black, shape=circle]
\tikzstyle{green}=[fill={rgb,255: red,0; green,128; blue,128}, draw=black, shape=circle]
\tikzstyle{blue}=[fill=blue, draw=black, shape=circle]
\tikzstyle{new_data}=[fill={rgb,255: red,216; green,48; blue,52}, draw={rgb,255: red,216; green,48; blue,52}, shape=circle, minimum width=0.75cm]
\tikzstyle{pale green}=[fill={rgb,255: red,173; green,231; blue,0}, draw=black, shape=circle, minimum width=1cm]
\tikzstyle{horizontal ellipse dashed}=[fill=white, draw=black, tikzit draw=magenta, tikzit shape=rectangle, minimum width=3cm, minimum height=0.75cm, ellipse, dashed]
\tikzstyle{minsize}=[fill={rgb,255: red,189; green,189; blue,189}, draw=none, shape=circle, minimum size=20 pt]
\tikzstyle{horizontal ellipse green}=[fill={rgb,255: red,191; green,255; blue,0}, draw=black, tikzit draw={rgb,255: red,191; green,255; blue,0}, tikzit shape=rectangle, minimum width=3cm, minimum height=0.75cm, ellipse, dashed]
\tikzstyle{horizontal ellipse blue}=[fill={rgb,255: red,107; green,203; blue,255}, draw=black, tikzit draw=blue, tikzit shape=rectangle, minimum width=3cm, minimum height=0.75cm, ellipse, dashed]
\tikzstyle{smallblack}=[fill=black, draw=black, shape=circle, inner sep=0 pt, minimum size=5 pt]
\tikzstyle{smallSquare}=[fill=white, draw=black, shape=rectangle, inner sep=0 pt, minimum size=6 pt]
\tikzstyle{smallCircle}=[fill={rgb,255: red,158; green,202; blue,225}, draw={rgb,255: red,49; green,130; blue,189}, shape=circle, inner sep=0 pt, minimum size=17 pt]
\tikzstyle{smallCircle_white}=[fill={rgb,255: red,49; green,130; blue,189}, draw={rgb,255: red,49; green,130; blue,189}, shape=circle, inner sep=0 pt, minimum size=17 pt]
\tikzstyle{big vertical circle dashed}=[fill=none, draw=red, shape=circle, tikzit shape=circle, dashed, minimum width=1.0cm, line width=1.5pt]
\tikzstyle{z_add}=[fill={rgb,255: red,252; green,146; blue,114}, draw={rgb,255: red,222; green,45; blue,38}, shape=circle, inner sep=0 pt, minimum size=17 pt]
\tikzstyle{z_add_white}=[fill={rgb,255: red,222; green,45; blue,38}, draw={rgb,255: red,222; green,45; blue,38}, shape=circle, inner sep=0 pt, minimum size=17 pt]
\tikzstyle{q_recovered}=[fill={rgb,255: red,255; green,115; blue,182}, draw={rgb,255: red,255; green,115; blue,182}, shape=circle, inner sep=0pt, minimum size=14pt]
\tikzstyle{q_recovered_white}=[fill={rgb,255: red,255; green,115; blue,182}, draw={rgb,255: red,255; green,115; blue,182}, shape=circle, inner sep=0pt, minimum size=14pt]
\tikzstyle{small green}=[fill={rgb,255: red,161; green,217; blue,155}, draw={rgb,255: red,49; green,163; blue,84}, shape=circle, inner sep=0pt, minimum size=17pt, opacity=1]
\tikzstyle{small green_white}=[fill={rgb,255: red,49; green,163; blue,84}, draw={rgb,255: red,49; green,163; blue,84}, shape=circle, inner sep=0pt, minimum size=17pt, opacity=1]
\tikzstyle{med red}=[fill=red, draw=red, shape=circle, inner sep=0pt, minimum size=5pt]
\tikzstyle{med blue}=[fill=blue, draw=blue, shape=circle, inner sep=0pt, minimum size=5pt]
\tikzstyle{med green}=[fill={rgb,255: red,229; green,245; blue,224}, draw={rgb,255: red,229; green,245; blue,224}, shape=circle, inner sep=0pt, minimum size=5pt]
\tikzstyle{med green_white}=[fill={rgb,255: red,229; green,245; blue,224}, draw={rgb,255: red,229; green,245; blue,224}, shape=circle, inner sep=0pt, minimum size=5pt]
\tikzstyle{P}=[fill={rgb,255: red,255; green,123; blue,123}, draw=black, shape=circle, minimum width=3cm]
\tikzstyle{largeblack}=[fill=black, draw=black, shape=circle, minimum size=10pt]
\tikzstyle{new style 0}=[fill=none, draw=red, shape=circle]
\tikzstyle{q_active}=[fill={rgb,255: red,158; green,202; blue,225}, draw={rgb,255: red,49; green,130; blue,189}, shape=circle, minimum size=17pt, inner sep=0 pt, dashed, line width=1pt]
\tikzstyle{q_active_white}=[fill={rgb,255: red,117; green,107; blue,177}, draw={rgb,255: red,117; green,107; blue,177}, shape=circle, minimum size=17pt, inner sep=0 pt, line width=1pt]
\tikzstyle{red_mu_change}=[fill={rgb,255: red,222; green,45; blue,38}, draw={rgb,255: red,222; green,45; blue,38}, shape=circle, minimum size=18pt]
\tikzstyle{leaf node}=[fill=none, draw=none, shape=circle, minimum size=26 pt]
\tikzstyle{leaf node z}=[fill={rgb,255: red,252; green,146; blue,114}, draw={rgb,255: red,222; green,45; blue,38}, shape=circle, minimum size=26 pt]

\tikzstyle{directed}=[->, line width=1pt]
\tikzstyle{undirected}=[-, line width=1pt]
\tikzstyle{directed red}=[draw=red, ->, line width=3pt]
\tikzstyle{directed green}=[draw={rgb,255: red,44; green,162; blue,95}, ->, line width=3.5pt, fill={rgb,255: red,44; green,162; blue,95}]
\tikzstyle{directed blue}=[draw={rgb,255: red,49; green,130; blue,189}, ->, line width=3.5pt, fill={rgb,255: red,49; green,130; blue,189}]
\tikzstyle{directed purple}=[draw={rgb,255: red,128; green,0; blue,128}, ->, line width=1pt]
\tikzstyle{undirected red}=[-, draw=red, line width=1pt]
\tikzstyle{undirected green}=[-, draw={rgb,255: red,0; green,107; blue,61}, line width=1pt]
\tikzstyle{undirected blue}=[-, draw=blue, line width=1pt]
\tikzstyle{undirected purple}=[-, draw={rgb,255: red,128; green,0; blue,128}, line width=1pt]
\tikzstyle{undirected dashed}=[-, line width=1pt, dashed]
\tikzstyle{orange dashed}=[-, draw={rgb,255: red,255; green,128; blue,0}, dashed, line width=1.5pt]
\tikzstyle{directed dash}=[->, dashed, line width=1pt]
\tikzstyle{blue dashed}=[draw=blue, dashed, line width=3pt, ->]
\tikzstyle{green dashed}=[draw={rgb,255: red,49; green,163; blue,84}, dashed, line width=3pt, ->]
\tikzstyle{blue filled}=[-, fill={rgb,255: red,222; green,235; blue,247}, draw=none, line width=1pt, opacity=0.5, tikzit fill=white]
\tikzstyle{red filled}=[-, fill={red!20}, line width=1pt, draw=red, opacity=0.5, tikzit fill=white]
\tikzstyle{green filled}=[-, line width=1pt, draw=none, opacity=0.5, tikzit fill=white, fill={rgb,255: red,229; green,245; blue,224}]
\tikzstyle{orange filled}=[-, fill={orange!20}, draw=orange, line width=1pt, opacity=0.5, tikzit fill=white]
\tikzstyle{undirected dashed}=[-, draw=black, dashed, line width=1pt]
\tikzstyle{thick}=[-, line width=3pt]
\tikzstyle{red dashed}=[line width=2pt, dashed, draw={rgb,255: red,222; green,45; blue,38}, ->]
\tikzstyle{white filled}=[-, draw=white, fill=white]
\tikzstyle{red shade}=[-, fill={rgb,255: red,254; green,224; blue,210}, draw={rgb,255: red,254; green,224; blue,210}]

\usepackage[accepted]{icml2025}

\usepackage{amsmath}
\usepackage{amssymb}
\usepackage{mathtools}
\usepackage{amsthm}
\usepackage{subcaption}

\usepackage[capitalize,noabbrev]{cleveref}

\usepackage[textsize=tiny]{todonotes}

\newcommand{\alglinelabel}{%
  \addtocounter{ALC@line}{-1}
  \refstepcounter{ALC@line}
  \label
}
\usepackage{datetime}
\usepackage{xcolor}
\usepackage[utf8]{inputenc}
\usepackage{graphicx}

\usepackage{tikz}

\usepackage[utf8]{inputenc} 
\usepackage[T1]{fontenc}    
\usepackage{xcolor}
\usepackage{url}            
\usepackage{booktabs}       
\usepackage{amsfonts,amsmath}       
\usepackage{nicefrac}       
\usepackage{microtype}      
\usepackage{graphicx}
\usepackage{wrapfig}
\usepackage{amsthm}
\usepackage{amsmath}
\usepackage{amsfonts}
\usepackage{mathtools}
\usepackage{graphicx}
\usepackage{pdfpages}
\usepackage{tablefootnote}
\usepackage{placeins}

\usepackage{xcolor}
\usepackage{colortbl}
\usepackage{tikz}
\usepackage{wrapfig}
\usepackage{subcaption}
\usepackage{pgffor}
\usepackage{float}
\usepackage{textcomp}
\usepackage{caption}

\usepackage{amssymb}
\usepackage{multicol}
\usepackage{url}

\usepackage{hyperref}

\usepackage{enumitem}

\numberwithin{equation}{section}

\definecolor{newblue}{rgb}{0.2,0.2,0.6} 

\newcommand{\KDE}{\textsf{KDE}}
\newcommand{\tree}{\mathcal{T}}

\newcommand{\graphk}{\mathsf{K}}

\newcommand{\twopartdefow}[3]
{
	\left\{
		\begin{array}{ll}
			#1 & \mbox{if } #2 \\
			#3 & \mbox{otherwise.}
		\end{array}
	\right.
}

\newcommand{\steinar}[1]{\PackageWarning{}{Comment: #1}{\color{orange} [{Steinar:} #1]}}

\newcommand{\lp}{\left (}
\newcommand{\rp}{\right )}

\usepackage{caption}        

\newtheorem{theorem}{Theorem}[section]

\newtheorem{lemma}[theorem]{Lemma}
\newtheorem*{lemma*}{Lemma}
\newtheorem*{theorem*}{Theorem}
\newtheorem{definition}[theorem]{Definition}
\newtheorem{remark}{Remark}

\newtheorem{assumption}{Assumption}[section]

\newcommand{\vol}{\mathrm{vol}}

\newcommand{\Proo}[1]{\mathbb{P}[\,#1\,]}

\newcommand{\R}{\mathbb{R}}

\renewcommand{\deg}{\mathrm{deg}}

\renewcommand{\leq}{\leqslant}
\renewcommand{\geq}{\geqslant}
\renewcommand{\le}{\leqslant}
\renewcommand{\ge}{\geqslant}

\renewcommand{\tilde}{\widetilde}
\renewcommand{\epsilon}{\varepsilon}

\newcommand{\wt}[1]{\widetilde{#1}}

\definecolor{newred}{rgb}{1.0, 0.0, 0.22}

\makeatletter
\@addtoreset{ALC@line}{algorithm}

\renewcommand{\theALC@line}{\arabic{ALC@line}}

\renewcommand{\theHALC@line}{\thealgorithm.\arabic{ALC@line}}
\makeatother

\icmltitlerunning{Dynamic Similarity Graph Construction with Kernel Density Estimation}

\DeclareCaptionFont{ninept}{\fontsize{9pt}{11pt}\selectfont}

\begin{document}

\twocolumn[
\icmltitle{Dynamic Similarity Graph Construction \\ with Kernel Density Estimation}




\begin{icmlauthorlist}
\icmlauthor{Steinar Laenen}{ed}
\icmlauthor{Peter Macgregor}{st}
\icmlauthor{He Sun}{ed}
\end{icmlauthorlist}

 \icmlaffiliation{ed}{University of Edinburgh, United Kingdom}
 
\icmlaffiliation{st}{University of St Andrews, United Kingdom}

\icmlcorrespondingauthor{He Sun}{h.sun@ed.ac.uk}

\icmlkeywords{Machine Learning, ICML}

\vskip 0.3in
]



\printAffiliationsAndNotice{}  

\begin{abstract}
  In the kernel density estimation (\textsc{\textsf{KDE}}) problem, we are given a set  $X$ of data points in $\mathbb{R}^d$, a kernel function $k: \mathbb{R}^d \times \mathbb{R}^d \rightarrow \mathbb{R}$, and a query point $\mathbf{q} \in \mathbb{R}^d$, and the objective is to quickly output an estimate of $\sum_{\mathbf{x} \in X} k(\mathbf{q}, \mathbf{x})$.
In this paper, we consider $\textsf{KDE}$ in the dynamic setting, and introduce a data structure that efficiently maintains the \textit{estimates} for a set of query points as data points are added to $X$ over time.
Based on this, we design a dynamic data structure that maintains a sparse approximation of the fully connected similarity graph on 
$X$, and develop a fast dynamic spectral clustering algorithm.
We further evaluate the effectiveness of our algorithms on both synthetic and real-world datasets.
\end{abstract}

\section{Introduction}\label{sec:introduction}

 Given a  set $X=\{\mathbf{x}_1,\ldots,\mathbf{x}_n\}$ of data points, a set $Q=\{\mathbf{q}_1,\ldots,\mathbf{q}_m\}$ of query points, and a kernel function $k: \mathbb{R}^d \times \mathbb{R}^d \rightarrow \mathbb{R}_{\geq 0}$, the \textsf{KDE} problem is to quickly approximate  
$
 	\mu_{\mathbf{q}} \triangleq \sum_{\mathbf{x}_i\in X}k(\mathbf{q},\mathbf{x}_i)$
for every $\mathbf{q}\in Q$.
As a basic question in computer science, this problem has been actively studied since the 1990s~\citep{LJ} and has comprehensive applications in machine learning and statistics~\citep{bakshi2023subquadratic, genovese2014nonparametric, scholkopf2018learning, schubert2014generalized}.  

In this paper we first study the \textsf{KDE} problem in the dynamic setting, where both the sets of data and query points change over time.  The objective is to dynamically update all the $\textsf{KDE}$ estimates $\mu_{\mathbf{q}}$ for every $\mathbf{q} \in Q$ as data points are added to $X$. 
Building on the framework for static \textsf{KDE} developed by  \citet{charikarKDE},
our algorithm processes: (i) insertions and deletions of query points, and (ii) insertions of data points in $\epsilon^{-2}\cdot n^{0.25 + o(1)}$ time for the Gaussian kernel\footnote{Our algorithm generalises to arbitrary kernel functions, with different powers of $n$ in the update time.}.
In particular, our algorithm maintains $(1 \pm \epsilon)$-approximate estimates of the kernel densities for \emph{every} query point $\mathbf{q} \in Q$ throughout the sequence of data point insertions; see
Theorem~\ref{thm:incremental_dynamic_kde} for the formal statement.  Although it is known that \textsc{KDE} \emph{estimators} can be maintained dynamically~\cite{liang2022dynamic},  to the best of our knowledge, this represents the first dynamic algorithm for the \textsf{KDE} problem that efficiently maintains query \emph{estimates} under data point insertions.

Among its many applications, an efficient algorithm for the   \textsf{KDE} problem on  $X=Q\subset \mathbb{R}^d$
can be used to speed up the construction of a similarity graph for $X$, which is  a key component  used in many graph-based clustering algorithms~(e.g., spectral clustering). These clustering algorithms usually have superior performance over traditional geometric clustering techniques~(e.g., $k$-means)~\citep{nips/NgJW01,sac/Luxburg07}, but in general lack a dynamic implementation.
Our second contribution addresses this challenge, and designs a dynamic algorithm that maintains a similarity graph for the dataset $X$ with expected amortised update time $n^{0.25 + o(1)}$ when new data points are added; see Theorem~\ref{thm:dynamic_cps} for the formal statement.  Our algorithm guarantees that, when the set $X_t$ of data points  at any time $t$ has a cluster structure, our dynamically maintained graph will have the same cluster structure as the fully connected graph on $X_t$; hence a downstream graph clustering algorithm will perform well.

Our  algorithms are experimentally compared against several baseline algorithms on 8 datasets, and these   experiments   confirm the sub-linear update time proven in theory. These experiments further demonstrate that
\begin{itemize}\itemsep 0.3pt
\item our dynamic \textsf{KDE} algorithm scales better to large datasets than several baselines, including the fast static \textsf{KDE} algorithm in \cite{charikarKDE}, and  
\item our dynamic similarity graph construction algorithm runs faster   than the fully-connected and $k$-nearest neighbour similarity graph baselines, and produces comparable clustering results when applying spectral clustering.
\end{itemize}

\paragraph{Related Work.}
Efficient algorithms for the kernel density estimation problem in low dimensions have been known for over two decades~\citep{gray2003nonparametric, LJ, ifgt}.
For the regime of $d = \Omega(\log n )$, there has been some recent progress to develop sub-linear query time algorithms~\citep{charikarKDE, charikar2017hashing, charikar2019multi} based on locality-sensitive hashing~\citep{andoni2008near, datar2004locality} and importance sampling using algorithms for computing approximate nearest neighbours~\citep{backurs2018efficient, karppa2022deann}.
Many sampling based methods incur a factor of $\varepsilon^{-2}$ in their complexity due to concentration bounds like Chebychev's inequality, which is costly for high accuracy (small $\varepsilon$). Techniques based on discrepancy theory and coresets have been developed to mitigate this, achieving $\varepsilon^{-1}$ dependence~\cite{phillips2020near, charikar2024quasi}.
There has also been recent work studying the approximation of kernel similarity graphs in the static setting~\citep{macgregor2024fastkde, quanrud2023spectral}.

Dynamic kernel density estimation has been studied in some restricted settings.
\citet{huang2024dynamic} give a dynamic variant of the fast Gauss transform~\citep{LJ} for low-dimensional data.
Given an initial dataset $X$, \citet{liang2022dynamic} give an efficient algorithm for maintaining a \textsf{KDE} estimator in which some data point $\mathbf{x}_i$ is replaced with a new point $\mathbf{z}$.
In the same setting, \citet{deng2022dynamic} present a dynamic data structure that maintains a spectral sparsifier of the kernel similarity graph for smooth kernels.

Our work also relates to a number of works on incremental spectral clustering~\citep{DHANJAL2014440, klopotek2024eigenvalue, LS24, MartinLV18,Ning, sun2020lifelong, zhou2019incremental}. However, these works usually assume a fixed vertex set~\citep{DHANJAL2014440, MartinLV18, Ning}, are limited to only handling  single edge updates~\citep{LS24}, or do not have  theoretical guarantees on their algorithm performance~\cite{klopotek2024eigenvalue, sun2020lifelong, zhou2019incremental}.

\section{Preliminaries}\label{sec:preliminaries}

 This section lists several facts we   use in the analysis,  and is organised as follows: Section~\ref{sec:lsh} gives a brief introduction to locality sensitive hashing, which we apply  in Section~\ref{sec:KDE} to discuss fast algorithms for Kernel Density Estimation. We informally define an approximate similarity graph in Section~\ref{sec:asg}.

\subsection{Locality Sensitive Hashing} \label{sec:lsh}
Given data $\mathbf{x}_1, \ldots, \mathbf{x}_n \in \R^d$, 
the goal of Euclidean locality sensitive hashing~(LSH) is to preprocess the data in a way such that, given a query point $\mathbf{y} \in \R^d$, we are able to quickly recover the data points close to $\mathbf{y}$. Informally speaking, a family $\mathcal{H}$ of hash functions $H: \R^d \rightarrow \mathbb{Z}$ is \emph{locality sensitive} if
there are values $r \in \R$, $c > 1$, and $p_1 > p_2$,  such that it holds
for $H$ drawn at random from $\mathcal{H}$ that
$
    \Proo{H(\mathbf{u}) = H(\mathbf{v})} \geq p_1$
when $\|\mathbf{u} - \mathbf{v}\| \leq r$, and
$
    \Proo{H(\mathbf{u}) = H(\mathbf{v})} \leq p_2$
when $\|\mathbf{u} - \mathbf{v}\| \geq c \cdot r$.
That is, the collision probability of close points is higher than that of far points.
\citet{datar2004locality} propose a locality sensitive hash family based on random projections, and their technique is further analysed by~\citet{andoni2008near}:

  \begin{lemma}[\cite{andoni2008near}]\label{lem:andoni-indyk}
 	Let $\mathbf{p}$ and $\mathbf{q}$ be any pair of points in $\mathbb{R}^d$. Then, for any fixed $r>0$, there exists a hash family $\mathcal{H}$ such that, if \[p_{\mathrm{near}}\triangleq p_1(r)\triangleq \mathbb{P}_{H\sim\mathcal{H}}[H(\mathbf{p})=H(\mathbf{q}) \mid ||\mathbf{p}-\mathbf{q}||\le r]\] and \[p_{\mathrm{far}}\triangleq p_2(r,c)\triangleq \mathbb{P}_{H\sim\mathcal{H}}[H(\mathbf{p})=H(\mathbf{q}) \mid ||\mathbf{p}-\mathbf{q}||\ge cr]\] for any $c\ge1$, then $$\rho\triangleq \frac{\log 1/p_{\mathrm{near}}}{\log 1/p_{\mathrm{far}}}\le \frac{1}{c^2}+O\left(\frac{\log t}{t^{1/2}}\right),$$ for some $t$, where $p_{\mathrm{near}}\ge \mathrm{e}^{-O(\sqrt{t})}$ and each evaluation takes $d t^{O(t)}$ time.
 \end{lemma}

 	 We follow  \citet{charikarKDE} and use $t=\log^{2/3}n$, which results in $n^{o(1)}$ evaluation time and $\rho=\frac{1}{c^2}+o(1)$. In this case, if $c=O\left(\log^{1/7}n\right)$, then 
 	$\rho^{-1} =c^2(1-\beta)$, for $\beta = o(1)$. 

\begin{definition}[bucket] 
    For any hash function $H: \R^d \rightarrow \mathbb{Z}$ and $\mathbf{x}\in \R^d$, let $B_H(\mathbf{x})$ be the set defined by
    $
    B_H(\mathbf{x}) \triangleq \{ \mathbf{x}' ~|~ H(\mathbf{x}) =H(\mathbf{x}')   \}$;
    we call $B_H(\mathbf{x})$ a bucket. 
\end{definition}

\subsection{Kernel Density Estimation} \label{sec:KDE}

 Given a set $X=\{\mathbf{x}_1,\ldots,\mathbf{x}_n\}$ of data points, a  set $Q=\{\mathbf{q}_1,\ldots,\mathbf{q}_m\}$ of query points, and a kernel function $k: \mathbb{R}^d \times \mathbb{R}^d \rightarrow \mathbb{R}_{\geq 0}$, the \textsf{KDE} problem is to compute 
\(
 	\mu_{\mathbf{q}} \triangleq k(\mathbf{q},X)\triangleq \sum_{\mathbf{x}_i\in X}k(\mathbf{q},\mathbf{x}_i)
\)
for every $\mathbf{q}\in Q$. We assume that\footnote{We make this assumption simply for the ease of our presentation, and setting $\mu_\mathbf{q} \geq \zeta$ for any constant $\zeta$ instead will not influence the asymptotic results of our work.}  $1 \leq \mu_\mathbf{q} \leq n$. While a direct computation of the $m$ values for every $\mathbf{q}\in Q$ requires $m n d$ operations, 
there are a number of works that develop faster algorithms for approximating these $m$ quantities.

  Our designed algorithms are based on the work of Charikar, Kapralov, Nouri, and Siminelakis~\citep{charikarKDE}. We refer to their algorithm as CKNS, and provide a brief overview. At a high level, the  CKNS algorithm is  based on importance sampling and,  for any query point $\mathbf{q}$, their objective is to sample a data point $\mathbf{x}_i \in X$ with probability approximately proportional to  $k(\mathbf{q}, \mathbf{x}_i)$. To achieve this, they introduce the notion of  \emph{geometric weight levels}   $\{L_j^{\mathbf{q}}\}_j$  defined as follows:
  
\begin{definition}[\cite{charikarKDE}]
\label{def:geometric_weight_levels}
For any query point $\mathbf{q}$, let $J_{\mu_\mathbf{q}} \triangleq \left\lceil\log \frac{2 n}{\mu_\mathbf{q}}\right\rceil$, and for $j \in [J_{\mu_\mathbf{q}}]$, let
 $L^{\mathbf{q}}_j\triangleq \left\{\mathbf{x}_i \in X: k(\mathbf{q}, \mathbf{x}_i) \in \left(2^{-j},2^{-j+1}\right] \right\}$.
  We define the corresponding distance levels as
  \[
    r_j = \max_{ \substack{ \mathbf{x},\mathbf{x}': \\ k(\mathbf{x},\mathbf{x}')\in \left(2^{-j},2^{-j+1}\right]} }  \|\mathbf{x}-\mathbf{x}'\|
  \]
  for any $j\in [J_{\mu_\mathbf{q}}]$,
 and define $$L^{\mathbf{q}}_{J_{\mu_\mathbf{q}}+1}\triangleq  X \setminus \left(\bigcup_{j\in[J_{\mu_\mathbf{q}}]}L^{\mathbf{q}}_j\right).$$
 \end{definition}

These $L_j^{\mathbf{q}}$'s for any query point $\mathbf{q}$ partition the data points into groups based on the kernel distances $k(\mathbf{q}, \mathbf{x}_i)$, progressing geometrically away from $\mathbf{q}$. Their key insight is that the number of data points in each level $L_j^{\mathbf{q}}$ is bounded:  
\begin{lemma}[\cite{charikarKDE}]\label{lem:sizeL}
 	It holds for any query point $\mathbf{q}$ and $j\in [J_{\mu_\mathbf{q}}]$    that $\left |L^{\mathbf{q}}_j\right |\le 2^{j}\cdot \mu_\mathbf{q}.$
 \end{lemma}
Hence,  one can sub-sample the  data points  with probability \(1 / (2^j \cdot \mu_\mathbf{q})\) for every $j\in [J_{\mu_\mathbf{q}}]$, and the sampled data points are stored in hash buckets using LSH. This data structure will allow for fast and good estimation of $\mu_\mathbf{q}$ for any query point $\mathbf{q}$.
We further follow \citet{charikarKDE}, and introduce the cost of a kernel $k$.

 \begin{lemma}[\cite{charikarKDE}]\label{lem:collprob}
 	Assume that kernel  $k$ induces weight level sets $L^{\mathbf{q}}_j$'s and corresponding distance levels $r_j$'s. Also, for any query $\mathbf{q}$,   integer  $i\in[J_{\mu_\mathbf{q}}+1]$, and $j \in [J_{\mu_\mathbf{q}}]$ satisfying $i>j$, let $\mathbf{p}\in L^{\mathbf{q}}_j$ and $\mathbf{p}'\in L^{\mathbf{q}}_i$.  Assuming  that $\mathcal{H}$ is an Andoni-Indyk LSH family designed for near distance $r_j$ (see Lemma~\ref{lem:andoni-indyk}), the following holds for  any integer $k\ge 1$: 
 	\begin{enumerate}
 		\item $\mathbb{P}_{H^*\sim\mathcal{H}^k}\left[ H^*(\mathbf{p})=H^*(\mathbf{q})\right] \ge p_{\mathrm{near},j}^k$,
 		\item $\mathbb{P}_{H^*\sim\mathcal{H}^k}\left[ H^*(\mathbf{p}')=H^*(\mathbf{q})\right] \le p_{\mathrm{near},j}^{kc^2(1-\beta)}$,
 	\end{enumerate}
 where $c\triangleq c_{i,j}\triangleq \min\left\{\frac{r_{i-1}}{r_j},\log^{1/7}n\right\}$, $p_{\mathrm{near},j} \triangleq p_1(r_j)$, and $\beta = o(1)$ from Lemma~\ref{lem:andoni-indyk}.  
 \end{lemma}

 \begin{definition}[Cost of a Kernel]\label{def:cost_kernel}
	 Suppose that a kernel $k$ induces distance levels $r_j$'s based on the kernel value $\mu_\mathbf{q}$ (see Definition~\ref{def:geometric_weight_levels}). For any $j\in[J_{\mu_\mathbf{q}}]$ we define the \emph{cost} of kernel $k$ for weight level $L^{\mathbf{q}}_j$ as 
		\begin{align*}
	\mathrm{cost}_{\mu_{\mathbf{q}}}(k,j)&\triangleq  \exp_2\left(\max_{i=j+1,\ldots,J_{\mu_\mathbf{q}}+1}  \left\lceil  \frac{i-j}{c_{i,j}^2(1-\beta)}     \right\rceil\right),
	\end{align*}
	where $c_{i,j}\triangleq \min\left\{\frac{r_{i-1}}{r_j},\log^{1/7}n\right\}$ and $\beta=o(1)$ from Lemma~\ref{lem:andoni-indyk}. We define the general \emph{cost} of a kernel $k$ as
	\(
	\mathrm{cost}(k)\triangleq \max_{\mu_\mathbf{q}, j\in[J_{\mu_\mathbf{q}}]} \mathrm{cost}_{\mu_{\mathbf{q}}}(k,j).
	\)
        For any $j \in [J_{\mu_\mathbf{q}}]$ we further define
        \begin{equation}\label{eq:kj} 
        k_j\triangleq - \frac{1}{ \log p_{\mathrm{near},j}} \cdot\max_{i=j+1,\ldots,J_{\mu_\mathbf{q}}+1}  \left\lceil  \frac{i-j}{c_{i,j}^2(1-\beta)}     \right\rceil.
        \end{equation}
\end{definition}

By the assumption that $1 \leq \mu_\mathbf{q} \leq n$, the cost of some popular kernels such as the Gaussian kernel $k_{\mathrm{g}}$, the $t$-student kernel $k_{t}$, and the exponential kernel $k_e$ are $\mathrm{cost(k_g)} = n^{(1+o(1)) \frac{1}{4}}$, $\mathrm{cost(k_t)} = n^{o(1)}$, and $\mathrm{cost(k_e)} = n^{(1+o(1))\frac{4}{27}}$,  respectively~\citep{charikarKDE}.

\subsection{Approximate Similarity Graphs}\label{sec:asg}

 Constructing a similarity graph from a set of data points   
is the first step of
most modern clustering algorithms.  For   any set $X = \{\mathbf{x}_1, \ldots, \mathbf{x}_n\}$ of data points in $\mathbb{R}^d$ and kernel function $k : \mathbb{R}^{d} \times \mathbb{R}^{d} \rightarrow \mathbb{R}_{\geq 0}$,  a similarity graph $F=(V,E,w)$ from $X$ can be   constructed as follows:  
  each $\mathbf{x}_i \in X$ is a vertex in $F$, and every  pair of vertices $\mathbf{x}_i$ and $\mathbf{x}_j$ is connected by an edge with weight $
w(\mathbf{x}_i, \mathbf{x}_j) = k(\mathbf{x}_i, \mathbf{x}_j)$. 
While this graph $F$ has $\Theta(n^2)$ edges by definition, we  can construct in $\widetilde{O}(n)$ time a sparse graph $G$ with $\widetilde{O}(n)$ edges  such that (i) every cluster in $F$ has  low conductance in $G$,  and (ii) the eigenvalue gaps of the normalised Laplacian matrices of $F$ and $G$ are approximately the same~\citep{macgregor2024fastkde}; 
these two conditions  ensure that a typical   clustering algorithm on $F$ and $G$   returns approximately the same result. We call such a  sparse graph $G$ an \emph{approximate similarity graph}, and refer the reader to Section~\ref{sec:background} in the appendix for its formal definition.

\subsection{Convention \& Assumption}

For ease of presentation, for any set $X\subset \mathbb{R}^d$ and $\mathbf{z}\in \mathbb{R}^d$, we always use $X\cup \mathbf{z}$ and $X\setminus \mathbf{z}$ to represent $X\cup \{\mathbf{z}\}$ and $X\setminus \{\mathbf{z}\}$. For a similarity graph $F$ constructed for any set $X=\{\mathbf{x}_1,\ldots, \mathbf{x}_n\}\subset \mathbb{R}^d$, 
we   use $\mathbf{x}_i$ to represent both the point in $\mathbb{R}^d$ and the corresponding vertex in $F$, as long as the underlying meaning of $\mathbf{x}_i$ is clear from context. We use $(\mathbf{x}_i, \mathbf{x}_j)$ to represent an edge with $\mathbf{x}_i$ and $\mathbf{x}_j$ as the endpoints, and we only consider undirected graphs. We use $\widetilde{O}(n)$ to represent $O(n\cdot\log^c n)$ for some constant $c$. 
The  $\log$ operator takes the base $2$.

\begin{assumption} \label{remark:upper_bound_points_added}
 Let $n_1 = |X|$ denote the number of data points at initialisation. We assume that, if $X_t\subset \mathbb{R}^d$ represents the set of data points  after $t$ updates, then $|X_t|  \leq n_1^\gamma$ for constant $\gamma > 0$. Moreover, based on  the  JL Lemma~\citep{johnson1984extensions}, we  always assume that   $d=O(\log |X_t|)$, and hence  our work ignores the dependency on $d$ in the algorithms' runtime.  
\end{assumption}

\section{Dynamic Kernel Density Estimation\label{sec:dynamic_kde}}

In this section we design a data structure to   dynamically maintain \textsf{KDE} estimates as new data 
and query points are added or removed over time. Our data structure supports $\textsc{Initialise}(X, Q, \varepsilon)$, which creates a hash data structure for the \textsf{KDE} estimates  based on $X$, and  supports operations for dynamically maintaining the data and query point sets as well as the corresponding estimates.  Our key technical contribution is the design of a \emph{query hash} data structure, in which we carefully store relevant hash values for query points, enabling efficient updates to query estimates as data points are added. The main components used in updating the data structure and their performance are  as follows:  

\begin{figure*}[h]
\vskip 0.2in
    \begin{center}
    \centerline{\resizebox{0.9\linewidth}{!}{%
    \begin{tikzpicture}
	\begin{pgfonlayer}{nodelayer}
		\node [style=none] (165) at (4.5, 4.5) {};
		\node [style=none] (166) at (5.5, 5.5) {};
		\node [style=none] (167) at (12.5, 5.5) {};
		\node [style=none] (168) at (13.5, 4.5) {};
		\node [style=none] (169) at (13.5, -9.5) {};
		\node [style=none] (170) at (12.5, -10.5) {};
		\node [style=none] (171) at (5.5, -10.5) {};
		\node [style=none] (172) at (4.5, -9.5) {};
		\node [style=none] (157) at (-8.5, 4.5) {};
		\node [style=none] (158) at (-7.5, 5.5) {};
		\node [style=none] (159) at (-0.5, 5.5) {};
		\node [style=none] (160) at (0.5, 4.5) {};
		\node [style=none] (161) at (0.5, -9.5) {};
		\node [style=none] (162) at (-0.5, -10.5) {};
		\node [style=none] (163) at (-7.5, -10.5) {};
		\node [style=none] (164) at (-8.5, -9.5) {};
		\node [style=none] (0) at (-7.25, 2.75) {};
		\node [style=none] (1) at (-5.25, 2.75) {};
		\node [style=none] (2) at (-3.25, 2.75) {};
		\node [style=none] (3) at (-1.25, 2.75) {};
		\node [style=none] (6) at (-7, 0.5) {};
		\node [style=none] (7) at (-5.5, 0.5) {};
		\node [style=none] (8) at (-5, 0.5) {};
		\node [style=none] (9) at (-3.5, 0.5) {};
		\node [style=none] (10) at (-3, 0.5) {};
		\node [style=none] (11) at (-1.5, 0.5) {};
		\node [style=none] (12) at (5.75, 2.75) {};
		\node [style=none] (13) at (7.75, 2.75) {};
		\node [style=none] (14) at (9.75, 2.75) {};
		\node [style=none] (15) at (11.75, 2.75) {};
		\node [style=none] (16) at (6, 0.5) {};
		\node [style=none] (17) at (7.5, 0.5) {};
		\node [style=none] (18) at (8, 0.5) {};
		\node [style=none] (19) at (9.5, 0.5) {};
		\node [style=none] (20) at (10, 0.5) {};
		\node [style=none] (21) at (11.5, 0.5) {};
		\node [style=small green] (22) at (-2.25, 1.5) {$\mathbf{x}_{6}$};
		\node [style=small green] (24) at (-4.25, 1.5) {$\mathbf{x}_4$};
		\node [style=small green] (25) at (-6.25, 1.5) {$\mathbf{x}_3$};
		\node [style=smallCircle] (27) at (6.75, 1.5) {$\mathbf{q}_1$};
		\node [style=none] (78) at (-4.25, -9.5) {$B_{H_{\mu_i, a, j, \ell}}(\mathbf{z})$};
		\node [style=none] (79) at (8.75, -9.5) {$B^*_{H_{\mu_i, a, j, \ell}}(\mathbf{z})$};
		\node [style=none] (86) at (-7.25, -5) {};
		\node [style=none] (87) at (-5.25, -5) {};
		\node [style=none] (88) at (-3.25, -5) {};
		\node [style=none] (89) at (-1.25, -5) {};
		\node [style=none] (90) at (-7, -7.25) {};
		\node [style=none] (91) at (-5.5, -7.25) {};
		\node [style=none] (92) at (-5, -7.25) {};
		\node [style=none] (93) at (-3.5, -7.25) {};
		\node [style=none] (94) at (-3, -7.25) {};
		\node [style=none] (95) at (-1.5, -7.25) {};
		\node [style=small green] (96) at (-2.25, -6.25) {$\mathbf{x}_{6}$};
		\node [style=small green] (98) at (-4.25, -6.25) {$\mathbf{x}_4$};
		\node [style=small green] (99) at (-6.25, -6.25) {$\mathbf{x}_3$};
		\node [style={z_add}] (101) at (-4.25, -4.75) {$\mathbf{z}$};
		\node [style=smallCircle] (103) at (8.75, 3) {$\mathbf{q}_3$};
		\node [style=smallCircle] (104) at (8.75, 1.5) {$\mathbf{q}_2$};
		\node [style=smallCircle] (105) at (10.75, 1.5) {$\mathbf{q}_4$};
		\node [style=none] (106) at (5.75, -5) {};
		\node [style=none] (107) at (7.75, -5) {};
		\node [style=none] (108) at (9.75, -5) {};
		\node [style=none] (109) at (11.75, -5) {};
		\node [style=none] (110) at (6, -7.25) {};
		\node [style=none] (111) at (7.5, -7.25) {};
		\node [style=none] (112) at (8, -7.25) {};
		\node [style=none] (113) at (9.5, -7.25) {};
		\node [style=none] (114) at (10, -7.25) {};
		\node [style=none] (115) at (11.5, -7.25) {};
		\node [style=smallCircle] (116) at (6.75, -6.25) {$\mathbf{q}_1$};
		\node [style={q_active}] (118) at (8.75, -4.75) {$\mathbf{q}_3$};
		\node [style=smallCircle] (119) at (8.75, -6.25) {$\mathbf{q}_2$};
		\node [style=smallCircle] (120) at (10.75, -6.25) {$\mathbf{q}_4$};
		\node [style={q_active}] (142) at (-19.5, -6) {$\mathbf{q}_3$};
		\node [style={z_add}] (149) at (-18, -6) {$\mathbf{z}$};
		\node [style=none] (173) at (9, 6.5) {\textbf{Query hash}};
		\node [style=none] (174) at (-4, 6.5) {\textbf{Data hash}};
		\node [style=none] (175) at (14, -2.25) {};
		\node [style=none] (176) at (-20, -2.25) {};
		\node [style=none] (177) at (-22, 2.5) {\Large{$t$}};
		\node [style=none] (178) at (-22, -7) {\Large{$t'$}};
		\node [style=small green] (179) at (-13.25, -9) {$\mathbf{x}_{6}$};
		\node [style=small green] (180) at (-19.5, -3.75) {$\mathbf{x}_4$};
		\node [style=small green] (181) at (-12.25, -4.25) {$\mathbf{x}_3$};
		\node [style=smallCircle] (182) at (-13.5, -5) {$\mathbf{q}_1$};
		\node [style=smallCircle] (184) at (-18, -3.5) {$\mathbf{q}_2$};
		\node [style=smallCircle] (185) at (-14.5, -8.25) {$\mathbf{q}_4$};
		\node [style=none] (186) at (-16, -3.5) {};
		\node [style=none] (187) at (-16, -10.5) {};
		\node [style=none] (188) at (-19.5, -7) {};
		\node [style=none] (189) at (-12.5, -7) {};
		\node [style=med green] (194) at (-14.75, -6.25) {{\color{gray}\small{$\mathbf{x}_5$}}};
		\node [style=med green] (195) at (-18.5, -9) {{\color{gray}\small{$\mathbf{x}_2$}}};
		\node [style=med green] (196) at (-16.75, -5.5) {{\color{gray}\small{$\mathbf{x}_1$}}};
		\node [style=none] (199) at (-18, -7) {};
		\node [style=none] (200) at (-18, -5) {};
		\node [style=none] (201) at (-19, -6) {};
		\node [style=none] (202) at (-17, -6) {};
		\node [style=none] (203) at (-16, -6) {};
		\node [style=none] (204) at (-18, -4) {};
		\node [style=none] (205) at (-20, -6) {};
		\node [style=none] (206) at (-18, -8) {};
		\node [style=none] (207) at (-16, -11.75) {(a)};
		\node [style=none] (208) at (-4, -11.75) {(b)};
		\node [style=none] (209) at (9, -11.75) {(c)};
		\node [style=smallCircle] (210) at (-19.5, 3.5) {$\mathbf{q}_3$};
		\node [style=small green] (213) at (-13.25, 0.5) {$\mathbf{x}_{6}$};
		\node [style=small green] (214) at (-19.5, 5.75) {$\mathbf{x}_4$};
		\node [style=small green] (215) at (-12.25, 5.25) {$\mathbf{x}_3$};
		\node [style=smallCircle] (216) at (-13.5, 4.5) {$\mathbf{q}_1$};
		\node [style=smallCircle] (217) at (-18, 6) {$\mathbf{q}_2$};
		\node [style=smallCircle] (218) at (-14.5, 1.25) {$\mathbf{q}_4$};
		\node [style=none] (219) at (-16, 6) {};
		\node [style=none] (220) at (-16, -1) {};
		\node [style=med green] (223) at (-14.75, 3.25) {{\color{gray}\small{$\mathbf{x}_5$}}};
		\node [style=med green] (224) at (-18.5, 0.5) {{\color{gray}\small{$\mathbf{x}_2$}}};
		\node [style=med green] (225) at (-16.75, 4) {{\color{gray}\small{$\mathbf{x}_1$}}};
		\node [style=none] (226) at (-15, 6) {$\mathbb{R}^d$};
		\node [style=none] (227) at (-19.5, 2.5) {};
		\node [style=none] (228) at (-12.5, 2.5) {};
		\node [style=none] (229) at (-15, -3.5) {$\mathbb{R}^d$};
	\end{pgfonlayer}
	\begin{pgfonlayer}{edgelayer}
		\draw [style=red shade] (206.center)
			 to [bend right=45] (203.center)
			 to [bend right=45] (204.center)
			 to [bend right=45] (205.center)
			 to [bend right=45] cycle;
		\draw [style=blue filled] (167.center)
			 to [bend left=45] (168.center)
			 to (169.center)
			 to [bend left=45, looseness=0.75] (170.center)
			 to (171.center)
			 to [bend left=45] (172.center)
			 to (165.center)
			 to [bend left=45, looseness=1.25] (166.center)
			 to cycle;
		\draw [style=green filled] (159.center)
			 to [bend left=45] (160.center)
			 to (161.center)
			 to [bend left=45, looseness=0.75] (162.center)
			 to (163.center)
			 to [bend left=45] (164.center)
			 to (157.center)
			 to [bend left=45, looseness=1.25] (158.center)
			 to cycle;
		\draw [style=undirected] (0.center) to (6.center);
		\draw [style=undirected] (6.center) to (7.center);
		\draw [style=undirected] (7.center) to (1.center);
		\draw [style=undirected] (1.center) to (8.center);
		\draw [style=undirected] (8.center) to (9.center);
		\draw [style=undirected] (9.center) to (2.center);
		\draw [style=undirected] (2.center) to (10.center);
		\draw [style=undirected] (10.center) to (11.center);
		\draw [style=undirected] (11.center) to (3.center);
		\draw [style=undirected] (12.center) to (16.center);
		\draw [style=undirected] (16.center) to (17.center);
		\draw [style=undirected] (17.center) to (13.center);
		\draw [style=undirected] (13.center) to (18.center);
		\draw [style=undirected] (18.center) to (19.center);
		\draw [style=undirected] (19.center) to (14.center);
		\draw [style=undirected] (14.center) to (20.center);
		\draw [style=undirected] (20.center) to (21.center);
		\draw [style=undirected] (21.center) to (15.center);
		\draw [style=undirected] (86.center) to (90.center);
		\draw [style=undirected] (90.center) to (91.center);
		\draw [style=undirected] (91.center) to (87.center);
		\draw [style=undirected] (87.center) to (92.center);
		\draw [style=undirected] (92.center) to (93.center);
		\draw [style=undirected] (93.center) to (88.center);
		\draw [style=undirected] (88.center) to (94.center);
		\draw [style=undirected] (94.center) to (95.center);
		\draw [style=undirected] (95.center) to (89.center);
		\draw [style=undirected] (106.center) to (110.center);
		\draw [style=undirected] (110.center) to (111.center);
		\draw [style=undirected] (111.center) to (107.center);
		\draw [style=undirected] (107.center) to (112.center);
		\draw [style=undirected] (112.center) to (113.center);
		\draw [style=undirected] (113.center) to (108.center);
		\draw [style=undirected] (108.center) to (114.center);
		\draw [style=undirected] (114.center) to (115.center);
		\draw [style=undirected] (115.center) to (109.center);
		\draw [style=undirected dashed] (176.center) to (175.center);
		\draw [style=directed] (187.center) to (186.center);
		\draw [style=white filled] (202.center)
			 to [bend right=45] (200.center)
			 to [bend right=45] (201.center)
			 to [bend right=45] (199.center)
			 to [bend right=45] cycle;
		\draw [style=directed] (188.center) to (189.center);
		\draw [style=directed] (220.center) to (219.center);
		\draw [style=directed] (227.center) to (228.center);
	\end{pgfonlayer}
\end{tikzpicture}
    }}
\caption{Illustration of   $\textsc{AddDataPoint}(\mathbf{z})$   for a single iteration $\mu_i \in M$, $a\in K_1$, $j \in [J_{\mu_i}]$, and $\ell \in [K_{2}]$. The first row illustrates  (a) the subsampled data points $Z \triangleq \{\mathbf{x}_3, \mathbf{x}_4, \mathbf{x}_6\}$ and query points $Q_{\mu_i} \triangleq \{\mathbf{q}_i\}_{i=1}^{4}$, (b) the bucketing of $Z$ by the hash function $H_{\mu_i, a, j, \ell}$,  as well as (c) the bucketing of $Q$ by the same hash function. The second row illustrates  (a) the relative location of a new arriving  data point $\mathbf{z}\in\mathbb{R}^d$, with shaded red region indicating $L_j^{\mathbf{z}}$,  (b) $\mathbf{z}$'s inclusion in the bucket $B_{H_{\mu_i, a, j, \ell}}(\mathbf{z})$, as well as   (c) the recovery of $\mathbf{q}_3 \in B^*_{H_{\mu_i, a, j, \ell}}(\mathbf{z})$ because $\mathbf{z} \in L_j^{\mathbf{q}_3}$.}\label{fig:update_kde}
 \end{center}
\vskip -0.2in
\end{figure*}
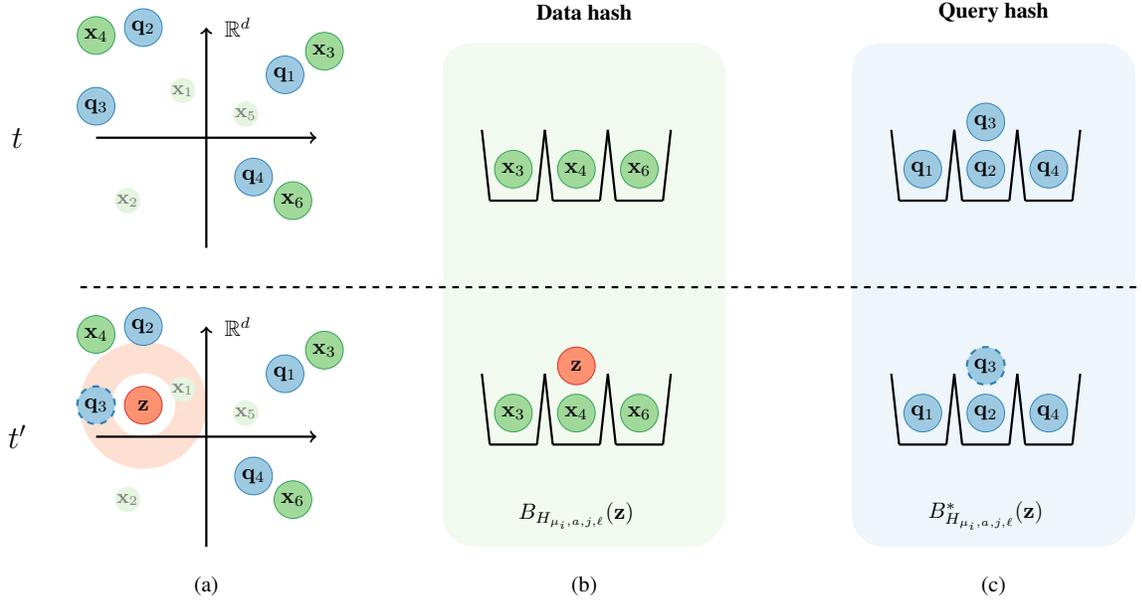

\begin{theorem}[\textbf{Main Result~1}]\label{thm:incremental_dynamic_kde}
 Let $k$ be a kernel function with $\mathrm{cost}(k)$ as defined in Definition~\ref{def:cost_kernel},  and   $X \subset \mathbb{R}^d$   a set of $n$ data points updated through data point insertions. Assuming   $Q = \emptyset$ initially\footnote{When $Q \neq \emptyset$ with $|Q| \triangleq m_1$, we have an  additional additive factor of $m_1 \cdot \epsilon^{-2} \cdot n_1^{o(1)} \cdot \mathrm{cost}(k)$ and $m_1 \cdot \epsilon^{-2} \cdot n^{o(1)} \cdot \mathrm{cost}(k)$ in the running time of the initialisation and data point update steps  respectively.}, the performance of 
the procedures  in Algorithm~\ref{alg:dynamic_kde} is as follows: 
\begin{itemize}
\item \underline{Initialisation}:  $\textsc{Initialise}(X, \emptyset, \varepsilon)$  creates a hash data structure for the \textsf{KDE}, and  runs in time  $\epsilon^{-2} \cdot n_1^{1+o(1)} \cdot \mathrm{cost}(k)$, where $n_1$ is the number of data points at initialisation.
\item \underline{Query Point Updates}: For every query point insertion $Q\leftarrow Q\cup \mathbf{q}$  and deletion   $Q\leftarrow Q\setminus \mathbf{q}$,   $\textsc{AddQueryPoint}(\mathbf{q})$ and $\textsc{DeleteQueryPoint}(\mathbf{q})$ update the corresponding sets and data structures.
Moreover, $\textsc{AddQueryPoint}(\mathbf{q})$   returns $\hat{\mu}_{\mathbf{q}}$ that achieves  a $(1\pm \varepsilon)$-multiplicative factor approximation of $\mu_{\mathbf{q}}$ with high probability.
\item \underline{Data Point Updates}: For every data point insertion $X\leftarrow X \cup \mathbf{z}$, $\textsc{AddDataPoint}(\mathbf{z})$ updates the corresponding sets and data structures, and returns the updated estimates $\hat{\mu}_\mathbf{q}$ that achieve $(1\pm \epsilon)$-multiplicative factor approximations of $\mu_\mathbf{q}$  for every maintained query point $\mathbf{q} \in Q$.
\end{itemize}
With high probability, the amortised running time for each update procedure is $\epsilon^{-2} \cdot n^{o(1)} \cdot \mathrm{cost}(k)$, where $n = |X|$ is the current number of data points. 
\end{theorem}

To examine the significance of Theorem~\ref{thm:incremental_dynamic_kde}, notice that the amortised update time $\epsilon^{-2} \cdot n^{o(1)} \cdot \mathrm{cost}(k) $ for the data point insertions is \emph{independent} of the number of query points $|Q|$. This makes our algorithm significantly more efficient than re-estimating the query points after every update. While previous work~\citep{liang2022dynamic} has shown that the CKNS \textsf{KDE} \emph{estimator} can be extended to the dynamic setting, our result shows that the \emph{estimates} of a set of query points can be efficiently updated.

\subsection{Analysis for the Initialisation}
 
The initialisation step prepares all the   data structures used for   subsequent data and query point updates. The main component used in  $\textsc{Initialise}(X, Q, \varepsilon)$ is the $\textsc{Preprocess}(X,\epsilon)$   procedure,
 which preprocesses  the data points in $X$ to ensure that the value $\mu_\mathbf{q}$ for any query point $\mathbf{q}$ can be fast approximated. To achieve this, $\textsc{Preprocess}(X,\epsilon)$ defines \[M \triangleq \left\{2^{k} \mid k \in \mathbb{Z}, 0 \leq k \leq \log\left( 2n_1 \right)\right\},
 \]
and indexes every $\mu_i \in M$ such that $\mu_0 \leq \ldots \leq \mu_{\log(2n_1)}$; note that $\mu_i = 2^{i}$. Then for $\mu_i \in M$ and $j \in [ \log (2\cdot n_1 / \mu_i)]$ it samples every data point in $X$ with probability 
$\min\left\{1 / (2^{j+1}\mu_i),1\right\}$, and employs a hash function $H_{\mu_i, a, j, \ell}$ chosen uniformly at random from  $\mathcal{H}^{k_j}$ with  $k_j=\tilde{O}(1)$~(cf.~Lemma~\ref{lem:collprob}) to add every sampled $\mathbf{x}\in X$ to the buckets $\{B_{H_{\mu_i, a, j, \ell}}(\mathbf{x}) \}_{\mu_i, a,j,\ell}$ indexed by all the possible $a\in[K_1]$ with $K_1 = O(\log n_1 \cdot \varepsilon^{-2})$, and $\ell \in [K_{2}]$  with $K_{2} = O(\log(n_1)\cdot \mathrm{cost}(k))$. In addition,  $\textsc{Preprocess}(X,\epsilon)$ samples every data point in $X$ with probability $1/(2n_1)$ for all the possible values of $\mu_i$ and $a$, and adds the sampled points to set $\{\wt{X}_{\mu_i, a}\}_{\mu_i, a}$.  We remark that our described $\textsc{Preprocess}(X,\epsilon)$ is almost the same as the   one presented in   \citet{charikarKDE} and,  although this data structure is sufficient to quickly output $\textsf{KDE}$ estimates, we need to store additional \emph{query-hash} buckets to update estimates when new data points arrive.

\subsection{Analysis for Updates}\label{sec:kde_incremental}

 When a new query point $\mathbf{q}$ arrives, $\textsc{AddQueryPoint}(\mathbf{q})$ performs the following operations: 
\begin{enumerate}
\item It computes  the \textsf{KDE} estimate $\hat{\mu}_\mathbf{q}$ of $\mu_{\mathbf{q}}$ using the hash-based data structure from   $\textsc{Initialise}(X, Q, \varepsilon)$.
\item It adds $\mathbf{q}$ to the sets $Q_{\mu_i}$ for every $\mu_i\in M$ that satisfies $\hat{\mu}_\mathbf{q} \leq \mu_i$, and adds $\mathbf{q}$   to the buckets $B^*_{H_{\mu_i,a,j,\ell}}(\mathbf{q})$ for $a \in [K_1]$, $j \in [J_{\mu_i}]$ and $\ell \in [K_{2}]$, which we call the \emph{query-hash}. 
\end{enumerate}
When a new data point $\mathbf{z}$ arrives, we invoke the $\textsc{AddDataPoint}(\mathbf{z})$ procedure. This procedure is our main technical contribution to enable dynamic updates of query estimates within the framework of~\citet{charikarKDE}, and works as follows: $\textsc{AddDataPoint}(\mathbf{z})$ checks whether the number of data points has doubled since the last construction~(or reconstruction) of the data structure, and re-initialises the data structure if it is the case. Otherwise, $\textsc{AddDataPoint}(\mathbf{z})$ performs the following operations:
\begin{enumerate}
\item It samples $\mathbf{z}$ with probability $\min\left\{1 / (2^{j+1}\mu_i),1\right\}$ for all possible $\mu_i\in M$, and adds the sampled $\mathbf{z}$ to the buckets  $B_{H_{\mu_i, a, j, \ell}}(\mathbf{z})$ for all $a\in K_1$, $j \in [J_{\mu_i}]$, and $\ell \in [K_{2}]$; it also
samples $\mathbf{z}$
with probability $1/(2n_1)$ for all the possible values of $\mu_i$ and $a$, and adds the sampled point to the set $\{\wt{X}_{\mu_i, a}\}_{\mu_i, a}$.
 Notice that the way that $\mathbf{z}$ is added in the buckets is exactly the same as the one when executing $\textsc{Initialise}(X \cup \mathbf{z}, Q, \varepsilon)$, and hence $\textsc{AddDataPoint}(\mathbf{z})$ correctly updates all the buckets. This bucket-updating procedure is similar to \citet{liang2022dynamic}, though in their dynamic setting a newly arriving point replaces an existing one; here, we add the new point instead.

 \item If $\mathbf{z}$ is sampled, $\textsc{AddDataPoint}(\mathbf{z})$ recovers all the  points $\mathbf{q} \in B^*_{H_{\mu_i, a, j, \ell}}(\mathbf{z})$ in the query hash that satisfies $\mathbf{q} \in Q_{\mu_i} \setminus \left(\bigcup_{j' < i} Q_{\mu_{j'}}\right)$ and $\mathbf{z} \in L_j^{\mathbf{q}}$. Every such  $\mathbf{q}$  is  exactly the point whose $\textsf{KDE}$ estimate $\hat{\mu}_{\mathbf{q}}$ would have included $\mathbf{z}$ if $\textsc{AddQueryPoint}(\mathbf{q})$ would have been called after running $\textsc{Initialise}(X \cup \mathbf{z}, Q, \varepsilon)$. Hence, the \textsf{KDE} estimates $\hat{\mu}_{\mathbf{q}}$ for the recovered $\mathbf{q}$ are updated appropriately.
\end{enumerate}
See Figure~\ref{fig:update_kde} for illustration. The correctness and running time analysis of $\textsc{AddDataPoint}(\mathbf{z})$ and 
 $\textsc{AddQueryPoint}(\mathbf{q})$ can be found in Section~\ref{sec:dynamic_kde_appendix} of the appendix. 
 
 Finally, the $\textsc{DeleteQueryPoint}(\mathbf{q})$ procedure simply removes any stored information about the query point $\mathbf{q}$ throughout all the maintained data, and its running time follows from the running time of $\textsc{AddQueryPoint}(\mathbf{q})$.

\section{Dynamic Similarity Graph Construction}\label{sec:dynamic_cps}

 In this section we design an approximate dynamic algorithm that constructs a similarity graph under a sequence of data point insertions, and analyse its performance. Given a set $X$ of $n_1$ points in $\mathbb{R}^d$ with  $d=O(\log n_1)$,  and a sequence of points $\{\mathbf{z}\}$ that will be added to $X$ over time, our designed algorithm consists of the
 \textsc{ConstructGraph} and \textsc{UpdateGraph} procedures, whose performance is   as follows:

\begin{theorem}[\textbf{Main Result~2}]\label{thm:dynamic_cps}
Let $k$ be a kernel function with $\mathrm{cost}(k)$ as defined in Definition~\ref{def:cost_kernel}, and  $X \subset \mathbb{R}^d$ a set of data points  updated through point insertions. Then, the following statements hold:

\begin{enumerate}
\item  \underline{The Initialisation Step}:  with probability at least $9/10$,  the $\textsc{ConstructGraph}$ procedure constructs an approximate similarity graph $G = (X, E, w_G)$ with $|E| = \widetilde{O}(n_1)$ edges, where $n_1 = |X|$ is the number of data points at initialisation. The running time of the initialisation step is $n_1^{1+o(1)} \cdot \mathrm{cost}(k)$. 
\item \underline{The Dynamic Update Step:} for every new arriving data point $ \mathbf{z}$, the  $\textsc{UpdateGraph}$ procedure updates the approximate similarity graph $G$, and with probability at least $9/10$ $G$ is an approximate similarity graph for $X \cup \mathbf{z}$. The expected amortised update time is $ n^{o(1)} \cdot \mathrm{cost}(k) $, where $n$ is the number of currently considered data points. 
\end{enumerate}
\end{theorem}

 On the significance of Theorem~\ref{thm:dynamic_cps},   first notice that the algorithm achieves an update time of $n^{o(1)}\cdot\mathrm{cost}(k) $. For the Gaussian kernel, this corresponds to an update time of $n^{(1+o(1))\frac{1}{4}}$, which is much faster than the $\wt{O}(n)$ time needed to update the fully connected similarity graph. Secondly, our result demonstrates that, as long as the dynamically changing set $X\subset \mathbb{R}^d$ of points presents a clear cluster structure, an approximate similarity graph  $G$ for $X$ can be dynamically maintained, and the conductance of every cluster in $G$ can be theoretically analysed, due to the formal definition of an approximate similarity graph~(Definition~\ref{def:asg}). This is another difference between our work and many heuristic clustering algorithms that lack a theoretical guarantee on their performance.

\subsection{Analysis for the Initialisation\label{sec:sg_init}}

 The main component of the initialisation step is our designed      $\textsc{ConstructGraph}(X)$ procedure, which   builds a (complete) binary tree $\tree$ for $X$, such that each leaf corresponds to a data point $\mathbf{x}_i \in X$, and  each internal node of $\tree$ corresponds to  a dynamic \textsf{KDE} data structure  (described in Section~\ref{sec:dynamic_kde}) on the descendant leaves/data points.
At a high level,   $\textsc{ConstructGraph}(X)$ applies  the  $\textsc{Sample}(X, \tree, \ell)$ procedure to 
     recursively traverse $\tree$ and sample $L=O(\log|X|)$ neighbours for every vertex   based on the \textsf{KDE}s maintained by the internal nodes. It also stores the  paths $\mathcal{P}_{\mathbf{x}_i, \ell}~(\mathbf{x}_i\in X, \ell\leq L)$, each of which   records the internal nodes that $\mathbf{x}_i$ visits when sampling its $\ell$th neighbour; with these stored paths the algorithm  can adaptively resample the tree as new data points arrive. 
In addition, the query points whose \textsf{KDE} estimates are dynamically maintained at any internal node $\tree'$ are the data points $\mathbf{x}_i$ whose sample path $\mathcal{P}_{\mathbf{x}_i, \ell}$ visit $\tree'$.

Our initialisation procedures and corresponding proofs follow   \citet{macgregor2024fastkde} at a high level, however there are several  crucial differences between the two algorithms. First of all, our algorithm explicitly tracks the sample paths $\mathcal{P}_{\mathbf{x}_i, \ell}$ to ensure we can adaptively resample edges of the similarity graph. Secondly, a subtle but key difference between our analysis and theirs is that the weight of every edge $(\mathbf{x}_i, \mathbf{x}_j)$ added by our algorithm is set to be $k(\mathbf{x}_i, \mathbf{x}_j) / \hat{w}(i,j)$. Here, $\hat{w}(i,j)$ depends on   $ \min \{ \tree.\textsf{kde}.\hat{\mu}_{\mathbf{x}_i}, \tree.\textsf{kde}.\hat{\mu}_{\mathbf{x}_j} \}$, where $\tree.\textsf{kde}$ is the dynamic \textsf{KDE} data structure maintained at the root of $\tree$. In particular, every sampled edge $(\mathbf{x}_i, \mathbf{x}_j)$  is added with this weight independent of the edge being sampled from  $\mathbf{x}_i$ or $\mathbf{x}_j$. This modification  allows for correct reweighting and resampling after a sequence of data point insertions.

\subsection{Analysis for Dynamic Updates\label{sec:sg_update}}

 The main technical contribution of our dynamic approximate similarity graph result is the $\textsc{UpdateGraph}(G, \mathcal{T}, \mathbf{z})$  procedure, which dynamically updates our  constructed graph such that the updated graph is an approximate similarity graph for $X\cup\mathbf{z}$. At a high level, $\textsc{UpdateGraph}(G, \mathcal{T}, \mathbf{z})$  works as follows:
\begin{enumerate}
    \item   for every new data point $\textbf{z}$,    $\textsc{UpdateGraph}(G, \mathcal{T}, \mathbf{z})$  
 creates a new leaf node for $\mathbf{z}$, and places it appropriately in   $\tree$ ensuring   that the new tree 
is  a complete binary tree; 
\item $\textsc{UpdateGraph}(G, \mathcal{T}, \mathbf{z})$  inspects the internal nodes from the new leaf $\mathbf{z}$ to the root of the tree, and for every such internal node it adds   $\mathbf{z}$ as a new data point in the corresponding dynamic \textsf{KDE} estimators;
\item 
$\textsc{UpdateGraph}(G, \mathcal{T}, \mathbf{z})$    further checks in every internal node along the sample path $\mathcal{P}_{\mathbf{x}_i, \ell}$ whether the $\textsf{KDE}$ estimate of any $\mathbf{x}_i$ has changed due to the insertion of $\mathbf{z}$. If so, $\mathcal{P}_{\mathbf{x}_i, \ell}$ is added to the set $\mathcal{A}$ of   paths that need to be updated. 
 For every $\mathcal{P}_{\mathbf{x}_i, \ell} \in \mathcal{A}$, $\textsc{UpdateGraph}(G, \mathcal{T}, \mathbf{z})$  finds the highest internal node $\tree'$ where the \textsf{KDE} estimate of $\mathbf{x}_i$ has changed,     removes the path from all nodes below $\tree'$, and resamples the corresponding edges; this is achieved through    $\textsc{Resample}(\tree, \mathcal{P}_{\mathbf{x}_i, \ell})$. Additionally, it employs   $\textsc{Sample}(\{\mathbf{z}\},\tree,\ell )$   to  sample $L$ new edges adjacent to  $\mathbf{z}$.
 \end{enumerate}
 See Figure~\ref{fig:resample_path} for illustration. It is easy to see that the total time complexity of $\textsc{UpdateGraph}(G, \mathcal{T}, \mathbf{z})$  is dominated by (i)  the time complexity of $\textsc{Sample}( \{\mathbf{z}\}, \tree, \ell )$   and $\textsc{Resample}(\tree, \mathcal{P}_{\mathbf{x}, \ell})$, and (ii) the total number of paths $\mathcal{A}$ that need to be reconstructed. We study the time complexity of these two procedures, and our result is as follows:

\begin{figure*}[t]
\vskip 0.2in
    \begin{center}
    \centerline{\resizebox{0.9\textwidth}{!}{%
    \begin{tikzpicture}
	\begin{pgfonlayer}{nodelayer}
		\node [style=minsize] (0) at (-6, 8) {};
		\node [style=minsize] (2) at (-8, 4) {};
		\node [style=minsize] (3) at (-4, 4) {};
		\node [style=minsize] (4) at (0, 4) {};
		\node [style=leaf node] (5) at (4, 4) {\LARGE{$\mathbf{x}_7$}};
		\node [style=leaf node] (6) at (8, 4) {\LARGE{$\mathbf{x}_8$}};
		\node [style=leaf node] (7) at (12, 4) {\LARGE{$\mathbf{x}_9$}};
		\node [style=leaf node] (8) at (16, 4) {\LARGE{$\mathbf{x}_{10}$}};
		\node [style=minsize] (11) at (2, 8) {};
		\node [style=minsize] (13) at (10, 8) {};
		\node [style=minsize] (15) at (18, 8) {};
		\node [style=minsize] (16) at (-2, 12) {};
		\node [style=minsize] (17) at (14, 12) {};
		\node [style=minsize] (18) at (6, 16) {};
		\node [style=none] (19) at (6, 18.5) {\Huge{$\mathcal{T}$}};
		\node [style=none] (20) at (-7.25, 17.5) {};
		\node [style=none] (21) at (-3.25, 17.5) {};
		\node [style=none] (22) at (-5.25, 18.75) {\LARGE{$\mathcal{P}_{\mathbf{x}_1, \ell}$}};
		\node [style=none] (24) at (-7.25, 14.5) {};
		\node [style=none] (25) at (-3.25, 14.5) {};
		\node [style=none] (26) at (-5.25, 15.75) {\LARGE{$\mathcal{P}_{\mathbf{x}_1, \ell'}$}};
		\node [style=none] (42) at (39, 18.5) {\Huge{$\mathcal{T}'$}};
		\node [style=none] (52) at (25.75, 14.5) {};
		\node [style=none] (53) at (29.75, 14.5) {};
		\node [style=none] (54) at (27.75, 15.75) {\LARGE{Resampled $\mathcal{P}_{\mathbf{x}_1, \ell'}$}};
		\node [style=none] (55) at (25.75, 17.5) {};
		\node [style=none] (56) at (29.75, 17.5) {};
		\node [style=none] (57) at (27.75, 18.75) {\LARGE{Unchanged $\mathcal{P}_{\mathbf{x}_1, \ell}$}};
		\node [style=none] (63) at (38.25, 11) {{\color[rgb]{0.87,0.18,0.15}\LARGE{$\hat{\mu}_{\mathbf{x}_i}$ change}}};
		\node [style=leaf node] (64) at (-9, 0) {\LARGE{$\mathbf{x}_1$}};
		\node [style=leaf node] (65) at (-7, 0) {\LARGE{$\mathbf{x}_2$}};
		\node [style=leaf node] (66) at (-5, 0) {\LARGE{$\mathbf{x}_3$}};
		\node [style=leaf node] (67) at (-3, 0) {\LARGE{$\mathbf{x}_4$}};
		\node [style=leaf node] (68) at (-1, 0) {\LARGE{$\mathbf{x}_5$}};
		\node [style=leaf node] (69) at (1, 0) {\LARGE{$\mathbf{x}_6$}};
		\node [style=leaf node] (70) at (20, 4) {\LARGE{$\mathbf{x}_{11}$}};
		\node [style=minsize] (71) at (27, 8) {};
		\node [style=minsize] (72) at (25, 4) {};
		\node [style=minsize] (73) at (29, 4) {};
		\node [style=minsize] (74) at (33, 4) {};
		\node [style=minsize] (75) at (37, 4) {};
		\node [style=leaf node] (76) at (41, 4) {\LARGE{$\mathbf{x}_8$}};
		\node [style=leaf node] (77) at (45, 4) {\LARGE{$\mathbf{x}_9$}};
		\node [style=leaf node] (78) at (49, 4) {\LARGE{$\mathbf{x}_{10}$}};
		\node [style={red_mu_change}] (79) at (35, 8) {};
		\node [style=minsize] (80) at (43, 8) {};
		\node [style=minsize] (81) at (51, 8) {};
		\node [style=minsize] (82) at (31, 12) {};
		\node [style=minsize] (83) at (47, 12) {};
		\node [style=minsize] (84) at (39, 16) {};
		\node [style=leaf node] (85) at (24, 0) {\LARGE{$\mathbf{x}_1$}};
		\node [style=leaf node] (86) at (26, 0) {\LARGE{$\mathbf{x}_2$}};
		\node [style=leaf node] (87) at (28, 0) {\LARGE{$\mathbf{x}_3$}};
		\node [style=leaf node] (88) at (30, 0) {\LARGE{$\mathbf{x}_4$}};
		\node [style=leaf node] (89) at (32, 0) {\LARGE{$\mathbf{x}_5$}};
		\node [style=leaf node] (90) at (34, 0) {\LARGE{$\mathbf{x}_6$}};
		\node [style=leaf node] (91) at (53, 4) {\LARGE{$\mathbf{x}_{11}$}};
		\node [style=leaf node] (92) at (36, 0) {\LARGE{$\mathbf{x}_7$}};
		\node [style=leaf node] (93) at (38, 0) {{\color[rgb]{0.87,0.18,0.15}\LARGE{$\mathbf{z}$}}};
		\node [style=none] (95) at (39.5, 1) {};
		\node [style=none] (96) at (38.5, 4) {};
		\node [style=none] (97) at (38, 5.25) {};
		\node [style=none] (98) at (36.75, 7.5) {};
		\node [style=none] (99) at (35, 9.75) {};
		\node [style=none] (100) at (32.75, 12) {};
		\node [style=none] (101) at (34.25, 12.25) {};
		\node [style=none] (102) at (38.5, 14.5) {};
		\node [style=big vertical circle dashed] (103) at (35, 8) {};
		\node [style=none] (104) at (6, -3.5) {\LARGE{(a)}};
		\node [style=none] (105) at (39, -3.5) {};
		\node [style=none] (106) at (39, -3.5) {\LARGE{(b)}};
	\end{pgfonlayer}
	\begin{pgfonlayer}{edgelayer}
		\draw [style=undirected] (2) to (0);
		\draw [style=undirected] (0) to (3);
		\draw [style=undirected] (4) to (11);
		\draw [style=undirected] (11) to (5);
		\draw [style=undirected] (13) to (7);
		\draw [style=undirected] (15) to (8);
		\draw [style=undirected] (0) to (16);
		\draw [style=undirected] (16) to (11);
		\draw [style=undirected] (13) to (17);
		\draw [style=undirected] (17) to (15);
		\draw [style=undirected] (16) to (18);
		\draw [style=undirected] (18) to (17);
		\draw [style=directed blue] (18) to (17);
		\draw (17) to (13);
		\draw [style=directed blue] (17) to (13);
		\draw [style=directed blue] (13) to (7);
		\draw [style=directed blue] (20.center) to (21.center);
		\draw [style=directed green] (18) to (16);
		\draw [style=directed green] (24.center) to (25.center);
		\draw [style=green dashed] (52.center) to (53.center);
		\draw [style=directed blue] (55.center) to (56.center);
		\draw [style=undirected] (15) to (70);
		\draw [style=undirected] (64) to (2);
		\draw [style=undirected] (2) to (65);
		\draw [style=undirected] (66) to (3);
		\draw [style=undirected] (67) to (3);
		\draw [style=undirected] (68) to (4);
		\draw [style=undirected] (4) to (69);
		\draw [style=undirected] (72) to (71);
		\draw [style=undirected] (79) to (75);
		\draw [style=undirected] (80) to (77);
		\draw [style=undirected] (81) to (78);
		\draw [style=undirected] (80) to (83);
		\draw [style=undirected] (83) to (81);
		\draw [style=undirected] (82) to (84);
		\draw [style=undirected] (84) to (83);
		\draw [style=directed blue] (84) to (83);
		\draw (83) to (80);
		\draw [style=directed blue] (83) to (80);
		\draw [style=directed blue] (80) to (77);
		\draw [style=directed green] (84) to (82);
		\draw [style=undirected] (81) to (91);
		\draw [style=undirected] (85) to (72);
		\draw [style=undirected] (72) to (86);
		\draw [style=undirected] (88) to (73);
		\draw [style=undirected] (89) to (74);
		\draw [style=undirected] (75) to (92);
		\draw [style=undirected] (75) to (93);
		\draw [style=red dashed] (95.center) to (96.center);
		\draw [style=red dashed] (97.center) to (98.center);
		\draw [style=red dashed] (99.center) to (100.center);
		\draw [style=red dashed] (101.center) to (102.center);
		\draw [style=directed green] (16) to (0);
		\draw [style=directed green] (0) to (3);
		\draw [style=directed green] (3) to (66);
		\draw [style=undirected] (6) to (13);
		\draw [style=undirected] (76) to (80);
		\draw [style=undirected] (87) to (73);
		\draw [style=undirected] (73) to (71);
		\draw [style=undirected] (71) to (82);
		\draw [style=green dashed] (82) to (103);
		\draw [style=green dashed] (103) to (74);
		\draw [style=green dashed] (74) to (90);
	\end{pgfonlayer}
\end{tikzpicture}
    }}
\caption{Illustration of updating $\tree$ after performing $\textsc{UpdateGraph}(\mathbf{z})$. In (a),     $\mathcal{P}_{\mathbf{x}_1, \ell}$ and $\mathcal{P}_{\mathbf{x}_1, \ell'}$ are  generated by $\textsc{Sample}(\{\mathbf{x}_1\}, \tree, \ell)$ and $\textsc{Sample}(\{\mathbf{x}_1\}, \tree, \ell')$,  and  correspond  to edges $(\mathbf{x}_1, \mathbf{x}_9)$ and $(\mathbf{x}_1, \mathbf{x}_3)$. (b) illustrates that, after adding $\mathbf{z}$, part of $\mathcal{P}_{\mathbf{x}_1, \ell'}$ is updated due to  $\textsc{Resample}(\tree', \mathcal{P}_{\mathbf{x}_1, \ell'})$, and $(\mathbf{x}_1, \mathbf{x}_3)$ is replaced by $(\mathbf{x}_1, \mathbf{x}_6)$; however, the update on $\mathbf{z}$ doesn't change $\mathcal{P}_{\mathbf{x}_1,\ell}$.}
 \label{fig:resample_path}
 \end{center}
\vskip -0.2in
\end{figure*}
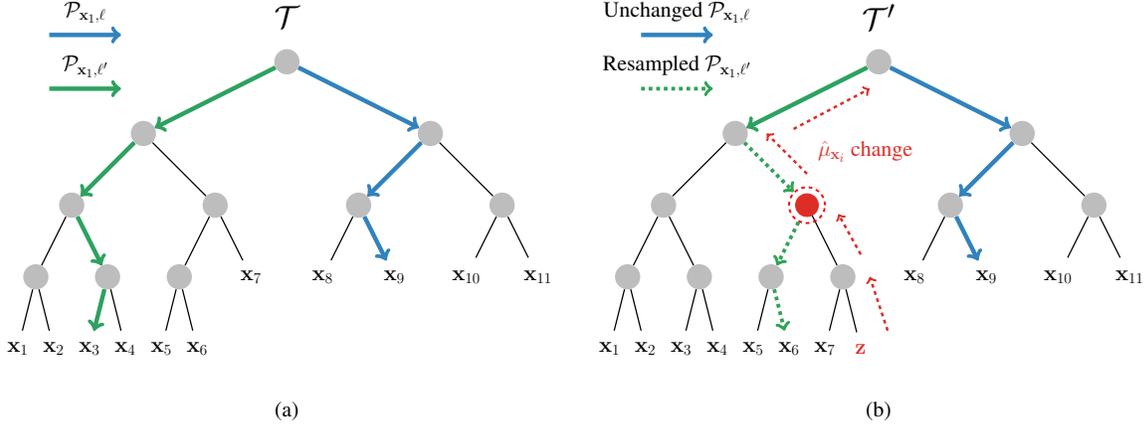


\begin{lemma}\label{lem:running_time_sample_procedures}
    For any   $\mathbf{x} \in \mathbb{R}^d$ and   $\ell \in \mathbb{N}$, the running time of    $\textsc{Sample}( \{\mathbf{x}\}, \tree, \ell )$ (Algorithm~\ref{alg:initialise_CPS}) and $\textsc{Resample}(\tree, \mathcal{P}_{\mathbf{x}, \ell})$ (Algorithm~\ref{alg:update_CPS_tree}) is $   n^{o(1)}\cdot \mathrm{
    cost}(k)$. 
    \end{lemma}

Given that \textsc{Sample}$({\mathbf{z}}, \mathcal{T}, \ell)$ is called $L = \widetilde{O}(1)$ times by \textsc{UpdateGraph}$(G, \mathcal{T}, \mathbf{z})$, and \textsc{Resample}$(\mathcal{T}, \mathcal{P}_{\mathbf{x}, \ell})$ is called $|\mathcal{A}|$ times, Lemma~\ref{lem:running_time_sample_procedures}  implies that the running time of \textsc{UpdateGraph} depends on the number of re-sampled paths $|\mathcal{A}|$. Therefore, to prove the time complexity of \textsc{UpdateGraph}, it remains to show that $\mathbb{E}[|\mathcal{A}|]$ is sufficiently small.

Bounding the expected number of re-sampled paths corresponds to bounding the number of query points whose \textsf{KDE} estimates are updated at each affected internal node $\mathcal{T}'$. However, applying Theorem~\ref{thm:incremental_dynamic_kde} directly is \emph{insufficient} because, in our approximate similarity graph, the dynamic \textsf{KDE} data structures start with $Q = X$ rather than $Q = \emptyset$. As such, more careful analysis is needed, and the following notation for the query points $Q$ will be used.

 \begin{definition}
    Let $\tree$ be the \textsf{KDE} tree constructed from $\textsc{ConstructGraph}(X)$, and   $\tree'$   an internal node of $\tree$. Then, for $0 \leq j \leq i \leq \lceil\log(2\cdot \tree.\textsf{kde}.n) \rceil$, we define the set
    \begin{multline}
    Q_{\mu_i \rightarrow \mu_j}(\tree') \triangleq \{ \mathbf{q} \in \tree.\textsf{kde}.Q \mid  \mu_i \leq k(\mathbf{q}, \tree.\textsf{kde}.X) \\
     < 2\mu_i \text{ and } k(\mathbf{q}, \tree'.\textsf{kde}.X) \leq \mu_j\},
    \end{multline}
    where $\tree'.\textsf{kde}$ is the dynamic \textsf{KDE} data structure maintained at $\tree'$.
\end{definition}

The set $Q_{\mu_i \rightarrow \mu_j}(\tree')$ represents the set of query points $\mathbf{q} \in X$ whose \textsf{KDE} estimates are bounded by $\mu_i$ when computed with respect to the data points $X$ at the root of the tree, and bounded by $\mu_j$ for  $j \leq i$ when computed with respect to the data points $X'$ represented at the internal node $\tree'$. Intuitively, this set captures the query points whose \textsf{KDE} estimates decrease when moving from the root of the tree to the internal node $\tree'$. These sets exhibit the following useful property.

\begin{lemma}\label{lem:prob_go_to_subtree}
 It holds for any $\mathbf{q} \in Q_{\mu_i \rightarrow \mu_j}(\tree')$  that
    \[
    \mathbb{P}[\mathbf{q} \in \tree'.\textsf{kde}.Q] = \wt{O}\lp\frac{\mu_j}{\mu_i}\rp. \]
\end{lemma}

To bound the number of maintained query points whose estimate is updated, we look at the expected number of collisions caused by hashing $\mathbf{z}$ in the dynamic KDE data structure $\tree'.\textsf{kde}$ at every affected internal node $\tree'$. Crucially, by separately analysing the contributions from query points in $Q_{\mu_{i'} \rightarrow \mu_i}(\mathcal{T}')$ for $i' \geq i$ and applying Lemma~\ref{lem:prob_go_to_subtree}, we are able to  bound the expected number of colliding points in the buckets $\mathcal{T}'.\textsf{kde}.B_{H_{\mu_i, a, j, \ell}}(\mathbf{z})$ sufficiently tightly at each affected internal node $\mathcal{T}'$.

\begin{lemma}[Informal version of Lemma~\ref{lem:expected_collisions_subtree}]\label{lem:expected_collisions_subtree_informal}
     Let $\mathbf{z}$ be the data point that is added to   $\tree$ through our designed update procedures, and   $\tree'$ be any internal node that lies on the path from the new leaf $\textsc{Leaf}(\mathbf{z})$ to the root of $\tree$. Then it holds for any $i,a,j$, and $\ell$ that 
    \begin{multline}
        \mathbb{E}_{H_{\mu_i, a, j, \ell}}[|\{\mathbf{q} \in \tree'.\textsf{kde}.Q_{\mu_i} \mid \tree'.\textsf{kde}.H_{\mu_i, a, j, \ell}(\mathbf{z}) \\
         \qquad = \tree'.\textsf{kde}.H_{\mu_i, a, j, \ell}(\mathbf{q}) \} |] = \wt{O}\lp \mu_{i} \cdot 2^{j+1} \rp.
    \end{multline}
\end{lemma}

Combining Lemma~\ref{lem:expected_collisions_subtree_informal} with the fact that $\mathbf{z}$ is sampled with probability $\min\left\{1 / (2^{j+1}\mu_i),1\right\}$ for all possible $i \in \left[\lceil\log(2\cdot \tree.\textsf{kde}.n) \rceil\right]$ and $j \in [J_{\mu_i}]$ along every affected internal node $\tree'$, and noting that there are $\wt{O}(1)$ such nodes, we obtain the following result.

\begin{lemma}\label{lem:bound_number_A}
 For every added $\mathbf{z}$, the  expected number of paths $\mathcal{A}$ that need to be resampled by $\textsc{UpdateGraph}(G, \mathcal{T}, \mathbf{z})$  satisfies   $\mathbb{E}[|\mathcal{A}|]=\widetilde{O}(1)$.  
 \end{lemma}

Combining Lemmas~\ref{lem:running_time_sample_procedures} and~\ref{lem:bound_number_A} with the running time analysis of other involved procedures proves the time complexity in the second part of Theorem~\ref{thm:dynamic_cps}. To show that our dynamically maintained $G$ is an approximate similarity graph, we prove in Lemma~\ref{lemma:correctnes_update_CPS} that running        $\textsc{ConstructGraph}(X)$     followed by $\textsc{UpdateGraph}(G, \mathcal{T}, \mathbf{z})$      is equivalent to running $\textsc{ConstructGraph}(X\cup\mathbf{z})$; hence the correctness of our constructed $G$ follows from the one for $\textsc{ConstructGraph}(X)$.

\section{Experiments}\label{sec:experiments}
\newcommand{\dynamicexact}{\textsc{Exact}}
\newcommand{\dynamicrs}{\textsc{DynamicRS}}
\newcommand{\naiveCKNS}{\textsc{CKNS}}
\newcommand{\ourDynamicKDE}{\textsc{Our Algorithm}}
\newcommand{\fc}{\textsc{FullyConnected}}
\newcommand{\knn}{\textsc{kNN}}
\newcommand{\ours}{\textsc{Our Algorithm}}

\newcommand{\kdecaption}{The experimental results for the dynamic \textsf{KDE} algorithms. Shaded cells correspond to the algorithm with the lowest running time.
The running times of the exact algorithm are 164.5, 2715.6, 2179.9, 5251.7, and 16279.9. }
\begin{table*}[t]
  \caption{\kdecaption \newline \label{tab:dynamic_kde}}
  \vspace{-0.3cm}
  \centering
  \resizebox{0.8\textwidth}{!}{
  \begin{tabular}{lcccccc}
    \toprule
 & \multicolumn{2}{c}{\textsc{CKNS}} & \multicolumn{2}{c}{\textsc{DynamicRS}} & \multicolumn{2}{c}{\textsc{Our Algorithm}}\\ 
    \cmidrule(lr){2-3} \cmidrule(lr){4-5} \cmidrule(lr){6-7}
    dataset & Time (s) & Err & Time (s) & Err & Time (s) & Err\\ 
    \midrule
    aloi & $ 619.0{\scriptstyle \pm  10.7}$ & $ 0.050{\scriptstyle \pm  0.006}$ & \cellcolor{gray!25}$ 19.7{\scriptstyle \pm  0.3}$ & \cellcolor{gray!25}$ 0.010{\scriptstyle \pm  0.003}$ & $ 46.9{\scriptstyle \pm  0.7}$ & $ 0.060{\scriptstyle \pm  0.021}$ \\
    msd & $ 14,360.0{\scriptstyle \pm  0.0}$ & $ 0.385{\scriptstyle \pm  0.000}$ & $ 1,887.8{\scriptstyle \pm  0.0}$ & $ 5.430{\scriptstyle \pm  0.000}$ & \cellcolor{gray!25}$ 306.4{\scriptstyle \pm  0.0}$ & \cellcolor{gray!25}$ 0.388{\scriptstyle \pm  0.000}$ \\
    covtype & $ 5,650.3{\scriptstyle \pm  109.0}$ & $ 0.159{\scriptstyle \pm  0.002}$ & $ 309.2{\scriptstyle \pm  2.4}$ & $ 0.018{\scriptstyle \pm  0.003}$ & \cellcolor{gray!25}$ 151.7{\scriptstyle \pm  4.5}$ & \cellcolor{gray!25}$ 0.196{\scriptstyle \pm  0.017}$ \\
    glove & $ 2,640.8{\scriptstyle \pm  1677.7}$ & $ 0.221{\scriptstyle \pm  0.229}$ & $ 1,038.6{\scriptstyle \pm  26.5}$ & $ 0.004{\scriptstyle \pm  0.005}$ & \cellcolor{gray!25}$ 445.6{\scriptstyle \pm  214.6}$ & \cellcolor{gray!25}$ 0.296{\scriptstyle \pm  0.469}$ \\
    census & $ 10,471.5{\scriptstyle \pm  160.6}$ & $ 0.080{\scriptstyle \pm  0.000}$ & $ 3,424.8{\scriptstyle \pm  5.2}$ & $ 0.005{\scriptstyle \pm  0.001}$ & \cellcolor{gray!25}$ 836.5{\scriptstyle \pm  44.6}$ & \cellcolor{gray!25}$ 0.102{\scriptstyle \pm  0.021}$ \\
    \bottomrule
  \end{tabular}}
\end{table*}

\newcommand{\simcaption}{Running time and NMI results for the dynamic similarity graph algorithms. For each dataset, the shaded cells correspond to the algorithm with the lowest running time.}
\begin{table*}[h]
  \caption{\simcaption \newline \label{tab:dynamic_sg}}
  \vspace{-0.3cm}
  \centering
  \resizebox{0.8\textwidth}{!}{
  \begin{tabular}{lcccccc}
    \toprule
    & \multicolumn{2}{c}{\textsc{FullyConnected}} & \multicolumn{2}{c}{\textsc{kNN}} & \multicolumn{2}{c}{\textsc{Our Algorithm}}\\ 
    \cmidrule(lr){2-3} \cmidrule(lr){4-5} \cmidrule(lr){6-7}
    dataset & Time (s) & NMI & Time (s) & NMI & Time (s) & NMI\\ 
    \midrule
    blobs & $ 72.8{\scriptstyle \pm  2.2}$ & $ 1.000{\scriptstyle \pm  0.000}$ & $ 383.6{\scriptstyle \pm  3.9}$ & $ 0.933{\scriptstyle \pm  0.095}$ & \cellcolor{gray!25}$ 21.2{\scriptstyle \pm  0.8}$ & \cellcolor{gray!25}$ 1.000{\scriptstyle \pm  0.000}$ \\
    cifar10 & $ 19,158.2{\scriptstyle \pm  231.6}$ & $ 0.001{\scriptstyle \pm  0.000}$ & $ 3,503.0{\scriptstyle \pm  490.6}$ & $ 0.227{\scriptstyle \pm  0.002}$ & \cellcolor{gray!25}$ 1,403.5{\scriptstyle \pm  152.4}$ & \cellcolor{gray!25}$ 0.339{\scriptstyle \pm  0.021}$ \\
    mnist & \cellcolor{gray!25}$ 1,328.3{\scriptstyle \pm  159.5}$ & \cellcolor{gray!25}$ 0.460{\scriptstyle \pm  0.000}$ & $ 5,796.3{\scriptstyle \pm  234.3}$ & $ 0.812{\scriptstyle \pm  0.003}$ & $ 1,470.3{\scriptstyle \pm  77.9}$ & $ 0.523{\scriptstyle \pm  0.011}$ \\
    \bottomrule
  \end{tabular}}
\end{table*}

In this section, we experimentally evaluate our proposed dynamic algorithms for \textsf{KDE} and approximate similarity graph construction on the Gaussian kernel.
All experiments are performed on a computer server with 64 AMD EPYC 7302 16-Core Processors and 500 Gb of RAM.
We report the 2-sigma errors for all numerical results based on $3$ repetitions of each experiment, and Section~\ref{app:experiments} gives additional experimental details and results.
Our code can be downloaded from \href{https://github.com/SteinarLaenen/Dynamic-Similarity-Graph-Construction-with-Kernel-Density-Estimation}{https://github.com/SteinarLaenen/Dynamic-Similarity-Graph-Construction-with-Kernel-Density-Estimation}.

 We evaluate the algorithms on a variety of real-world and synthetic data, and we summarise their properties in Table~\ref{tab:full-dataset-info} in Section~\ref{app:experiments}.
The datasets cover a variety of domains, including synthetic data (blobs~\citep{scikit-learn}), images (mnist~\citep{lecun_mnist_1998}, aloi~\citep{aloi_dataset}), image embeddings (cifar10~\citep{resnet, cifar}), word embeddings (glove~\citep{glove_dataset}),  mixed numerical datasets (msd~\citep{msd-main-ref}, covtype~\citep{covtype_dataset}, and census~\citep{us_census_data}).

\subsection{Dynamic \textsf{KDE}}

To the best of our knowledge, our proposed algorithm is the first which solves the dynamic kernel density estimation problem.
For this reason, we compare our algorithm against the following baseline approaches:
\begin{enumerate}
    \item \dynamicexact: the exact kernel density estimate, computed incrementally as data points are added;
    \item \dynamicrs: a dynamic \textsf{KDE} estimator based on uniform random sampling of the data. For all experiments, we uniformly subsample the data with sampling probability $0.1$; 
    \item \naiveCKNS: we use the fast kernel density estimation algorithm proposed by \citet{charikarKDE}, and fully re-compute the estimates every time the data is updated.
\end{enumerate}
For each dataset, we set the parameter $\sigma$ of the Gaussian kernel such that the average kernel density $\mu_\mathbf{q} \cdot n^{-1} \approx 0.01$~\citep{karppa2022deann}.
We split the data points into chunks of size 1,000 for aloi, msd, and covtype, and size 10,000 for glove and census.
Then, we add one chunk at a time to the set of data points $X$ and the set of query points $Q$. 
At each iteration, we evaluate the kernel density estimates $\hat{\mu}_\mathbf{q}$ produced by each algorithm with the relative error~\citep{karppa2022deann}
 \[
    \mathrm{err} = \frac{1}{|Q|} \sum_{\mathbf{q} \in Q} \left| (\hat{\mu}_\mathbf{q} - \mu_\mathbf{q}) / \mu_\mathbf{q} \right|.
\]
Table~\ref{tab:dynamic_kde} gives the total running time and final relative error for each algorithm, and Figure~\ref{fig:kde_update_times} shows the time taken to update the data structure for the census dataset at each iteration.
From these results, we observe that our algorithm scales better to large datasets than the baseline algorithms, while maintaining low relative errors.
 Figure~\ref{fig:kde_update_times} further shows that the update time of our algorithm is sub-linear in the number of data points, as shown theoretically in  Theorem~\ref{thm:incremental_dynamic_kde}.
The update time of the other algorithms is linear in $n$ and their total running time is quadratic.

\begin{figure}[h]
    \centering
    \includegraphics[width=0.45\textwidth]{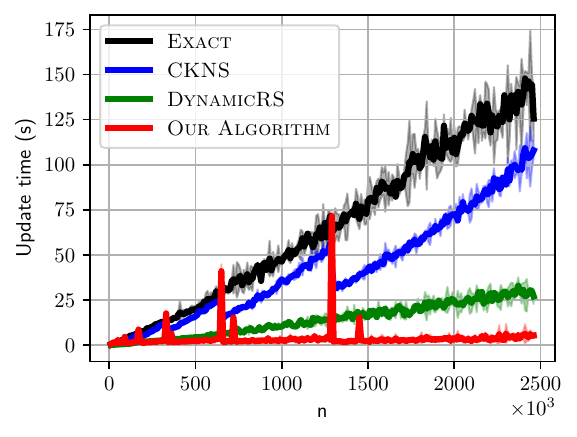}
    \vspace{-0.4cm}
    \caption{Dynamic \textsf{KDE} update time on the census dataset \label{fig:kde_update_times}}
\end{figure}

\subsection{Dynamic Clustering}
For the dynamic similarity graph algorithm, we compare against the two baseline algorithms:
\begin{enumerate}
    \item \fc: the fully-connected similarity graph with the Gaussian kernel; 
    \item \knn: the $k$-nearest neighbour graph, for $k = 20$.
We split the datasets into chunks of $1,000$ and add each chunk to the dynamically constructed similarity graph, adding one complete ground-truth cluster at a time.
\end{enumerate}
At each iteration, we apply the spectral clustering algorithm to the constructed similarity graph and report the normalised mutual information (NMI)~\citep{nmi} with respect to the ground truth clusters.
Table~\ref{tab:dynamic_sg} shows the total running time and final NMI values for each algorithm on each dataset.
From these results, we see that our algorithm achieves a competitive NMI value with faster running time than the baseline algorithms.

\section{Conclusion}
This paper develops dynamic algorithms for \KDE\ and approximate similarity graph constructions. Compared with many heuristic methods such as the \dynamicrs~algorithm and dynamic $k$-NN graphs, our algorithms have   theoretically proven approximation guarantees. Our empirical evaluation further demonstrates their  competitive performance, with our dynamic \KDE\ algorithm showing better scalability on large datasets and our similarity graph construction achieving faster running time while maintaining comparable clustering quality.

In addition to their theoretical advances, our algorithms could have several other applications. For example,  our dynamic \KDE\ tree is able to  maintain randomly sampled neighbours of a given vertex by edge weight of a  similarity graph, which is a common primitive in many algorithms for similarity graphs~\citep{bakshi2023subquadratic}, and our \KDE\ tree could facilitate the design of dynamic variants of these algorithms

\section*{Impact Statement}

This paper presents work whose goal is to advance the field of Machine Learning. There are many potential societal consequences of our work, none  of which we feel must be specifically highlighted here.

\section*{Acknowledgements}

The  third author of the paper is supported by an EPSRC Early Career Fellowship (EP/T00729X/1).

\bibliography{reference}
\bibliographystyle{icml2025}

\newpage
\appendix
\onecolumn

\section{Additional Background Knowledge} \label{sec:background}

 This section presents additional background knowledge used in our analysis, and is organised as follows: 
Section~\ref{sec:notation} lists further notation for graphs and useful facts in spectral graph theory; Section~\ref{sec:CPS_appendix} formally defines  the concept of an approximate similarity graph.

\subsection{Notation\label{sec:notation}}

Let  $G=(V,E,w)$ be an undirected graph of $n$ vertices and weight function $w: V\times V\rightarrow \mathbb{R}_{\geq 0}$. 
For any edge $e = (u, v) \in E$, we write $w_G(u,v)$ or $w_G(e)$ to express the   weight of $e$. For a vertex $u \in V$, we denote its  \emph{degree}  by $\deg_G(u) \triangleq \sum_{v \in V} w_G(u,v)$, and the volume for any $S \subseteq V$ is defined as $\vol_G(S)\triangleq \sum_{u\in S} \deg_{G}(u)$. For any two subsets $S, T \subset V$, we define the \emph{cut value} between $S$ and $T$ by  $w_G(S, T) \triangleq \sum_{e \in E_G(S, T)} w_e$, where $E_G(S, T)$ is the set of edges between $S$ and $T$. Moreover, for any $S\subset V$,  the conductance of $S$ is defined as 
\[
\Phi_G(S) \triangleq \frac{w_G(S, V\setminus S) }{\min\{\vol_G(S), \vol_{G}(V \setminus S)\}}
\]
if $S \neq \emptyset$, and $\Phi_G(S) = 1$ if $S = \emptyset$.  For any integer $k\geq 2$, we call 
 subsets of vertices $A_1,\ldots, A_k$ a $k$-way partition of $G$ if $\bigcup_{i=1}^k A_i = V$ and 
 $A_i\cap A_j=\emptyset$ for different $i$ and $j$. We   define the \emph{$k$-way expansion} of $G$ by
\begin{equation*}
    \rho_G(k) \triangleq \min_{\mathrm{partitions \:} A_1, \dots, A_k} \max_{1\leq i \leq k} \Phi_G(A_i).
\end{equation*}

Part of our analysis is based on algebraic properties of graphs, hence we define graph Laplacian matrices.  For a graph $G = (V, E, w)$,  let $D_G\in\mathbb{R}^{n\times n}$ be the diagonal matrix defined by $D_G(u,u) = \deg_G(u)$ for all $u \in V$. We  denote by  $A_G\in\mathbb{R}^{n\times n}$  the  \emph{adjacency matrix}  of $G$, where $A_G(u,v) = w_G(u,v)$ for all $u, v \in V$. The \emph{normalised Laplacian matrix} of $G$ is defined as \[
\mathcal{L}_G \triangleq I - D_G^{-1/2} A_G D_G^{-1/2},\] where $I$ is the $n \times n$ identity matrix.  The normalised Laplacian $\mathcal{L}_G$ is symmetric and real-valued, and   has $n$ real eigenvalues which we   write as  $\lambda_1(\mathcal{L}_G) \leq \ldots \leq \lambda_n(\mathcal{L}_G)$. We sometimes refer to the $i$th eigenvalue of $\mathcal{L}_G$ as $\lambda_i$ if it is clear from the context. It is known that  $\lambda_1(\mathcal{L}_G)=0$ and $\lambda_n(\mathcal{L}_G)\leq 2$~\citep{chung1997spectral}.  The following result will be used in our analysis.

\begin{lemma}[higher-order Cheeger inequality, \cite{higherCheeg}]\label{lem:Higher Cheeger}
 It holds for any graph $G$ and $k \geq 2$ that 
\begin{equation}
    \frac{\lambda_k(\mathcal{L}_G)}{2} \leq \rho_G(k) = O \lp k^3 \rp \sqrt{\lambda_k(\mathcal{L}_G)}. \label{eq:Higher Cheeger}
\end{equation}
\end{lemma}

\subsection{Approximate Similarity Graph}\label{sec:CPS_appendix}

We first introduce the notion of cluster-preserving sparsifiers.

\begin{definition}[Cluster-preserving sparsifier, \cite{SZ19}]\label{def:CPS}
     Let $F=(V, E, w)$ be any graph with $k$ clusters, and $\{S_i\}^{k}_{i=1}$ a $k$-way partition of $F$ corresponding to $\rho_F(k)$. We call a re-weighted subgraph $G=(V, E' \subset E, w_G)$ a cluster-preserving sparsifier of $F$ if (i) $\Phi_{G}(S_i) = O(k \cdot \Phi_{F}(S_i))$ for $1 \leq i \leq k$, and (ii) $\lambda_{k+1}(\mathcal{L}_G) = \Omega(\lambda_{k+1}(\mathcal{L}_F))$.
\end{definition}

Notice that 
 graph $F=(V,E, w)$ has exactly $k$ clusters if (i) $F$ has $k$ disjoint subsets $S_1,\ldots, S_k$ of low conductance, and (ii) any $(k + 1)$-way partition of $F$ would include some  
$A \subset V$ of high conductance,
which would be implied by a lower bound on $\lambda_{k+1}(\mathcal{L}_F)$ due to \eqref{eq:Higher Cheeger}.
Together with the well-known eigen-gap heuristic~\citep{sac/Luxburg07} and theoretical analysis on spectral clustering~\citep{peng_partitioning_2017}, these two conditions  ensure that the $k$ optimal clusters in $F$ have low conductance in $G$ as well. Based on this, we define approximate similarity graphs as follows:

\begin{definition}[Approximate Similarity Graph] \label{def:asg}
For any set $X\subset \mathbb{R}^d$ of $n$ data points and the fully connected similarity graph $F$ on $X$, we call a  sparse graph $G$ with $\widetilde{O}(n)$ edges an approximate similarity graph on $X$ if $G$ is a cluster-preserving sparsifier of $F$.  
\end{definition}

We call $G$ an approximate similarity graph in the extended abstract if $G$ satisfies the properties of Definition~\ref{def:asg}.

\section{Omitted Detail from Section~\ref{sec:dynamic_kde}}\label{sec:dynamic_kde_appendix}
This section provides the detailed explanations omitted from Section~\ref{sec:dynamic_kde}, and is organised as follows: Section~\ref{sec:dynamic_kde_initialise_query} analyses the initialisation and querying procedures. 
Section~\ref{sec:dynamic_kde_add_data} analyses the dynamic update step for adding data points. Finally, Section~\ref{sec:dynamic_kde_proof_theorem} proves Theorem~\ref{thm:incremental_dynamic_kde}.

Algorithm~\ref{alg:dynamic_kde} describes all the used procedures and corresponding subprocedures, whose performance is summarised in Theorem~\ref{thm:incremental_dynamic_kde}.

\begin{algorithm}[H]
  \caption{$\textsc{DynamicKDE}(X, Q, \varepsilon)$}  
  \label{alg:dynamic_kde} 
  \begin{algorithmic}[1]

    \STATE \textbf{Members} 
    \STATE \hspace{1em} $\hat{\mu}_\mathbf{q}$ \algorithmiccomment{Query estimates for every point $\mathbf{q} \in Q$}
    \STATE \hspace{1em} $\varepsilon$ \COMMENT{Precision parameter for \textsf{KDE} estimate}
    \STATE \hspace{1em} For $\mu_i \in M$, create set $Q_{\mu_i}$ \COMMENT{Set of  data points with  query estimate  less than $\mu_i$} 
    \PROCEDURE{Initialise($X$, $Q$, $\varepsilon$)}
    \STATE $n', n \leftarrow |X|$, $\bar{m} \gets \left\lceil\frac{C}{\epsilon^2}\right\rceil$, $\bar{N} \gets \left\lceil \log(2n')\right\rceil$   \COMMENT{$C$ is a universal constant} \alglinelabel{alg:dynamic_kde_constant}
    \STATE $K_1 \gets \bar{m} \cdot \bar{N}$ \COMMENT{Number of independent estimators used} \alglinelabel{algdef:K1}
    \STATE $J_{\mu_i} \gets \left\lceil\log {\frac{2 n'}{\mu_i}}\right\rceil$ for $\mu_i \in M$ \COMMENT{See~Definition~\ref{def:geometric_weight_levels}}
    \STATE $\textsc{PreProcess}(X, \varepsilon)$ \COMMENT{Initialise the~\cite{charikarKDE} data structure (Algorithm~\ref{alg:preprocess_dataset})} 
    \STATE $\textsc{PreProcessQueryPoints}(X, Q, \varepsilon)$ \COMMENT{ (Algorithm~\ref{alg:full_hash_procedures})}
    \ENDPROCEDURE 

    \PROCEDURE{AddQueryPoint($\mathbf{q}$)}{}
    \STATE $\hat{\mu}_{\mathbf{q}} \gets \textsc{QueryPoint}(X, \mathbf{q}, \varepsilon)$ \COMMENT{(Algorithm~\ref{alg:querypoints})}
    \STATE Store $\hat{\mu}_{\mathbf{q}}$
    \STATE  Add $\mathbf{q}$ to all sets $Q_{\mu_i}$ such that $\hat{\mu}_\mathbf{q} \leq \mu_i$
    \STATE $\textsc{AddFullHash}(\mathbf{q})$ \COMMENT{(Algorithm~\ref{alg:full_hash_procedures})}
    \ENDPROCEDURE

    \PROCEDURE{AddDataPoint($\mathbf{z}$)}{}
    \STATE $n \gets n+1$
    \IF{$n - n' > n'$} \alglinelabel{line:dynamic_kde:recompute_start}
    \STATE {\textsc{Initialise}$(X \cup \mathbf{z}, Q, \varepsilon)$ \COMMENT{We reconstruct the data structure}}
    \ELSE{
    \STATE $\textsc{AddPointAndUpdateQueries}(\mathbf{z}, Q, \varepsilon)$ \COMMENT{(Algorithm~\ref{alg:addpointandupdatequeries})}
    }
    \ENDIF
    \ENDPROCEDURE

    \PROCEDURE{DeleteQueryPoint($\mathbf{q}$)}{}
    \STATE $\textsc{DeleteFromData}(\mathbf{z})$ \COMMENT{(Algorithm~\ref{alg:delete_querypoint})}
    \STATE Delete $\hat{\mu}_{\mathbf{q}}$
    \ENDPROCEDURE

  \end{algorithmic}
\end{algorithm}

\begin{algorithm}[H]
	\caption{DynamicKDE Preprocessing}  
	\label{alg:preprocess_dataset} 
	\begin{algorithmic}[1]
		
		\PROCEDURE{$\textsc{PreProcess}(X,\epsilon)$}{}
	
		\STATE \textbf{Input:} the set $X$ of data points, and the precision estimate $\epsilon$
  	 
        \FOR{$\mu_i \in M$}  
		\FOR[$K_1 = O(\log n'/\epsilon^2)$ independent repetitions]{$a=1,2,\dots,K_1$}
		\FOR[$J_{\mu_i}=\left\lceil\log \frac{2 n'}{\mu_i}\right\rceil$ geometric weight levels]{$j=1,2,\dots,J_{\mu_i}$} 
		\STATE $K_{2,j}\gets 200\log n' \cdot p_{\mathrm{near},j}^{-k_j}$ 
        \COMMENT{See Lemma~\ref{lem:collprob} and \eqref{eq:kj} for def. of $p_{\mathrm{near},j}$ and $k_j$} \alglinelabel{line:q-K2}
		\STATE $p_{\text{samp}} \gets \min \left\{\frac{1}{2^{j+1}\mu_i},1\right\}$ \alglinelabel{line:setpalg1} 
		\STATE Sample every $\mathbf{x} \in X$ w.p. $p_{\text{samp}}$, and let $Z_{\mu_i, j}$ be the set of sampled elements  \alglinelabel{line:sub-sample4} 
		\FOR {$\ell=1,2,\dots,K_{2,j}$}  
		\STATE Draw a hash function $H_{\mu_i,a,j,\ell}$ from hash family $\mathcal{H}^{k_j}$ (Lemma~\ref{lem:collprob}) 

        \FOR{$\mathbf{x} \in Z_{\mu_i, j}$}
        \STATE Store $\mathbf{x}$ in the bucket $B_{H_{\mu_i,a,j,\ell}}(\mathbf{x})$
        \ENDFOR

		\ENDFOR
		\ENDFOR
		\STATE Sample every $\mathbf{x} \in X$ w.p. $\frac{1}{2 n'}$, and let $\wt{X}_{\mu_i, a}$ be the set of sampled elements.\alglinelabel{line:beyondJsample}
		\STATE Store $\wt{X}_{\mu_i,a}$ \COMMENT{Set $\wt{X}_{\mu_i, a}$ will be used to recover points beyond $L_{J+1}$}
		\ENDFOR
    \ENDFOR
		\ENDPROCEDURE
	\end{algorithmic}
\end{algorithm}

\begin{algorithm}[H]
	\caption{DynamicKDE Query Procedures}
	\label{alg:querypoints} 
	\begin{algorithmic}[1]
		
		{\PROCEDURE{$\textsc{QueryMuEstimate}(X,\mathbf{q},\epsilon,\mu_i)$}{}
		\STATE {\textbf{Input:} set $X$ of data points, query point $\mathbf{q}$,  precision estimate $\epsilon$, and \textsf{KDE} estimate $\mu_i$.} 	
		\FOR[$O(\log n'/\epsilon^2)$ independent repetitions]{$a=1,2,\dots,K_1$}
		\FOR{$j=1,2,\dots,J_{\mu_i}$}
		\STATE $K_{2,j}\gets 200\log n' \cdot p_{\mathrm{near},j}^{-k_j}$ 
		\FOR {$\ell=1,2,\dots,K_{2,j}$}
		\STATE Recover points $\mathbf{x} \in B_{H_{\mu_i, a,j,\ell}}(\mathbf{q})$ such that $\mathbf{x} \in L^{\mathbf{q}}_j$ \alglinelabel{alg:querypoint:recovered_points_layer}
		\ENDFOR
		\ENDFOR

		\STATE Recover points $\mathbf{x} \in \wt{X}_{\mu_i, a}$ such that $\mathbf{x} \in L^{\mathbf{q}}_{J_{\mu_i}+1}$. \alglinelabel{alg:querypoint:recovered_points_outside_layer}
		\STATE $S \gets $ set of all recovered points in this iteration
		\FOR{$\mathbf{x}_{i'}\in S$} \alglinelabel{alg:querypoint:start_compute_estimate}
		\STATE $w_{i'} \gets k(\mathbf{x}_{i'},\mathbf{q})$
		\IF {$\mathbf{x}_{i'}\in L^{\mathbf{q}}_j$ for some $j\in[J_{\mu_i}]$}
		\STATE $p_{i'}\gets\min\left\{\frac{1}{2^{j+1}\mu_i},1\right\}$
		\ELSIF{$\mathbf{x}_i\in X\setminus \left(\bigcup_{j\in[J_{\mu_i}]}L^{\mathbf{q}}_j\right)$} 
		\STATE $p_{i'}\gets \frac{1}{2n'}$ 
		\ENDIF
		\ENDFOR
		\STATE $Z_{\mathbf{q}, a}\gets\sum_{\mathbf{x}_{i'}\in S} w_{i'} / p_{i'}$\alglinelabel{line:est4}
        \STATE Store $Z_{\mathbf{q}, a}$
		\ENDFOR
        \FOR[Get median of $\bar{N} = O(\log n)$ means of size $O(1/\varepsilon^{2})$]{$b=1,2,\dots,\bar{N}$} \alglinelabel{alg:query_point_start_median_compute}
        \STATE $\bar{Z}_{\mathbf{q},b} \gets \frac{1}{\bar{m}} \sum_{a=(b-1)\bar{m}+1}^{b\bar{m}} Z_{\mathbf{q},a}$ \alglinelabel{alg:querypoint:emperical_mean_Z}
        \ENDFOR
        \STATE \textbf{return} $\mathrm{Median}\left(\bar{Z}_{\mathbf{q}, 1}, \bar{Z}_{\mathbf{q}, 2}, \ldots, \bar{Z}_{\mathbf{q}, \bar{N}}\right)$ \alglinelabel{alg:querypoint:end_compute_estimate}
		\ENDPROCEDURE
    \STATE 
    \PROCEDURE{$\textsc{QueryPoint}(X,\mathbf{q},\epsilon)$}{}
    \STATE {\textbf{Input:}   set $X$ of data points, query point $\mathbf{q}$,   precision estimate $\epsilon$.}
    \FOR {$\mu_i \in [\mu_{\log(2n')}, \mu_{\log(2n') - 1}, \ldots, \mu_{1}, \mu_{0}] $} \alglinelabel{alg:querypoint:start_different_estimates}
    \IF {$\textsc{QueryMuEstimate}(X,\mathbf{q},\epsilon,\mu_i) > \mu_i$} \alglinelabel{alg:querypoint:estimate_check}
    \STATE \textbf{return} $\textsc{QueryMuEstimate}(X,\mathbf{q},\epsilon,\mu_{i+1})$ \COMMENT{If estimate is larger than $\mu_i$, return the previous estimate.} \alglinelabel{alg:querypoint:used_query_estimate}
    \ENDIF
    \ENDFOR
    \STATE \textbf{return} $0$ 
    \ENDPROCEDURE}
	\end{algorithmic}

\end{algorithm}

\begin{algorithm}[H]
  \caption{DynamicKDE Full Hash Procedures}
  \label{alg:full_hash_procedures}
  \begin{algorithmic}[1]

    \PROCEDURE{PreProcessQueryPoints($X, Q, \varepsilon$)}{}
    \STATE \textbf{Input:} the set $X$ of data points, the set of query points $Q$, and the precision estimate $\varepsilon$.
    \FOR{$\mathbf{q} \in Q$}
      \STATE $\hat{\mu}_\mathbf{q} \gets \textsc{QueryPoint}(X, \mathbf{q}, \varepsilon)$ 
             \alglinelabel{alg:full_hash_procedures:query_point}
      \STATE Store $\hat{\mu}_\mathbf{q}$
      \STATE Add $\mathbf{q}$ to all sets $Q_{\mu_i}$ such that $\hat{\mu}_\mathbf{q} \leq \mu_i$
    \ENDFOR

    \FOR{$\mu_i \in M$}
      \FOR{$a = 1, 2, \dots, K_1$}
        \FOR{$j = 1, 2, \dots, J_{\mu_i}$}
          \STATE $K_{2,j} \gets 200 \log n' \cdot p_{\mathrm{near},j}^{-k_j}$
          \FOR{$\ell = 1, 2, \dots, K_{2,j}$}
            \FOR{$\mathbf{q} \in Q_{\mu_i}$}
              \STATE Store $\mathbf{q}$ in full bucket
              $B^*_{H_{\mu_i,a,j,\ell}}(\mathbf{q})$ corresponding to hash value 
              $H_{\mu_i,a,j,\ell}(\mathbf{q})$ 
              \alglinelabel{alg:preprocess_dataset_copy_hash}
            \ENDFOR
          \ENDFOR
        \ENDFOR
      \ENDFOR
    \ENDFOR
    \ENDPROCEDURE
    \STATE
    \PROCEDURE{AddFullHash($\mathbf{q}$)}{}
    \STATE \textbf{Input:} new query point $\mathbf{q}$
    \FOR{$\mu_i \geq \hat{\mu}_\mathbf{q}$}
      \FOR{$a = 1, 2, \dots, K_1$}
        \FOR{$j = 1, 2, \dots, J_{\mu_i}$}
          \STATE $K_{2,j} \gets 200 \log n' \cdot p_{\mathrm{near},j}^{-k_j}$
          \FOR{$\ell = 1, 2, \dots, K_{2,j}$}
            \STATE Store $\mathbf{q}$ in full bucket 
            $B^*_{H_{\mu_i,a,j,\ell}}(\mathbf{q})$ corresponding to hash value 
            $H_{\mu_i,a,j,\ell}(\mathbf{q})$
          \ENDFOR
        \ENDFOR
      \ENDFOR
    \ENDFOR
    \ENDPROCEDURE

  \end{algorithmic}
\end{algorithm}

\begin{algorithm}[H]
  \caption{DynamicKDE Delete Procedures}
  \alglinelabel{alg:delete_querypoint}
  \begin{algorithmic}[1]

    \PROCEDURE{DeleteFromData($\mathbf{z}$)}{}
    \STATE \textbf{Input:} Query point $\mathbf{q}$ to remove
    \FOR{$\mu_i \in M$}
      \STATE Remove $\mathbf{q}$ from $Q_{\mu_i}$
      \FOR{$a = 1, 2, \dots, K_1$}
        \FOR{$j = 1, 2, \dots, J_{\mu_i}$}
          \STATE $K_{2,j} \gets 200 \log n' \cdot p_{\mathrm{near},j}^{-k_j}$
          \FOR{$\ell = 1, 2, \dots, K_{2,j}$}
            \STATE Remove $\mathbf{q}$ from full bucket $B^*_{H_{\mu_i,a,j,\ell}}(\mathbf{q})$
            \alglinelabel{alg:delete_querypoint_line}
          \ENDFOR
        \ENDFOR
      \ENDFOR
    \ENDFOR
    \ENDPROCEDURE

  \end{algorithmic}
\end{algorithm}

\subsection{Analysis of the $\textsc{Initialise}$ and $\textsc{AddQueryPoint}$ procedures}\label{sec:dynamic_kde_initialise_query}
We first analyse the $\textsc{Initialise}(X, Q, \varepsilon)$ and $\textsc{AddQueryPoint}(\mathbf{z})$ procedures in Algorithm~\ref{alg:dynamic_kde}, whose corresponding subprocedures are presented in Algorithms~\ref{alg:preprocess_dataset},~\ref{alg:querypoints}, and~\ref{alg:full_hash_procedures}. At a high level, the analysis begins by bounding the expected number of data points within each hash bucket (Lemma~\ref{lemma:number_in_query_bucket}). Subsequently, it establishes the query time complexity (Lemma~\ref{lem:query-time-gen}), shows  that the estimator produced by $\textsc{QueryMuEstimate}$ is nearly unbiased (Lemma~\ref{lem:unbias4}), and provides a good approximation to the true \KDE~value with high probability (Lemma~\ref{lemma:variance}).

We assume that $\mu_i$ is an estimate satisfying $\mu_\mathbf{q} \leq \mu_i$; we  will justify this assumption  in Remark~\ref{rem:estimate_mu}. We first analyse the expected number of data points to be sampled in each bucket $B_{H_{\mu_i,a,j,\ell}}(\mathbf{q})$.
\begin{lemma}\label{lemma:number_in_query_bucket}
    For  any $a$, $\mu_i$, $j$, $\ell$, it holds for  $\mathbf{q} \in Q_{\mu_{i}}$ that \[
     \mathbb{E}_{H_{\mu_i, a, j,\ell}} \left[ 
 \left|\left\{\mathbf{x} \in \wt{X}_{\mu_i, j} \mid H_{\mu_i,a,j,\ell}(\mathbf{q}) = H_{\mu_i,a,j,\ell}(\mathbf{x}) \right\} \right|\right] = \wt{O}(1),
    \]
     any for $\mathbf{q} \in Q_{\mu_{i'}}$ that
        \[
    \mathbb{E}_{H_{\mu_i, a, j,\ell}}  \left[ 
        \left|\left\{\mathbf{x} \in X \mid H_{\mu_i,a,j,\ell}(\mathbf{x}) = H_{\mu_i,a,j,\ell}(\mathbf{q})| \right\} \right|\right] = \widetilde{O}\left(2^{j+1} \mu_{i'}\right).
    \]
\end{lemma}
\begin{proof}
      We compute the expected number of collisions in the bucket $B_{H_{\mu_i,a,j,\ell}}(\mathbf{q})$, and our analysis is by case distinction. 
    
\textit{Case~1:  $j' \leq j$.} It holds by  Lemma~\ref{lem:sizeL}  that 
\[
\left|L^{\mathbf{q}}_{j'}\right| \leq 2^{j'} \mu_\mathbf{q} \leq 2^{j'} \mu_i,\]
which upper bounds the number of points that could collide from these geometric weight levels.  Since 
every data point is sampled with probability $1/\left (2^{j+1}\mu_{i}\right)$ in this iteration,   the expected number of sampled data points is $O(1)$.

\textit{Case~2: $j < j' \leq J_{\mathrm{\mu}_i} +1$. } 
We analyse the effect of the  LSH. Note that in the $j$th iteration, we choose an LSH function whose corresponding distance level is $r_j$, and use
\begin{align*}
		k \triangleq k_j = - \frac{1}{ \log p_{\mathrm{near},j}} \cdot\max_{i=j+1,\ldots,J_{\mu_i}+1}  \left\lceil  \frac{i-j}{c_{i,j}^2(1-\beta)}     \right\rceil.
	\end{align*}
	as the number of concatenations. 
	Then, it holds  for $\mathbf{p}\in L^{\mathbf{q}}_{j'}$ that
	\begin{align*}
	\mathbb{P}_{H^*\in \mathcal{H}^{k}}\left[  H^*(\mathbf{p})=H^*(\mathbf{q})\right] \le p^{kc^2(1-\beta)},
	\end{align*}
	where $c \triangleq c_{i,j}= \min\left\{   \frac{r_{i-1}}{r_j},O \lp \log^{1/7}n' \rp\right\}$ and $p \triangleq p_{\mathrm{near},j}$. Hence,  the expected number of points from weight level $L^{\mathbf{q}}_{j'}$ in the query hash bucket is  
 $	O \lp 2^{j' - j} \rp \cdot p^{kc^2(1-\beta)}=\wt{O}(1)$,
 where the last line holds by the choice of $k$. 
 Combining the two cases proves the first statement.

  The second statement holds for the same analysis, but   we have instead that $|L^{\mathbf{q}}_{j'}| \leq 2^{j'} \mu_\mathbf{q} \leq 2^{j'} \mu_{i'}$ for $\mathbf{q} \in Q_{\mu_{i'}}$.
\end{proof}

The query time complexity for $\textsc{QueryMuEstimate}(X,\mathbf{q},\epsilon,\mu_i)$  follows from \cite{charikarKDE}.

\begin{lemma}[Query Time Complexity, \cite{charikarKDE}]\label{lem:query-time-gen}
	For any kernel $k$, the expected running time of   $\textsc{QueryMuEstimate}(X,\mathbf{q},\epsilon,\mu_i)$  (Algorithm~\ref{alg:querypoints}) is  $\epsilon^{-2}\cdot n_1^{o(1)}\cdot \mathrm{cost}(k)$.
\end{lemma}

Next we show that our returned estimator gives a good approximation with high probability. 
 
\begin{lemma}\label{lem:unbias4}
	 For any $\mathbf{q}\in \R^d$,   $\mu_\mathbf{q}\in (0,2n_1]$,   $\mu_i \ge \mu_\mathbf{q}$,   $\epsilon\in \left(1 / n_1^5,1\right)$, the  estimator $Z_{\mathbf{q},a}$ for   $a\in [K_1]$ constructed in $\textsc{QueryMuEstimate}(X,\mathbf{q},\epsilon,\mu_i)$ (Algorithm~\ref{alg:querypoints}) satisfies that
	\[
	\left(1-n_1^{-9}\right)\cdot \mu_{\mathbf{q}}\leq \mathbb{E}[Z_{\mathbf{q},a}] \leq \mu_{\mathbf{q}}.
	\]
\end{lemma}

\begin{proof}
We first fix   arbitrary $j=j^*$ and $a=a^*$, and sample some point $\mathbf{p} \in L^{\mathbf{q}}_{j^*}$. 
  By Lemma~\ref{lem:collprob} we have 
	\begin{align*}
\mathbb{P}_{H^*\sim\mathcal{H}^{k_j}}\left[ H^*(\mathbf{p})=H^*(\mathbf{q})\right] \ge p_{\mathrm{near},j}^{k_{j}}.
	\end{align*}
 Since  we repeat this process for $K_{2,j}=200\log n_1'\cdot p_{\mathrm{near},j}^{-k_j}$ times,  it holds with high probability that 
any sampled point   $\mathbf{p}$ from band $L^{\mathbf{q}}_{j^*}$ is recovered in at least one  phase. By applying the union bound, the probability that all the sampled points are recovered is at least $1-n_1^{-9}$.

We define $Z \triangleq Z_{\mathbf{q},a}$, and have that  \[
\mathbb{E}[Z] = \sum_{i=1}^{n_1}\frac{\mathbb{E}[\chi_{i}]}{p_{i}}\cdot w_{i},\]
where $\chi_i=1$ if point $\mathbf{x}_i$ is sampled and $\chi_i=0$ otherwise. Hence, it holds that $\left(1-n_1^{-9}\right) p_{i}\leq \mathbb{E}[\chi_{i}]\leq p_{i},$
which implies that 
	$
 \left(1-n_1^{-9}\right)\mu_{\mathbf{q}}\leq \mathbb{E}[Z] \leq \mu_{\mathbf{q}}$.
\end{proof}

\begin{remark}\label{rem:estimate_mu}
  Lemma~\ref{lem:unbias4} shows that the estimator  $Z_{\mathbf{q}, a}$ is unbiased (up to some small inverse polynomial error) for   any choice of $\mu_i \ge \mu_{\mathbf{q}}$.
 Therefore, when  $\mu_i \geq 4\mu_{\mathbf{q}}$,  by Markov's inequality  the probability that the estimator's returned value is larger than   $\mu_i$ is at most  $1/4$. By taking $O(\log n_1)$ independent estimates, one can conclude that $\mu_i$ is higher than $\mu_{\mathbf{q}}$ if the median of the estimated values is below $\mu_i$, and this estimate is correct with high probability. This is achieved on Lines~\ref{alg:query_point_start_median_compute}--\ref{alg:querypoint:end_compute_estimate} of Algorithm~\ref{alg:querypoints}. To ensure that we find a value of $\mu_i$ that satisfies $\mu_i/4< \mu_{\mathbf{q}}\le \mu_i$ with high probability, 
 on Lines~\ref{alg:query_point_start_median_compute}--\ref{alg:querypoint:used_query_estimate} the algorithm starts with $\mu_i = 2n_1$ and   repeatedly halves the estimate until  finding an estimate $\hat{\mu}_\mathbf{q} > \mu_i$; at this point the algorithm  returns the previous estimate based on $\mu_{i+1}$. 
\end{remark}

\begin{lemma}[\cite{charikarKDE}]\label{lemma:variance}
For every $\mathbf{q}\in \R^d$,   $\mu_\mathbf{q}\in (0,2n_1]$, $\epsilon\in \left(1 / n_1^{5},1\right)$,  
and $\mu_i$ satisfying $ \mu_i/4\le \mu_\mathbf{q}\le \mu_i$, the procedure  \textsc{QueryMuEstimate}($X,\mathbf{q},\epsilon,\mu_i$) (Algorithm~\ref{alg:querypoints} in the appendix) outputs a $(1\pm \epsilon)$-approximation to $\mu_\mathbf{q}$ with high probability.  
\end{lemma}
\begin{proof}
    	 Let $Z \triangleq Z_{\mathbf{q},a}$, and we have that 
	\begin{align}
	\mathbb{E}[Z^2]	&=\mathbb{E}\left[\left(\sum_{\mathbf{p}_i\in X} \chi_i\cdot  \frac{w_i}{p_i}\right)^2\right]  \nonumber\\
	&=\sum_{i\ne j} \mathbb{E}\left[\chi_i\chi_j\cdot  \frac{w_iw_j}{p_ip_j}\right]+\sum_{i\in [n_1]}\mathbb{E}\left[\chi_i \cdot \frac{w_i^2}{p_i^2}\right] \nonumber\\
	&\le \sum_{i\ne j}w_iw_j + \sum_{i\in [n_1]}\frac{w_i^2}{p_i}\cdot \mathbb{I}[p_i=1]+ \sum_{i\in [n_1]}\frac{w_i^2}{p_i}\cdot \mathbb{I}[p_i\ne1] \nonumber \\
	&\le \left(\sum_{i\in [n_1]}w_i\right)^2  +\max_{i\in [n_1]} \left\{\frac{w_i}{p_i}\cdot \mathbb{I}[p_i\ne 1]\right\}\sum_{i\in [n_1]}w_i\nonumber \\
	&\le 2\left(\mu_\mathbf{q}\right)^2+ \max_{j\in [J_{\mu_i}],\mathbf{p}_i\in L^\mathbf{q}_{j}}\left\{w_{i}\cdot 2^{j+1}\right\}\cdot \mu_i \cdot \mu_{\mathbf{q}} \nonumber\\
	&\le 8\mu_i^2, \label{eq:upper_bound_variance}
	\end{align}
        where the second inequality   follows from \[\frac{w_{i}^{2}}{p_{i}}\cdot  \mathbb{I}[p_{i}=1] \leq w_i^2\]and \[(\sum_{i} w_{i})^{2} = \sum_{i \neq j} w_{i} w_{j} + \sum_{i \in[n_1]} {w_{i}^{2}},
        \]
        and the third inequality follows from $(\mu_{\mathbf{q}})^2 = (\sum_{i} w_{i})^{2} $ and $p_j \geq 1/\left(2^{j+1} \mu_i \right)$.

	 Let \[\bar{Z} \triangleq \bar{Z}_{\mathbf{q},b} = \frac{1}{\bar{m}} \sum_{a=(b-1)\bar{m}+1}^{b\bar{m}} Z_{\mathbf{q},a}\] be the empirical mean of $\bar{m}$ such estimates, as computed on Line~\ref{alg:querypoint:emperical_mean_Z} of Algorithm~\ref{alg:querypoints}. We have that
	\begin{align*}
	\mathbb{P}\left[|\bar{Z}-\mu_{\mathbf{q}}|\geq \epsilon \mu_{\mathbf{q}}\right] &\le \mathbb{P}\left[|\bar{Z}-\mathbb{E}[Z]|\geq \epsilon \mu_{\mathbf{q}}-\left|\mathbb{E}[Z]-\mu_{\mathbf{q}}\right|\right]\\
	&\le \mathbb{P}\left[|\bar{Z}-\mathbb{E}[Z]|\geq (\epsilon-n_1^{-9})\mu_{\mathbf{q}}\right]\\
	&\le \frac{\mathbb{E}[\bar{Z}^{2}]}{\left(\epsilon-n_1^{-9}\right)^{2}(\mu_{\mathbf{q}})^{2}}\\
	&\le \frac{1}{\bar{m}}\frac{128(\mu_{\mathbf{q}})^{2}}{\left(\epsilon-n_1^{-9}\right)^{2}(\mu_{\mathbf{q}})^{2}},
	\end{align*}
	where the first inequality follows from $ | \bar{Z} - \mu_{\mathbf{q}} | \leq |\bar{Z} - \mathbb{E}[Z]| + | \mathbb{E}[Z] - \mu_{\mathbf{q}} |$, the second one follows from $\mathbb{E}[Z] \geq (1 - n_1^{-9}) \mu_{\mathbf{q}}$ (Lemma~\ref{lem:unbias4}), the third one follows from Chebyshev's inequality and the last one follows from $\mathbb{E}[\bar{Z}^2] \leq \mathbb{E}[Z^2]/\bar{m} \leq 8 \mu_i^2/\bar{m} $ and $\mu_i \leq 4 \mu_{\mathbf{q}}$. By setting  $\bar{m}=\frac{C}{\epsilon^{2}}$ for a large enough constant $C$ and taking the median of $O(\log(1/\delta))$ of these means we achieve a $(1\pm \epsilon)$-approximation with probability at least $1-\delta$ per query.
\end{proof}

\subsection{Analysis of the $\textsc{AddDataPoint}$ procedure}\label{sec:dynamic_kde_add_data}
We now analyse  $\textsc{AddDataPoint}(\mathbf{z})$  in Algorithm~\ref{alg:dynamic_kde}. 
If the number of data points has doubled, $\textsc{AddDataPoint}(\mathbf{z})$ calls the $\textsc{Initialise}(X, Q, \varepsilon)$ procedure. Otherwise,  $\textsc{AddDataPointAndUpdateQueries}(\mathbf{z}, Q, \varepsilon)$ is called, which we describe in Algorithm~\ref{alg:addpointandupdatequeries}.

\begin{algorithm}[H]
  \caption{DynamicKDE Update Procedures}
  \label{alg:addpointandupdatequeries}
  \begin{algorithmic}[1]

    \PROCEDURE{AddPointAndUpdateQueries($\mathbf{z}, Q, \varepsilon$)}{}
    \STATE \textbf{Input:} New data point $\mathbf{z}$, the set of query points $Q$, and the precision estimate $\varepsilon$.
    \FOR{$\mu_i \in M$} \label{alg:addupdate:line:mu_estimate}
      \FOR{$a = 1, 2, \dots, K_1$}
        \FOR{$j = 1, 2, \dots, J_{\mu_i}$}
          \STATE $p_{\text{sampling}} \gets \min \Bigl\{\frac{1}{2^{j+1}\mu_i}, 1\Bigr\}$
          \IF{$\mathbf{z}$ is sampled with probability $p_{\text{sampling}}$}
            \label{alg:addupdate:line:samplepoint_start}
            \STATE $K_{2,j} \gets 200 \log n' \cdot p_{\mathrm{near},j}^{-k_j}$
            \FOR{$\ell = 1, 2, \dots, K_{2,j}$}
              \label{alg:addupdate:line:start_resample_loop}
              \STATE Store $\mathbf{z}$ in bucket $B_{H_{\mu_i,a,j,\ell}}(\mathbf{z})$
                \label{alg:addupdate:line:hash_point}
              \STATE Recover $\mathbf{q} \in B^*_{H_{\mu_i,a,j,\ell}}(\mathbf{z})$ 
                such that $\mathbf{q} \in Q_{\mu_i} \setminus \bigcup_{j' < i} Q_{\mu_{j'}}$ 
                and $\mathbf{q} \in L^{\mathbf{z}}_j$
                \label{alg:addupdate:line:samplepoint_end}
            \ENDFOR
          \ENDIF
        \ENDFOR

        \STATE Sample $\mathbf{z}$ with probability $\frac{1}{2n'}$ 
          \label{alg:addupdate:outside_geometric_level_start}
        \IF{$\mathbf{z}$ is sampled}
          \STATE Add $\mathbf{z}$ to $\widetilde{X}_{\mu_i,a}$ 
            \label{alg:addupdate:outside_geometric_level_addpoint}
          \STATE Recover $\mathbf{q} \in Q$ such that 
            $\mathbf{q} \in Q_{\mu_i} \setminus \bigcup_{j' < i} Q_{\mu_{j'}}$ 
            and $\mathbf{q} \in L^{\mathbf{z}}_{J_{\mu_i}+1}$
            \label{alg:addupdate:q_recover_outside_geometriclevel}
        \ENDIF \label{alg:addupdate:outside_geometric_level_end}

        \STATE $S \gets \text{set of all recovered points from the full hash in this iteration}$
          \label{alg:addupdate:recovered_points}
        \FOR{$\mathbf{q} \in S$}
          \label{alg:addupdate:update_estimates_start}
          \STATE $w_\mathbf{q} \gets k(\mathbf{z}, \mathbf{q})$
          \IF{$\mathbf{z} \in L^{\mathbf{q}}_j$ for some $j \in [J_{\mu_i}]$}
            \STATE $p_{\mathbf{q}} \gets \min \Bigl\{\frac{1}{2^{j+1}\mu_i}, 1\Bigr\}$
          \ELSIF{$\mathbf{z} \in X \setminus \bigcup_{j \in [J_{\mu_i}]} L^{\mathbf{q}}_j$}
            \STATE $p_{\mathbf{q}} \gets \frac{1}{2n'}$
          \ENDIF
          \STATE $Z_{\mathbf{q}, a} \mathrel{+}= w_{\mathbf{q}} / p_{\mathbf{q}}$
            \label{alg:addupdate:line:update_estimates_end}
          \STATE $\bar{Z}_{\mathbf{q}, \lceil a/\bar{m} \rceil} \mathrel{+}= 
            w_{\mathbf{q}} / (\bar{m}  p_{\mathbf{q}})$
            \COMMENT{Update empirical mean (Algorithm~\ref{alg:querypoints}, 
            Line~\ref{alg:querypoint:emperical_mean_Z})}
          \STATE $\hat{\mu}_\mathbf{q} \gets \mathrm{Median}\Bigl(
            \bar{Z}_{\mathbf{q},1}, \bar{Z}_{\mathbf{q},2}, \ldots, \bar{Z}_{\mathbf{q},\bar{N}}
            \Bigr)$
            \COMMENT{Update median}
            \label{alg:addupdate:line:update_estimate_median}
          \IF{$\hat{\mu}_\mathbf{q} > \mu_i$}
            \label{alg:addupdate:line:overestimate}
            \STATE $Q_{\mu_i} \gets Q_{\mu_i} \setminus \{\mathbf{q}\}$
            \STATE Remove $\mathbf{q}$ from every 
              $B^*_{H_{\mu_i,a',j',\ell'}}(\mathbf{q})$ 
              for all $a' \in [K_1], j' \in [J_{\mu_i}], \ell' \in K_{2,j}$
              \label{alg:addupdate:remove_datapoint_full_hash}
            \STATE $\hat{\mu}_\mathbf{q} \gets \textsc{QueryPoint}(X, \mathbf{q}, \varepsilon)$
              \label{alg:addupdate:re-estimate_query}
          \ENDIF
          \label{alg:addupdate:update_estimates_end}
        \ENDFOR
      \ENDFOR
    \ENDFOR
    \ENDPROCEDURE

  \end{algorithmic}
\end{algorithm}

To analyse the performance of $\textsc{AddDataPoint}(\mathbf{z})$, the analysis first establishes that the \KDE~estimates are correctly updated after running the procedure (Lemma~\ref{lemma:correctness_update_incremental}), ensuring that $\hat{\mu}_\mathbf{q}$ remains a $(1 \pm \epsilon)$-approximation of $\mu_\mathbf{q}$ with high probability. Subsequently, it demonstrates that the total number of times any individual query point $\mathbf{q}$ is updated during a sequence of $T$ data point insertions is $\wt{O}(1)$ with high probability (Lemma~\ref{lemma:bound_number_updates}).

\begin{lemma}\label{lemma:correctness_update_incremental}
    After running the $\textsc{AddDataPoint}(\mathbf{z})$ procedure, it holds with high probability for every $\mathbf{q} \in Q$ that $\hat{\mu}_\mathbf{q}$ is a $(1 \pm \epsilon)$-approximation of $\mu_\mathbf{q}$. 
\end{lemma}
\begin{proof} 
     We prove that running  the initialisation procedure $\textsc{Initialise}(X \cup \mathbf{z}, Q, \varepsilon)$ is the same as  running $\textsc{Initialise}(X, Q, \varepsilon)$,  and then running $\textsc{AddDataPoint}(\mathbf{z})$.
\begin{itemize}[leftmargin=1cm]
    \item  We first examine  the $\textsc{AddPointAndUpdateQueries}(\mathbf{z}, Q, \varepsilon)$ procedure (Algorithm~\ref{alg:addpointandupdatequeries}). On  Lines~\ref{alg:addupdate:line:samplepoint_start}--\ref{alg:addupdate:line:hash_point} and Lines~\ref{alg:addupdate:outside_geometric_level_start}--\ref{alg:addupdate:outside_geometric_level_addpoint}, the algorithm samples the point $\mathbf{z}$ with probability $\min\left\{\frac{1}{2^{j+1} \mu_i}, 1\right\}$ and $\frac{1}{2n'}$, respectively: If $\mathbf{z}$ is sampled in Lines~\ref{alg:addupdate:line:samplepoint_start}--\ref{alg:addupdate:line:hash_point},  the algorithm stores $\mathbf{z}$ in the bucket corresponding to hash value $H_{\mu_i, a, j, \ell}(\mathbf{z})$;  if $\mathbf{z}$ is sampled in  Lines~\ref{alg:addupdate:outside_geometric_level_start}--\ref{alg:addupdate:outside_geometric_level_addpoint}, the algorithm adds $\mathbf{z}$ to the set $\wt{X}_{\mu_i, a}$. This is the same as  sampling and storing  $\mathbf{z}$ in the $\textsc{Preprocess}(X \cup \mathbf{z}, \varepsilon)$ procedure (Algorithm~\ref{alg:preprocess_dataset}) on Lines~\ref{line:sub-sample4} and~\ref{line:beyondJsample}, which is called during $\textsc{Initialise}(X \cup \mathbf{z}, Q, \varepsilon)$.
    Hence, after running $\textsc{Initialise}(X, Q, \varepsilon)$ followed by $\textsc{AddPointAndUpdateQueries}(\mathbf{z}, Q, \varepsilon)$, the stored points in $H_{\mu_i, a, j, \ell}$ and $\wt{X}_{\mu_i, a}$ for all $\mu_i \in M$, $a\in [K_1]$, $j \in [J_{\mu_i}]$, and $\ell \in [K_{2,j}]$ are the same as the ones after running $\textsc{Initialise}(X \cup \mathbf{z}, Q, \varepsilon)$.

    \item   Next, we prove that the estimates $\hat{\mu}_\mathbf{q}$ are updated correctly for every $\mathbf{q} \in Q$.
Without loss of generality, let $\mathbf{q}$ be a query point such that $\mathbf{q} \in Q_{\mu_i} \setminus \bigcup_{j' < i} Q_{\mu_{j'}}$. 
    
    \begin{itemize}
        \item  We first note that, when running $\textsc{Initialise}(X, Q, \varepsilon)$ (Algorithm~\ref{alg:dynamic_kde}), the \textsf{KDE} estimate for $\mathbf{q}$ is returned by running $\textsc{QueryMuEstimate}(X, \mathbf{q}, \varepsilon, \mu_i)$ (Line~\ref{alg:querypoint:used_query_estimate} of Algorithm~\ref{alg:querypoints}).
        \item When running $\textsc{Initialise}(X \cup \mathbf{z}, Q, \varepsilon)$, if $\mathbf{z}$ is sampled on Line~\ref{line:sub-sample4} during $\textsc{Preprocess}(X \cup \mathbf{z}, \varepsilon)$ for any iteration $a\in [K_1]$, $j \in [J_{\mu_i}]$, then $\mathbf{z}$  is stored in the  bucket $B_{H_{\mu_i, a, j, \ell}} (\mathbf{z})$   for all $\ell \in [K_{2,j}]$. Moreover, if  $\mathbf{z}$ is sampled on Line~\ref{line:beyondJsample} for any iteration $a\in [K_1]$,  then $\mathbf{z}$ is stored in $\wt{X}_{\mu_i, a}$.
    In this case, if $H_{\mu_i, a, j, \ell}(\mathbf{q}) = H_{\mu_i, a, j, \ell}(\mathbf{z})$ and $\mathbf{z} \in L^\mathbf{q}_j$ for some $\ell \in [K_{2,j}]$, or if $\mathbf{z} \in L^{\mathbf{q}}_{J_{\mu_i}+1}$, then $\mathbf{z}$ would be included in the set of recovered points $S$ for the iteration $a\in [K_1]$, and consequently in the estimator $Z_{\mathbf{q}, a}$  when $\textsc{QueryMuEstimate}(X \cup \mathbf{z}, \mathbf{q}, \varepsilon, \mu_i)$ is called during $\textsc{PreProcessQueryPoints}(X \cup \mathbf{z}, Q, \varepsilon)$ (Algorithm~\ref{alg:full_hash_procedures}).
        
        \item  On the other hand, we notice that,  when running $\textsc{AddPointAndUpdateQueries}(\mathbf{z}, Q, \varepsilon)$, $\mathbf{q}$ is recovered if   (i) $H_{\mu_i,a,j,\ell}(\mathbf{q}) = H_{\mu_i,a,j,\ell}(\mathbf{z})$ and $\mathbf{q} \in L_{j}^\mathbf{z}$ (Line~\ref{alg:addupdate:line:samplepoint_end} of Algorithm~\ref{alg:addpointandupdatequeries})   or  (ii) $\mathbf{q} \in L_{J_{\mu_i} + 1}^\mathbf{z}$  (Line~\ref{alg:addupdate:q_recover_outside_geometriclevel} of Algorithm~\ref{alg:addpointandupdatequeries}). 
          Furthermore, (i) $\mathbf{q} \in L^\mathbf{z}_{j'}$ if and only if $\mathbf{z} \in L^\mathbf{q}_{j'}$ for any $j' \in [J_{\mu_i}+1]$, and (ii) the buckets $B^*_{H_{\mu_i,a,j,\ell}}(\mathbf{q})$ and $B_{H_{\mu_i,a,j,\ell}}(\mathbf{q})$ are populated using the same hash function (Line~\ref{alg:preprocess_dataset_copy_hash} of Algorithm~\ref{alg:full_hash_procedures}).
        Therefore, $\mathbf{q}$ is recovered at iteration $a \in [K_1]$ when running $\textsc{AddPointAndUpdateQueries}(\mathbf{z}, Q, \varepsilon)$ if and only if $\mathbf{z}$ is included in the estimator $Z_{\mathbf{q}, a}$ when $\textsc{QueryMuEstimate}(X \cup \mathbf{z}, \mathbf{q}, \varepsilon, \mu_i)$ is called   during $\textsc{PreProcessQueryPoints}(X \cup \mathbf{z}, Q, \varepsilon)$.
          Then, the estimator $Z_{\mathbf{q}, a}$ is updated accordingly by adding $\mathbf{z}$ through Lines~\ref{alg:addupdate:update_estimates_start}--\ref{alg:addupdate:line:update_estimate_median} of Algorithm~\ref{alg:addpointandupdatequeries} as it would be done through Lines~\ref{alg:querypoint:start_compute_estimate}--\ref{alg:querypoint:end_compute_estimate} of Algorithm~\ref{alg:querypoints}.
        \item  Finally, if  $\hat{\mu}_{\mathbf{q}} > \mu_i$, then we re-estimate the query point on Lines~\ref{alg:addupdate:line:overestimate}--\ref{alg:addupdate:re-estimate_query}, to ensure we have the correct estimate $\hat{\mu}_\mathbf{q}$. We also update the set $Q_{\mu_i}$ accordingly, and   remove $\mathbf{q}$ from every $B^*_{H_{\mu_i, a', j', \ell'}}(\mathbf{z})$ for all $a' \in [K_1]$, $j' \in [J_{\mu_i}]$, and $\ell' \in [K_{2,j}]$.
    \end{itemize}
\end{itemize}
Combining everything together, we have shown that performing the initialisation procedure $\textsc{Initialise}(X \cup \mathbf{z}, Q, \varepsilon)$ is the same as  running $\textsc{Initialise}(X, Q, \varepsilon)$, followed by $\textsc{AddDataPoint}(\mathbf{z})$.  
\end{proof}

Next, we prove how many times any individual query point $\mathbf{q}$ is updated as the data points are inserted using the $\textsc{AddDataPoint}(\mathbf{z})$ procedure. 
 Let 
$X^{\mathbf{q}}_1 \triangleq \{\mathbf{x}_1, \ldots, \mathbf{x}_{n_1}\}$ 
be the set of  points presented at the time when $\mathbf{q}$ is added, and let 
$Z_T \triangleq \{\mathbf{z}_1, \ldots, \mathbf{z}_T\}$
be the points added up until the query time $T$. We use $\mathbf{z}_t$ to denote the new point added at time $t$. Note that it holds   for the points $X_T$ at time $T \geq 1$ that
$X^{\mathbf{q}}_T = X^{\mathbf{q}}_1 \cup Z_T.$
 Next, we define the event $\mathcal{F}^{\mathbf{q}}_{t}$ that
\begin{equation}\label{eq:correct_estimate_time_t}
\widehat{\mu}_{\mathbf{q}, t} \in (1 \pm \varepsilon) \cdot k(\mathbf{q}, X_t), 
\end{equation}
where $\widehat{\mu}_{\mathbf{q}, t}$ is the maintained query estimate for $\mathbf{q}$ at time $t$ from Algorithm~\ref{alg:dynamic_kde}. By Lemma~\ref{lemma:variance}, we know that $\mathcal{F}^{\mathbf{q}}_{t}$ happens with high probability. Moreover, for a large enough constant $C$ on Line~\ref{alg:dynamic_kde_constant} of Algorithm~\ref{alg:dynamic_kde}, we can ensure that this happens with high probability at every time step $t$. Therefore in the following we assume $\mathcal{F}^{\mathbf{q}}_{t}$ happens.
 We also introduce the following notation.
\begin{definition}\label{def:last_mu_upperbound}
    We define $T_{\mu_i}^\mathbf{q}$ to be the time step such that $$k(\mathbf{q}, X^{\mathbf{q}}_{T_{\mu_i}^\mathbf{q}}) = k(\mathbf{q}, X^{\mathbf{q}}_1) + \sum_{t=1}^{T_{\mu_i}^\mathbf{q}}k(\mathbf{q}, \mathbf{z}_t) \leq \mu_i$$ and $$k(\mathbf{q}, X^{\mathbf{q}}_{T_{{\mu_i}}^\mathbf{q} + 1}) = k(\mathbf{q}, X^{\mathbf{q}}_1) + \sum_{t=1}^{T_{\mu_i}^{\mathbf{q}}+1}k(\mathbf{q}, \mathbf{z}_t) > \mu_i.$$ 
\end{definition}
 By definition, $T_{\mu_i}^\mathbf{q}$ is the last time step at which the \textsf{KDE} value of $\mathbf{q}$ is at most $\mu_i$. The next lemme analyses  the   number of times a query point $\mathbf{q}$ is updated.

\begin{lemma}\label{lemma:bound_number_updates}
    Let $\mathbf{q}$ be a maintained query point by our algorithm. Then the total number of updates $\mathcal{U}^{\mathbf{q}}_T$ during $T$ insertions is, with high probability,
$\mathcal{U}^{\mathbf{q}}_T = \wt{O}(1).$
\end{lemma}

\begin{proof} 
Since the \textsf{KDE} data structure can be re-initialised at most $\log(T) = \widetilde{O}(1)$ times~(cf. Line~\ref{line:dynamic_kde:recompute_start} of Algorithm~\ref{alg:dynamic_kde}) through the sequence of $T$ updates,  it suffices for us to analyse the number of times $\mathbf{q}$ is updated between different re-initialisations; we assume this in the remaining part of the proof.
 To analyse the expected number of times that $\mathbf{q}$ (Line~\ref{alg:addupdate:line:update_estimates_end} of Algorithm~\ref{alg:addpointandupdatequeries}) is updated throughout the sequences of updates $Z_T$, we define the   random variable $Y^{\mathbf{q}}_{a, t}$ by
\[
    Y^{\mathbf{q}}_{a, t} \triangleq \twopartdefow{1}{\text{estimate } Z_{\mathbf{q},a} \text{ is updated at time } t}{0}{}
\]
 Let $\mathcal{E}^{\mathbf{q}}_{a, t}$ be the event that estimate $Z_{\mathbf{q},a}$ is updated at time $t$, and  we assume without loss of generality that  $T^{\mathbf{q}}_{\mu_{i'-1}} < t \leq T^{\mathbf{q}}_{\mu_{i'}}$ for some $\mu_{i'}$. First note that estimate $Z_{\mathbf{q},a}$ is updated if $\mathbf{q}\in S$  (Line~\ref{alg:addupdate:recovered_points} of Algorithm~\ref{alg:addpointandupdatequeries}). Furthermore, because $\widehat{\mu}_{\mathbf{q}, {t-1}} \in (1 \pm \varepsilon) \cdot k(\mathbf{q}, X_{t-1})$ by \eqref{eq:correct_estimate_time_t} and  $T^{\mathbf{q}}_{\mu_{i'-1}} < t \leq T^{\mathbf{q}}_{\mu_{i'}}$, one of the following holds: 
\begin{enumerate}[label=(\roman*), leftmargin=1cm]
    \item  $\mathbf{q} \in Q_{\mu_{i'}} \setminus \left( \bigcup_{j' < i'} Q_{\mu_{j'}}\right) $;
    \item $\mathbf{q} \in Q_{\mu_{i'-1}} \setminus \left( \bigcup_{j' < i'-1} Q_{\mu_{j'}}\right) $;
    \item $\mathbf{q} \in Q_{\mu_{i'+1}} \setminus  \left(\bigcup_{j' < i'+1} Q_{\mu_{j'}}\right)$.
\end{enumerate} 
Additionally, it   holds that $\mathbf{q} \in L^{\mathbf{z}_t}_{j'}$ for some $j' \in J_{\mu_{i'}}$.

 By these conditions, $\mathbf{q}$ is included in $S$ at time $t$ if and only if $\mathbf{z}_t$ is sampled at either the iteration for   $\mu_{i'} \in M$, $\mu_{i'-1} \in M$ or $\mu_{i'+1} \in M$  (Line~\ref{alg:addupdate:line:mu_estimate} of Algorithm~\ref{alg:addpointandupdatequeries}), and the corresponding iteration $j' \in J_{\mu_{i'+1}}$ on Line~\ref{alg:addupdate:line:samplepoint_start} of Algorithm~\ref{alg:addpointandupdatequeries}. Therefore, it holds for $t\in (T^{\mathbf{q}}_{\mu_{i'-1}}, T^{\mathbf{q}}_{\mu_{i'}}]$ that 
\begin{equation}\label{eq:upper_bound_update_probability}
\mathbb{P}\left[\mathcal{E}^{\mathbf{q}}_{a, t}\right] \leq \frac{1}{2^{j'+1} \cdot \mu_{i'-1}} \leq \frac{1}{2^{j'} \cdot \mu_{i'}} \leq \frac{2 k(\mathbf{q}, \mathbf{z}_t)}{\mu_{i'}}, 
\end{equation}
  where the last inequality uses the fact that    $\mathbf{q} \in L^{\mathbf{z}_t}_{j'}$ (Definition~\ref{def:geometric_weight_levels}). Similarly, we have that  
\[
\mathbb{P}\left[\mathcal{E}^{\mathbf{q}}_{a, t}\right] \geq \frac{ k(\mathbf{q}, \mathbf{z}_t)}{4 \mu_{i'}}.
\]
 Let
$\mathcal{U}^{\mathbf{q}}_T \triangleq \sum_{a=1}^{K_1} \sum_{t=1}^{T} Y^{\mathbf{q}}_{a, t}$
  be the total number of times that the query point $\mathbf{q}$ is updated, and we have that
 \begin{align*}
    \mathbb{E}\left[ \mathcal{U}^{\mathbf{q}}_T \right] &= \sum_{a=1}^{K_1} \sum_{t=1}^{T} \mathbb{E}\left[Y^{\mathbf{q}}_{a, t}\right] \\
    &= \sum_{a=1}^{K_1} \sum_{t=1}^{T} \mathbb{P}\left[\mathcal{E}^{\mathbf{q}}_{a, t}\right] \\
    &= \sum_{a=1}^{K_1} \sum_{\mu_{i'} \in M} \sum_{t=T^{\mathbf{q}}_{\mu_{i'-1}}}^{T^\mathbf{q}_{\mu_{i'}}} \mathbb{P}\left[\mathcal{E}^{\mathbf{q}}_{a, t}\right] \\
    &\leq \sum_{a=1}^{K_1} \sum_{\mu_{i'} \in M} \sum_{t=T^{\mathbf{q}}_{\mu_{i'-1}}}^{T^\mathbf{q}_{\mu_{i'}}} \frac{2 k(\mathbf{q}, \mathbf{z}_t)}{\mu_{i'}} \\
    &\leq \sum_{a=1}^{K_1} \sum_{\mu_{i'} \in M} \frac{2 \mu_{i'}}{\mu_{i'}} \\
    & = 2 \cdot K_1 \cdot |M|\\
    & =  \wt{O}(1),
\end{align*}
 where the first inequality follows by  \eqref{eq:upper_bound_update_probability} and the second one holds by the fact that  \[\sum_{t=T^{\mathbf{q}}_{\mu_{i'-1}}}^{T^\mathbf{q}_{\mu_{i'}}} k(\mathbf{q}, \mathbf{z}_t) \leq \mu_{i'}.\] Similarly, we have that
\[
 \mathbb{E}\left[ \mathcal{U}^{\mathbf{q}}_T \right] \geq \frac{1}{4} \cdot K_1 \cdot |M|   .\]
By the Chernoff bound,  it holds that
\begin{equation*}
    \mathbb{P}\left[ \mathcal{U}^{\mathbf{q}}_T \geq 10 \cdot \mathbb{E}\left[ \mathcal{U}^{\mathbf{q}}_T \right] \right] \leq \left( \frac{\mathrm{e}^9}{10^{10}} \right)^{K_1\cdot |M|/4}  \leq \exp\lp -K_1 \cdot |M|\rp = o(n^{-c})
\end{equation*}
for some constant $c$, and we have  with high probability that
$
    \mathcal{U}^{\mathbf{q}}_T \leq 20 \cdot K_1 \cdot |M| = \wt{O}(1)$, 
which proves the statement.
\end{proof}

\subsection{Proof of Theorem~\ref{thm:incremental_dynamic_kde}}\label{sec:dynamic_kde_proof_theorem}

\begin{proof}We start with proving the first statement. Notice that   $\textsc{PreProcess}(X, \varepsilon)$   goes through $M\cdot K_1\cdot J_{\mu_i}\cdot K_{2,j}$ iterations, where  $M=O(\log(n_1))$, $K_1 = O(\varepsilon^{-2}\cdot \log(n_1))$, $ J_{\mu_i} = O(\log(n_1))$ and $K_{2,j} = O(\log(n_1) \cdot \mathrm{cost}(k))$.  Since   $k_j=\tilde{O}(1)$ by definition, the algorithm   concatenates $\tilde{O}(1)$ LSH functions.  By Lemma~\ref{lem:andoni-indyk},  
 the evaluation time of $H^*(\mathbf{x})$ for any $H^*\in \mathcal{H}^{k_j}$ is $n_1^{o(1)}$,  and hashing all $n_1$ points yields the running time of $\epsilon^{-2}\cdot n_1^{1+o(1)}\cdot \mathrm{cost}(k)$ for $\textsc{PreProcess}(X, \varepsilon)$ in the worst case. Since  we start with an empty set of query points $Q=\emptyset$, the running time of $\textsc{PreProcessQueryPoints}(X, Q, \varepsilon)$ can be omitted. This proves the first statement.

The guarantees for $\textsc{AddQueryPoint}(\mathbf{q})$ in the second statement follow from Lemma~\ref{lem:query-time-gen} and Lemma~\ref{lemma:variance}. The running time for $\textsc{DeleteQueryPoint}(\mathbf{q})$ follows from the running time guarantee for $\textsc{AddQueryPoint}(\mathbf{q})$.

Now we prove the third statement. The correctness of the updated estimate of $\mu_\mathbf{q}$ follows from Lemma~\ref{lemma:correctness_update_incremental}. To prove the time complexity, we notice that, when running $\textsc{AddDataPoint}(\mathbf{z})$, the procedure goes through $M\cdot K_1\cdot J_{\mu_i}\cdot K_{2,j}$ iterations, where $M=O(\log n)$, $K_1 = O(\log(n)/\varepsilon^2)$, $ J_{\mu_i} = O(\log n)$, and $K_{2,j} = O(\log(n) \cdot \mathrm{cost}(k))$.   In the worst case, $\mathbf{z}$ is sampled in every iteration on Line~\ref{alg:addupdate:line:samplepoint_start} of Algorithm~\ref{alg:addpointandupdatequeries}, and needs to be stored in the bucket $B_{H_{\mu_i, a, j, \ell}}(\mathbf{z})$. Therefore, the running time of updating all the buckets $B_{H_{\mu_i, a, j, \ell}}(\mathbf{z})$ for a new  $\mathbf{z}$  is at most $\epsilon^{-2}\cdot n^{o(1)}\cdot \mathrm{cost}(k)$. To analyse Lines~\ref{alg:addupdate:update_estimates_start}--\ref{alg:addupdate:update_estimates_end} of Algorithm~\ref{alg:addpointandupdatequeries}, we perform an amortised analysis. By Lemma~\ref{lemma:bound_number_updates}, it holds   with high probability that every $\mathbf{q} \in Q$ is updated $\wt{O}(1)$ times throughout the sequence of data point updates. When $\mathbf{q} \in Q$ is updated, the total running time for Lines~\ref{alg:addupdate:update_estimates_start}--\ref{alg:addupdate:update_estimates_end} is   \[\wt{O}(\epsilon^{-2} \cdot K_{2,j} + \epsilon^{-2} \cdot \mathrm{cost}(k)) = \wt{O}(\epsilon^{-2} \cdot \mathrm{cost}(k)),\] due to Lines~\ref{alg:addupdate:remove_datapoint_full_hash} and~\ref{alg:addupdate:re-estimate_query}. Let $T$ be the total number of query and data point insertions at any point throughout the sequence of updates, and   $m\triangleq|Q|$. Then the amortised update time is
  \[
  \wt{O} \lp \frac{m \cdot \epsilon^{-2} \cdot \mathrm{cost}(k)}{T} \rp = \wt{O} \lp \frac{m \cdot \epsilon^{-2} \cdot \mathrm{cost}(k)}{m} \rp  = \wt{O}\lp\epsilon^{-2} \cdot \mathrm{cost}(k) \rp,
  \]
 where the second inequality follows from  $T \geq m$, as the algorithm started with an empty query set $Q=\emptyset$.
 Combining everything together proves the running time.
\end{proof}

\section{Omitted Detail from Section~\ref{sec:dynamic_cps}} \label{sec:appendix_dyn_cps}

 This section presents all the detail omitted from Section~\ref{sec:dynamic_cps}, and is organised as follow: Section~\ref{sec:app_init_cps} presents and analyses the algorithm for the initialisation step;  Section~\ref{sec:app_dyn_cps} presents and analyses the algorithm for the dynamic update step.

\subsection{The Initialisation Step\label{sec:app_init_cps}}

 In this subsection we present the algorithms used in the initialisation step, and analyse its correctness as well as complexity. We first introduce the tree data structure (Algorithm~\ref{alg:tree_data_structure}) and the initialisation procedures for constructing an approximate similarity graph (Algorithm~\ref{alg:initialise_CPS}). The analysis then establishes the probability that a sampling path passes through any internal node of the constructed tree (Lemma~\ref{lem:nodeprob}), which is crucial for the subsequent correctness and time complexity analysis that follow \citet{macgregor2024fastkde} to a large extent.
 
 The following tree data structure will be used in the design of our procedures.

\begin{algorithm}[H]
  \caption{Tree Data Structure}
  \label{alg:tree_data_structure}
  \begin{algorithmic}[1]

    \STATE \textsc{Leaf}($\mathbf{x}_i$)
    \STATE \hspace{1em} \textbf{Input:} data point $\mathbf{x}_i$
    \STATE \hspace{1em} $\textsf{data} \gets \mathbf{x}_i$
      \COMMENT{Stores the data point}
    \STATE \hspace{1em} $\textsf{paths} \gets \textsc{Nil}$
      \COMMENT{Stores the sampling paths ending at this leaf}

    \STATE \textsc{Node}($X$)
    \STATE \hspace{1em} \textbf{Input:} set of data points $X$
    \STATE \hspace{1em} $\textsf{data} \gets X$
      \COMMENT{Stores the data points in the subtree rooted at this node}
    \STATE \hspace{1em} $\textsf{size} \gets |X|$
      \COMMENT{Number of data points in the subtree rooted at this node}
    \STATE \hspace{1em} $\textsf{kde} \gets \textsc{Nil}$
      \COMMENT{Stores the DynamicKDE structure}
    \STATE \hspace{1em} $\textsf{left} \gets \textsc{Nil}$
      \COMMENT{Left child node}
    \STATE \hspace{1em} $\textsf{right} \gets \textsc{Nil}$
      \COMMENT{Right child node}
    \STATE \hspace{1em} $\textsf{parent} \gets \textsc{Nil}$
      \COMMENT{Parent node}
    \STATE \hspace{1em} $\textsf{paths} \gets \textsc{Nil}$
      \COMMENT{Stores the sampling paths passing through this node}

  \end{algorithmic}
\end{algorithm}

 Based on this data structure, the main procedures used for constructing an approximate similarity graph in the initialisation step are presented in Algorithm~\ref{alg:initialise_CPS}. We remark that, for any set $X$ of data points,   we always set $\epsilon=1/\log^3|X|$ when running the  $\textsc{DyanmicKDE.Initialise}$ procedure. Choosing a fixed value of $\epsilon$ in this section allows us to  simplify  the presentation of the analysis without loss of generality. 

\begin{algorithm}
  \caption{Initialisation Procedures for Constructing an Approximate Similarity Graph}
  \label{alg:initialise_CPS}
  \begin{algorithmic}[1]

    \PROCEDURE{InitialiseTree($X$)}{}
    \STATE \textbf{Input:} set $X$ of data points
    \IF{$|X| = 1$}
      \STATE \textbf{return} $\textsc{Leaf}(X)$ 
        \COMMENT{Leaves store individual data points $\mathbf{x}_i \in X$}
    \ELSE
      \STATE $\tree \gets \textsc{Node}(X)$
      \STATE $m \gets 2^{\lfloor \log (|X| / 2) \rfloor}$, 
      $\hat{n} \gets |X|$ 
        \COMMENT{Nearest power of 2 less than or equal to $\hat{n}/2$}
      \STATE $X_L \gets X[1 : \hat{n} - m]$, 
      $X_R \gets X[\hat{n} - m +1 : \hat{n}]$ 
        \COMMENT{Split the dataset into two}
      \STATE $\tree_L \gets \textsc{InitialiseTree}(X_L)$, 
      $\tree_R \gets \textsc{InitialiseTree}(X_R)$
      \STATE $\tree.\textsf{left} \gets \tree_L$, 
      $\tree.\textsf{right} \gets \tree_R$
      \STATE $\tree_L.\textsf{kde} \gets \textsc{DynamicKDE.Initialise}\bigl(X_L, \emptyset, \tfrac{1}{\log^3 n}\bigr)$
        \COMMENT{(Algorithm~\ref{alg:dynamic_kde})}
      \STATE $\tree_R.\textsf{kde} \gets \textsc{DynamicKDE.Initialise}\bigl(X_R, \emptyset, \tfrac{1}{\log^3 n}\bigr)$
      \STATE \textbf{return} $\tree$
    \ENDIF
    \ENDPROCEDURE

    \STATE   

    \PROCEDURE{Sample($S, \tree, \ell$)}{}
    \STATE \textbf{Input:} set $S$ of points $\mathbf{x}_i$, \textsf{KDE} tree $\tree$ representing data points $X$, sample index $\ell$
    \STATE \textbf{Output:} $E = \{(\mathbf{x}_i, \mathbf{x}_j)\ \text{for some $i,j$}\}$
    \FOR{$\mathbf{x}_i \in S$}
      \STATE $\mathcal{P}_{\mathbf{x}_i, \ell} \gets \mathcal{P}_{\mathbf{x}_i, \ell} \cup \tree$
      \STATE $\tree.\textsf{paths} \gets \tree.\textsf{paths} \cup \{\mathcal{P}_{\mathbf{x}_i, \ell}\}$
        \COMMENT{Update and store sample paths}
    \ENDFOR
    \IF{$\textsc{IsLeaf}(\tree)$}
      \STATE \textbf{return} $S \times \tree.\textsf{data}$
    \ELSE
      \STATE $\tree_L \gets \tree.\textsf{left}, X_L \gets \tree_L.\textsf{data}$
      \STATE $\tree_R \gets \tree.\textsf{right},  X_R \gets \tree_R.\textsf{data}$
      \FOR{$\mathbf{x}_i \in S$}
        \STATE $\tree_L.\hat{\mu}_{\mathbf{x}_i} \gets \tree_L.\textsf{kde}.\textsc{AddQueryPoint}(\mathbf{x}_i)$ 
          \label{alg:initialise_cps:sample_addquery}
          \COMMENT{(Algorithm~\ref{alg:dynamic_kde})}
        \STATE $\tree_R.\hat{\mu}_{\mathbf{x}_i} \gets \tree_R.\textsf{kde}.\textsc{AddQueryPoint}(\mathbf{x}_i)$
      \ENDFOR
      \STATE $S_L \gets \emptyset, \quad S_R \gets \emptyset$
      \FOR{$\mathbf{x}_i \in S$}
        \STATE $r \sim \mathrm{Unif}[0,1]$
        \IF{$r \le \frac{\tree_L.\hat{\mu}_{\mathbf{x}_i}}
                        {\tree_L.\hat{\mu}_{\mathbf{x}_i} + \tree_R.\hat{\mu}_{\mathbf{x}_i}}$}
          \label{alg:initialise_cps:sample_prob}
          \STATE $S_L \gets S_L \cup \{\mathbf{x}_i\}$
        \ELSE
          \STATE $S_R \gets S_R \cup \{\mathbf{x}_i\}$
        \ENDIF
      \ENDFOR
      \STATE \textbf{return} $\textsc{Sample}(S_L, \tree_L, \ell) \;\cup\; \textsc{Sample}(S_R, \tree_R, \ell)$
    \ENDIF
    \ENDPROCEDURE

    \STATE   

    \PROCEDURE{ConstructGraph($X$)}{}
    \STATE \textbf{Input:} set of data points $X$
    \STATE $\tree \gets \textsc{InitialiseTree}(X)$
    \STATE $\tree.\textsf{kde} \gets \textsc{DynamicKDE.Initialise}\bigl(X, X, \tfrac{1}{\log^3 |X|}\bigr)$
      \COMMENT{(Algorithm~\ref{alg:dynamic_kde})}
    \STATE $E \gets \emptyset$
    \FOR{$\ell \in [L]$} 
      \label{alg:line:construct_sparsifier:start_sampling}
      \STATE $E_\ell \gets \textsc{Sample}(X, \tree, \ell)$
      \STATE $E \gets E \cup E_\ell$
      \FOR{$(\mathbf{x}_i, \mathbf{x}_j) \in E_\ell$}
        \STATE $\widehat{w}(i,j) \gets L \cdot k(\mathbf{x}_i, \mathbf{x}_j)\big/
          \min\{\tree.\textsf{kde}.\hat{\mu}_{\mathbf{x}_i}, \tree.\textsf{kde}.\hat{\mu}_{\mathbf{x}_j}\}$
        \IF{$\min\{\tree.\textsf{kde}.\hat{\mu}_{\mathbf{x}_i}, \tree.\textsf{kde}.\hat{\mu}_{\mathbf{x}_j}\}
              = \tree.\textsf{kde}.\hat{\mu}_{\mathbf{x}_j}$}
          \STATE $\mathcal{B}_{\mathbf{x}_j} \gets \mathcal{B}_{\mathbf{x}_j} \cup \{\mathbf{x}_i\}$
            \label{alg:line:construct_sparsifier:add_min_neighbor}
            \COMMENT{Track neighbors with higher degree}
        \ENDIF
        \STATE $w_G(\mathbf{x}_i, \mathbf{x}_j) \mathrel{+}= 
          \frac{k(\mathbf{x}_i, \mathbf{x}_j)}{\widehat{w}(i,j)}$
          \label{alg:line:construct_sparsifier:reweight_factor}
      \ENDFOR 
      \label{alg:line:construct_sparsifier:end_sampling}
    \ENDFOR

    \STATE \textbf{return} $\tree,\; G \triangleq (X, E, w_G)$
    \ENDPROCEDURE

  \end{algorithmic}
\end{algorithm}

 To analyse the correctness and time complexity of the algorithm, we first 
 prove that, for any data point $\mathbf{x}_i \in X$, the probability that its sampling path $\mathcal{P}_{\mathbf{x}_i, \ell}$ passes through any internal node $\tree'$ depends on the \textsf{KDE} value of $\mathbf{x}_i$ with respect to $\tree'.\textsf{data}$.

\begin{lemma} \label{lem:nodeprob}
For any point $\mathbf{x}_i \in X$, tree $\tree$ constructed by $\textsc{ConstructGraph}$ (Algorithm~\ref{alg:initialise_CPS}), and sampling path $\mathcal{P}_{\mathbf{x}_i, \ell}$ (for any $\ell \in [L]$), the probability that $\mathcal{P}_{\mathbf{x}_i, \ell}$  passes through any internal node $\tree'$ of $\tree$ is given by
\[
    \frac{k(\mathbf{x}_i, \tree'.\textsf{data})}{2 k(\mathbf{x}_i, \tree.\textsf{data})} \leq  \mathbb{P}\left[\mathcal{P}_{\mathbf{x}_i, \ell} \in \tree'.\textsf{paths}\right]  \leq \frac{2 k(\mathbf{x}_i, \tree'.\textsf{data})}{k(\mathbf{x}_i, \tree.\textsf{data})}.
\]
\end{lemma}
\begin{proof}
 Let
$X = \{\mathbf{x}_1, \ldots, \mathbf{x}_{n_1}\}$
be the input data points for the $\textsc{ConstructGraph}(X, \varepsilon)$ procedure in Algorithm~\ref{alg:initialise_CPS}.
Then, in each recursive call (at some internal root $\tree''$) to $\textsc{Sample}$ (Algorithm~\ref{alg:initialise_CPS}) we are given the data points $X_L \triangleq \tree''.\textsf{left}.\textsf{data}$ and $X_R \triangleq \tree''.\textsf{left}.\textsf{data}$
as input and assign $\mathcal{P}_{\mathbf{x}_{i}, \ell}$ to either $\tree_L'' \triangleq \tree''.\textsf{left}$ or $\tree_R'' \triangleq \tree''.\textsf{right}$. 
By Line~\ref{alg:initialise_cps:sample_prob} of Algorithm~\ref{alg:initialise_CPS}, we have that the probability of assigning $\mathcal{P}_{\mathbf{x}_{i}, \ell}$ to
$\tree_L''.\textsf{paths}$ is
\[
    \mathbb{P}\left[ \mathcal{P}_{\mathbf{x}_{i}, \ell} \in \tree_L''.\textsf{paths} ~|~ \mathcal{P}_{\mathbf{x}_{i}, \ell} \in \tree''.\textsf{paths} \right] = \frac{\tree_L''.\hat{\mu}_{\mathbf{x}_i}}{\tree''.\hat{\mu}_{\mathbf{x}_i}}.
\]
By  the performance guarantee of the \textsf{KDE} algorithm (Theorem~\ref{thm:incremental_dynamic_kde}), we have that $\tree''.\hat{\mu}_{\mathbf{x}_i} \in (1 \pm \epsilon) \cdot k(\tree''.\textsf{data}, \mathbf{x}_i).$
This gives 
\begin{align} 
    \left(\frac{1 - \epsilon}{1 + \epsilon}\right) \frac{k(\mathbf{x}_i, \tree''_L.\textsf{data})}{k(\mathbf{x}_i, \tree''.\textsf{data})} & \leq \mathbb{P}\left[ \mathcal{P}_{\mathbf{x}_{i}, \ell} \in \tree_L''.\textsf{paths} ~|~ \mathcal{P}_{\mathbf{x}_{i}, \ell} \in \tree''.\textsf{paths} \right] \nonumber\\
    & \leq \left(\frac{1 + \epsilon}{1 - \epsilon}\right) \frac{k(\mathbf{x}_i, \tree''_L.\textsf{data})}{k(\mathbf{x}_i, \tree''.\textsf{data})}. \label{eq:oneprob}
\end{align}

Next, notice that it holds for  a sequence of internal nodes $\tree_1, \tree_2, \ldots, \tree_r$  with $\tree_i.\textsf{parent} = \tree_{i+1}$~($1 \leq i \leq r-1$) that
\[
\mathbb{P}\left[ \mathcal{P}_{\mathbf{x}_i, \ell} \in \tree_1.\textsf{paths}\right] = \prod_{1\leq j\leq r-1} \mathbb{P}\left[ \mathcal{P}_{\mathbf{x}_i, \ell} \in \tree_j.\textsf{paths} | \mathcal{P}_{\mathbf{x}_i, \ell} \in \tree_{j+1}.\textsf{paths} \right], 
\]
where each term in the right  hand side above corresponds to one level of recursion of the $\textsc{Sample}$ procedure in Algorithm~\ref{alg:initialise_CPS} and there are at most $\lceil \log_2(n_1) \rceil$ terms.
Then, by  setting  $\tree_r = \tree$, $\tree_1 = \tree'$, \eqref{eq:oneprob}, and the fact that the denominator and numerator of adjacent terms cancel out, we have
\begin{align*}
    \left(\frac{1 - \epsilon}{1 + \epsilon} \right)^{\lceil \log(n_1) \rceil} \frac{k(\mathbf{x}_i, \tree'.\textsf{data})}{k(\mathbf{x}_i, \tree.\textsf{data})}
    & \leq \mathbb{P}\left[\mathcal{P}_{\mathbf{x}_i, \ell} \in \tree'.\textsf{paths} \right]\\
    & \leq \left(\frac{1 + \epsilon}{1 - \epsilon}\right)^{\lceil \log(n_1)\rceil} \frac{k(\mathbf{x}_i, \tree'.\textsf{data})}{k(\mathbf{x}_i, \tree.\textsf{data})}.
\end{align*}

 For the lower bound, we have that
\begin{align*}
    \left(\frac{1 - \epsilon}{1 + \epsilon}\right)^{\lceil \log(n_1)\rceil} & \geq \left(1 - 2 \epsilon\right)^{\lceil \log(n_1)\rceil}  \geq 1 - 3 \epsilon \log(n_1)    \geq 1/2,
\end{align*}
where the final inequality follows by the condition of $\epsilon$ that $\epsilon \leq 1 / \log^3(n_1)$.

 For the upper bound, we similarly have
\begin{align*}
   \left(\frac{1 + \epsilon}{1 - \epsilon}\right)^{\lceil \log(n_1)\rceil} & \leq  \left(1 + 3 \epsilon\right)^{\lceil \log(n_1)\rceil}  \leq \exp \left( 3 \epsilon \lceil \log(n_1)\rceil  \right)  \leq \mathrm{e}^{2/3} \leq 2,
\end{align*}
where the first inequality follows since $\epsilon \leq 1 / \log^3(n_1)$. 
\end{proof}

  The remaining part of our analysis is very  similar to the proof presented  in~\cite{macgregor2024fastkde}. 

For each added edge,   $\textsc{ConstructGraph}(X)$ computes the   estimate defined by 
\[
    \widehat{w}(i,j) \triangleq 6 C\cdot  \frac{ \log n_1}{\lambda_{k+1}}\cdot \frac{k(\mathbf{x}_i, \mathbf{x}_j)}{\min\{\hat{\mu}_{\mathbf{x}_i}, \hat{\mu}_{\mathbf{x}_j}\}},
\]
where for the ease of notation we denote $\tree.\textsf{kde}.\hat{\mu}_{\mathbf{x}_i} \triangleq \hat{\mu}_{\mathbf{x}_i}$ and $\tree.\textsf{kde}.\hat{\mu}_{\mathbf{x}_j} \triangleq \hat{\mu}_{\mathbf{x}_j}$. If an edge $(\mathbf{x}_i, \mathbf{x}_j)$ is sampled, then the edge is included with weight $k(\mathbf{x}_i, \mathbf{x}_j) / \widehat{w}(i,j)$. The algorithm in \cite{macgregor2024fastkde} is almost the same  as   our algorithm,  with the only difference that in their case  the edge is included with weight $k(\mathbf{x}_i, \mathbf{x}_j) / \widehat{p}(i,j)$, where
\[
\widehat{p}(i,j) \triangleq 6 C\cdot  \frac{ k(\mathbf{x}_i, \mathbf{x}_j) \cdot \log n_1}{\lambda_{k+1}}\cdot \lp\frac{1}{\hat{\mu}_{\mathbf{x}_i}} + \frac{1}{\hat{\mu}_{\mathbf{x}_j}}\rp - \lp 6 C\cdot  \frac{ k(\mathbf{x}_i, \mathbf{x}_j) \cdot \log n_1}{\lambda_{k+1}} \rp^2 \cdot \frac{1}{\hat{\mu}_{\mathbf{x}_i} \cdot \hat{\mu}_{\mathbf{x}_j}}. 
\]
 Notice that
\[
\widehat{p}(i,j) \leq 6 C\cdot  \frac{ k(\mathbf{x}_i, \mathbf{x}_j) \cdot \log n_1}{\lambda_{k+1}}\cdot \lp\frac{1}{\hat{\mu}_{\mathbf{x}_i}} + \frac{1}{\hat{\mu}_{\mathbf{x}_j}}\rp \leq 2 \widehat{w}(i,j).
\]
Assuming without loss of generality that \[6 C\cdot  \frac{ k(\mathbf{x}_i, \mathbf{x}_j) \cdot \log n_1}{\lambda_{k+1}}\cdot \left(\frac{1}{\hat{\mu}_{\mathbf{x}_i}}+\frac{1}{\hat{\mu}_{\mathbf{x}_j}}\right) < 1,
\] we have that
\[
\widehat{p}(i,j) \geq 3 C\cdot  \frac{ k(\mathbf{x}_i, \mathbf{x}_j) \cdot \log n_1}{\lambda_{k+1}}\cdot \lp\frac{1}{\hat{\mu}_{\mathbf{x}_i}} + \frac{1}{\hat{\mu}_{\mathbf{x}_j}}\rp \geq \frac{\widehat{w}(i,j)}{2}.
\]
As such, the scaling factor $\widehat{w}(i,j)$ in our algorithm and  $\widehat{p}(i,j)$ in the algorithm of~\cite{macgregor2024fastkde} are within a constant factor of each other.  Therefore, to prove the first statement of Theorem~\ref{thm:dynamic_cps}, one can  follow the proof of Theorem~2 of~\cite{macgregor2024fastkde}, while replacing each $\widehat{p}(i,j)$ with $\widehat{w}(i,j)$ appropriately.

\subsection{The Dynamic Update Step\label{sec:app_dyn_cps}}

 In this subsection we present the algorithm used in the dynamic update step, and analyse its correctness as well as complexity.  Our main algorithm for dynamically updating an approximate similarity graph is described in Algorithm~\ref{alg:update_CPS}, and the $\textsc{Resample}$ procedure can be found in Algorithm~\ref{alg:update_CPS_tree}.

Our analysis for the dynamic update step  first covers the running time of the $\textsc{Sample}( \{\mathbf{x}\}, \tree, \ell )$  and $\textsc{Resample}(\tree, \mathcal{P}_{\mathbf{x}, \ell})$  procedures (Lemma~\ref{lem:running_time_sample_procedures}). It then proves a sequence of lemmas bounding   the total number of sampling paths that need to be resampled (Lemmas~\ref{lem:dynamic_cps_tree_update_total_paths}, \ref{lemma:changed_neighbors}, and \ref{lem:update_time_dynamic_CPS}). Finally, the correctness follows by the fact that   the update step is equivalent to a full reconstruction of the tree after an update (Lemma~\ref{lemma:correctnes_update_CPS}).

\begin{algorithm}[H]
\caption{Dynamic Update Algorithm for Constructing an Approximate Similarity Graph}
\label{alg:update_CPS}
\begin{algorithmic}[1]

\PROCEDURE{UpdateGraph($G=(X, E, w_G), \tree, \mathbf{z}$)}{}
\STATE \textbf{Input}: an approximate similarity graph $G$, \textsf{KDE} tree $\tree$, new data point $\mathbf{z}$

\STATE $\mathcal{A} \gets \textsc{AddDataPointTree}(\tree, \mathbf{z})$ 
\COMMENT{(Algorithm~\ref{alg:update_CPS_tree})} 
\label{alg:update_CPS:paths_to_resample}
\STATE $E_{\mathrm{new}} \gets \emptyset$

\FOR[Sample $L$ neighbours from the new vertex $\mathbf{z}$]{$\ell=1,2,\ldots,L$}  \label{alg:update_CPS:start_sample_new_vertex}
    \STATE $(\mathbf{z}, \mathbf{x}_j) \gets \textsc{Sample}(\{\mathbf{z}\}, \tree, \ell)$
    \STATE $E_{\mathrm{new}} \gets E_{\mathrm{new}} \cup \{(\mathbf{z}, \mathbf{x}_j)\}$
    \STATE $\widehat{w}(i,j) \gets 
      \dfrac{L \cdot k(\mathbf{z}, \mathbf{x}_j)}
      {\min\{\tree.\textsf{kde}.\hat{\mu}_{\mathbf{z}}, \tree.\textsf{kde}.\hat{\mu}_{\mathbf{x}_j}\}}$
    \STATE $w_G(\mathbf{z}, \mathbf{x}_j) \mathrel{+}= 
      \dfrac{k(\mathbf{z}, \mathbf{x}_j)}{\widehat{w}(i,j)}$
    \IF{$\min\{\tree.\textsf{kde}.\hat{\mu}_{\mathbf{z}}, \tree.\textsf{kde}.\hat{\mu}_{\mathbf{x}_j}\} 
        = \tree.\textsf{kde}.\hat{\mu}_{\mathbf{x}_j}$}
        \STATE $\mathcal{B}_{\mathbf{x}_j} \gets \mathcal{B}_{\mathbf{x}_j} \cup \{\mathbf{z}\}$ \label{alg:update_CPS:end_sample_new_vertex}
    \ENDIF
\ENDFOR 

\STATE $E \gets E \cup E_{\mathrm{new}}$

\FOR{$\mathbf{x}_i$ such that $\tree.\textsf{kde}.\hat{\mu}_{\mathbf{x}_i}$ has changed}
\label{alg:update_CPS:start_reweight_neighbors}
    \STATE Let $\deg_{\mathrm{old}}$ be the old estimate of $\tree.\textsf{kde}.\hat{\mu}_{\mathbf{x}_i}$
    \FOR{$\mathbf{x}_j \in \mathcal{B}_{\mathbf{x}_i}$}
        \STATE 
        $w_G(\mathbf{x}_i, \mathbf{x}_j) \gets w_G(\mathbf{x}_i, \mathbf{x}_j) 
        \cdot \dfrac{\deg_{\mathrm{old}}}{\min\{\tree.\textsf{kde}.\hat{\mu}_{\mathbf{x}_i}, 
        \tree.\textsf{kde}.\hat{\mu}_{\mathbf{x}_j}\}}$
        \COMMENT{Update scaling factor of adjacent edges}
        \IF{$\min\{\tree.\textsf{kde}.\hat{\mu}_{\mathbf{x}_i}, \tree.\textsf{kde}.\hat{\mu}_{\mathbf{x}_j}\} 
            = \tree.\textsf{kde}.\hat{\mu}_{\mathbf{x}_j}$}
            \STATE $\mathcal{B}_{\mathbf{x}_i} \gets \mathcal{B}_{\mathbf{x}_i} \setminus \{\mathbf{x}_j\}$ \label{alg:update_CPS:end_reweight_neighbors}
        \ENDIF
    \ENDFOR
\ENDFOR 

\FOR{$\mathcal{P}_{\mathbf{x}_i, \ell} \in \mathcal{A}$}
\label{alg:update_CPS:start_resample_edges}
    \STATE Let $\tree'$ be the parent of the highest internal node where 
    $\mathcal{P}_{\mathbf{x}_i, \ell}$ was fetched
    \FOR{$\tree^* \text{ below } \tree' \text{ such that } 
          \mathcal{P}_{\mathbf{x}_i, \ell} \in \tree^*.\textsf{paths}$}
    \label{alg:update_CPS:start_remove_edges}
        \STATE Remove $\mathcal{P}_{\mathbf{x}_i, \ell}$ from $\tree^*.\textsf{paths}$
        \STATE Remove $\mathbf{x}_i$ from the query set of 
               $\tree^*.\textsf{left}.\textsf{kde}$ and $\tree^*.\textsf{right}.\textsf{kde}$
        \STATE Remove $\tree^*$ from $\mathcal{P}_{\mathbf{x}_i, \ell}$
    \ENDFOR
    
    \STATE Let $\mathbf{x}_j$ be the previous sampled neighbour of $\mathbf{x}_i$ 
           (i.e., leaf in $\mathcal{P}_{\mathbf{x}_i, \ell}$)
    \STATE $w_G(\mathbf{x}_i, \mathbf{x}_j) \mathrel{-}= 
           \dfrac{k(\mathbf{x}_i, \mathbf{x}_j)}{\widehat{w}(i,j)}$
           \COMMENT{where $\widehat{w}(i,j)$ is the previous re-scaling factor}
    \IF{$w_G(\mathbf{x}_i, \mathbf{x}_j) = 0$}
        \STATE $E \gets E \setminus \{(\mathbf{x}_i, \mathbf{x}_j)\}$ \label{alg:update_CPS:end_remove_edges}
    \ENDIF 

    \STATE $(\mathbf{x}_i, \mathbf{x}^*_j) \gets 
           \textsc{Resample}(\tree', \mathcal{P}_{\mathbf{x}_i, \ell}, \varepsilon)$
    \COMMENT{Resample path (Algorithm~\ref{alg:update_CPS_tree})}
    \label{alg:update_CPS:start_resample_edge}
    
    \IF{$(\mathbf{x}_i, \mathbf{x}^*_j) \notin E$}
        \STATE $E \gets E \cup \{(\mathbf{x}_i, \mathbf{x}^*_j)\}$
    \ENDIF
    
    \STATE $\widehat{w}^*(i,j) \gets 
      \dfrac{L \cdot k(\mathbf{x}_i, \mathbf{x}^*_j)}
      {\min\{\tree.\textsf{kde}.\hat{\mu}_{\mathbf{x}_i}, \tree.\textsf{kde}.\hat{\mu}_{\mathbf{x}^*_j}\}}$
    
    \IF{$\min\{\tree.\textsf{kde}.\hat{\mu}_{\mathbf{x}_i}, \tree.\textsf{kde}.\hat{\mu}_{\mathbf{x}^*_j}\} 
        = \tree.\textsf{kde}.\hat{\mu}_{\mathbf{x}^*_j}$}
        \STATE $\mathcal{B}_{\mathbf{x}^*_j} \gets \mathcal{B}_{\mathbf{x}^*_j} \cup \{\mathbf{x}_i\}$
        \COMMENT{Update neighbours with higher degrees}
    \ENDIF
    
    \STATE $w_G(\mathbf{x}_i, \mathbf{x}_j) \mathrel{+}= 
           \dfrac{k(\mathbf{x}_i, \mathbf{x}_j)}{\widehat{w}^*(i,j)}$
\label{alg:update_CPS:end_resample_edges}
\ENDFOR

\ENDPROCEDURE

\end{algorithmic}
\end{algorithm}

\begin{algorithm}[H]
\caption{Tree Update Procedures for Constructing an Approximate Similarity Graph}
\label{alg:update_CPS_tree}
\begin{algorithmic}[1]

\PROCEDURE{AddDataPointTree($\tree, \mathbf{z}$)}{}
  \STATE \textbf{Input}: \textsf{KDE} tree/node $\tree$, new data point $\mathbf{z}$
  \IF{$\textsc{IsLeaf}(\tree)$}
    \STATE $\mathbf{x} \gets \tree.\textsf{data}$
    \STATE $\mathcal{A} \gets \tree.\textsf{parent}.\textsf{paths}$ 
    \COMMENT{Store paths that need to be resampled}
    \STATE $\tree_{\mathrm{new}} \gets \textsc{Node}(\{\mathbf{x}, \mathbf{z}\})$
    \STATE $\tree_{\mathrm{new}}.\textsf{left} \gets \textsc{Leaf}(\mathbf{x})$
    \STATE $\tree_{\mathrm{new}}.\textsf{right} \gets \textsc{Leaf}(\mathbf{z})$
    \STATE $\tree_{\mathrm{new}}.\textsf{kde} \gets \textsc{DynamicKDE.Initialise}
                  (\{\mathbf{x}, \mathbf{z}\}, \emptyset, \varepsilon)$ 
    \COMMENT{Initialise new \textsf{KDE} data structure}
    \STATE Replace the leaf $\tree$ with node $\tree_{\mathrm{new}}$
    \STATE \textbf{return} $\mathcal{A}$
  \ELSE
    \STATE $\tree.\textsf{kde}.\textsc{AddDataPoint}(\mathbf{z})$ 
           \COMMENT{(Algorithm~\ref{alg:dynamic_kde})}
           \label{alg:update_CPS_tree:update_kde}
    \STATE Let $\tilde{A}$ be the set of points $\mathbf{x}_i \in \tree.\textsf{kde}.Q$ such that 
           $\tree.\textsf{kde}.\hat{\mu}_{\mathbf{x}_i}$ changes after adding $\mathbf{z}$.
    \STATE $\mathcal{A} \gets \{\mathcal{P}_{\mathbf{x}_i, \ell} \in \tree.\textsf{parent}.\textsf{paths} 
                    \mid \mathbf{x}_i \in \tilde{A}\}$
    \IF{$\tree.\textsf{left}.\textsf{size} \leq \tree.\textsf{right}.\textsf{size}$}
    \label{alg:update_CPS_tree:add_in_tree}
      \STATE \textbf{return} 
             $\mathcal{A} \cup \textsc{AddDataPointTree}(\tree.\textsf{left}, \mathbf{z})$
    \ELSE
      \STATE \textbf{return} 
             $\mathcal{A} \cup \textsc{AddDataPointTree}(\tree.\textsf{right}, \mathbf{z})$
    \ENDIF
  \ENDIF
\ENDPROCEDURE
\STATE
\PROCEDURE{Resample($\tree, \mathcal{P}_{\mathbf{x}_i, \ell}$)}{}
  \STATE \textbf{Input}: \textsf{KDE} tree/node $\tree$, and sampling path $\mathcal{P}_{\mathbf{x}_i, \ell}$
  \STATE $\mathcal{P}_{\mathbf{x}_i, \ell} \gets \mathcal{P}_{\mathbf{x}_i, \ell} \cup \{\tree\}$
  \STATE $\mathcal{T}.\textsf{paths} \gets \mathcal{T}.\textsf{paths} 
                      \cup \{\mathcal{P}_{\mathbf{x}_i, \ell}\}$
  \COMMENT{Update and store sample paths}
  \IF{$\textsc{IsLeaf}(\tree)$}
    \STATE \textbf{return} $\mathbf{x}_i \times \tree.\textsf{data}$
  \ELSE
    \STATE $\tree_L \gets \tree.\textsf{left}, \quad X_L \gets \tree.\textsf{left}.\textsf{data}$
    \STATE $\tree_R \gets \tree.\textsf{right}, \quad X_R \gets \tree.\textsf{right}.\textsf{data}$
    \STATE $\tree_L.\textsf{kde}.\textsc{AddQueryPoint}(\mathbf{x}_i)$ 
           \textbf{if} $\mathbf{x}_i \notin \tree_L.\textsf{kde}.Q$ 
           \COMMENT{(Algorithm~\ref{alg:dynamic_kde})}
           \label{alg:update_CPS_tree:resample_addquery}
    \STATE $\tree_R.\textsf{kde}.\textsc{AddQueryPoint}(\mathbf{x}_i)$ 
           \textbf{if} $\mathbf{x}_i \notin \tree_R.\textsf{kde}.Q$
    \STATE $r \sim \mathrm{Unif}[0, 1]$
    \IF{$r \leq 
         \dfrac{\tree_L.\hat{\mu}_{\mathbf{x}_i}}{\tree_L.\hat{\mu}_{\mathbf{x}_i} 
               + \tree_R.\hat{\mu}_{\mathbf{x}_i}}$}
      \STATE \textbf{return} $\textsc{Resample}(\tree_L, \mathcal{P}_{\mathbf{x}_i, \ell})$
    \ELSE
      \STATE \textbf{return} $\textsc{Resample}(\tree_R, \mathcal{P}_{\mathbf{x}_i, \ell})$
    \ENDIF
  \ENDIF
\ENDPROCEDURE

\end{algorithmic}
\end{algorithm}

\subsubsection{Running Time Analysis}
 
 \begin{proof}[Proof of Lemma~\ref{lem:running_time_sample_procedures}]
 The running time of the two  procedures is dominated by the recursive calls to $\textsc{AddQueryPoint}(\mathbf{x})$. By Theorem~\ref{thm:incremental_dynamic_kde}, the running time of adding a query point is $\epsilon^{-2}\cdot n^{o(1)}\cdot \mathrm{cost}(k)$.
 Since the depth of the tree $\tree$ is at most $\lceil \log n\rceil$,  there are at most $\lceil \log n\rceil$ recursive calls to $\textsc{Sample}$ and $\textsc{Resample}$. Hence, the total running time of $\textsc{Sample}$ and $\textsc{Resample}$ is $\epsilon^{-2}\cdot n^{o(1)}\cdot \mathrm{cost}(k)$.
\end{proof}

\begin{proof}[Proof of Lemma~\ref{lem:prob_go_to_subtree}] 
    By the tree construction, we have that 
    $
    \mathbb{P}[\mathbf{q} \in \tree'.\textsf{kde}.Q] = \mathbb{P}[\exists \ell' \text{ such that }\mathcal{P}_{\mathbf{q}, \ell'} \in \tree'.\textsf{paths}]$,
    and 
    \begin{align*}
        \mathbb{P}[\exists \ell' \text{ such that }\mathcal{P}_{\mathbf{q}, \ell'} \in \tree'.\textsf{paths}] &\leq L \cdot \mathbb{P}[\mathcal{P}_{\mathbf{q}, \ell} \in \tree'.\textsf{paths}] \leq L \cdot \frac{2 k(\mathbf{q}, \tree'.\textsf{data})}{k(\mathbf{q}, \tree.\textsf{data})} \\
        &= L \cdot \frac{2 k(\mathbf{q}, \tree'.\textsf{kde}.X)}{k(\mathbf{q}, \tree.\textsf{kde}.X)}  = \wt{O}\lp\frac{\mu_j}{\mu_i}\rp,
    \end{align*}
    where the first inequality holds by the union bound, the second inequality follows by Lemma~\ref{lem:nodeprob}, and the last line holds by the definition of $Q_{\mu_i \rightarrow \mu_j}(\tree')$ and   $L = \wt{O}(1)$.
\end{proof}

Next, we state Lemma~\ref{lem:expected_collisions_subtree_informal} more precisely and provide its proof.

\begin{lemma}\label{lem:expected_collisions_subtree}
     Let $\mathbf{z}$ be the data point that is added to   $\tree$ through the $\textsc{AddDataPointTree}(\tree, \mathbf{z})$ procedure in Algorithm~\ref{alg:update_CPS_tree}, and   $\tree'$ be any internal node that lies on the path from the new leaf $\textsc{Leaf}(\mathbf{z})$ to the root of $\tree$. Then it holds for any  $i \in \left[\lceil\log(2\cdot \tree.\textsf{kde}.n') \rceil\right]$, $a \in \tree'.\textsf{kde}.K_1$, $j \in [J_{\mu_i}]$ and $\ell$  that
    $$\mathbb{E}_{H_{\mu_i, a, j, \ell}}\left[\left|\{\mathbf{q} \in \tree'.\textsf{kde}.Q_{\mu_i} \mid \tree'.\textsf{kde}.H_{\mu_i, a, j, \ell}(\mathbf{z}) = \tree'.\textsf{kde}.H_{\mu_i, a, j, \ell}(\mathbf{q}) \} \right|\right] = \wt{O}\lp \mu_{i} \cdot 2^{j+1} \rp.$$
\end{lemma}
\begin{proof}
     We first remark that, except for the dynamic \textsf{KDE} structure stored at the root $\tree.\textsf{kde}$, it does not necessarily hold that $ \tree'.\textsf{kde}.Q = \tree'.\textsf{kde}.X$; this is because that the query points stored at internal nodes are the ones whose sample paths passed through this node, and the data points are the leaves of the subtree $\tree'$. Hence, to analyse the expected number of colliding points in the bucket $\tree'.\textsf{kde}.B_{H_{\mu_i, a, j, \ell}}(\mathbf{z})$, we need to  separately analyse the contributions from $\mathbf{q} \in Q_{\mu_{i'} \rightarrow \mu_i}(\tree')$ for $i' \geq i$.
  To achieve this, we apply  Lemma~\ref{lemma:number_in_query_bucket} and have for $i' \geq i$ that
\begin{equation}\label{eq:bound_full_hash_all_data}
        \mathbb{E}_{H_{\mu_i, a, j, \ell}}\left[\left|\left\{\mathbf{q} \in \tree.\textsf{kde}.Q_{\mu_{i'}} \mid \tree'.\textsf{kde}.H_{\mu_i, a, j, \ell}(\mathbf{z}) = \tree'.\textsf{kde}.H_{\mu_i, a, j, \ell}(\mathbf{q})\right\}\right|\right] = O(2^{j+1} \cdot \mu_{i'}).
    \end{equation}
   Therefore, it holds that 
    \begin{align}
      &\mathbb{E}_{H_{\mu_i, a, j, \ell}}\left[\left|\{  \mathbf{q} \in \tree'.\textsf{kde}.Q_{\mu_i} \mid \tree'.\textsf{kde}.H_{\mu_i, a, j, \ell}(\mathbf{z}) = \tree'.\textsf{kde}.H_{\mu_i, a, j, \ell}(\mathbf{q})\} \right|\right] \nonumber \\
      &  = \mathbb{E}_{H_{\mu_i, a, j, \ell}} \nonumber \\
     & \quad \left[ \sum_{i' \geq i} \left|\{  \mathbf{q} \in \tree'.\textsf{kde}.Q_{\mu_i} 
 \mid \tree'.\textsf{kde}.H_{\mu_i, a, j, \ell}(\mathbf{z}) = \tree'.\textsf{kde}.H_{\mu_i, a, j, \ell}(\mathbf{q})  \text{ and } \mathbf{q} \in Q_{\mu_{i'} \rightarrow \mu_i}(\tree')\} \right| \right] \nonumber \\
     &  = \sum_{i' \geq i} \mathbb{E}_{H_{\mu_i, a, j, \ell}}\nonumber \\
     & \qquad \left[ \left|\{ \mathbf{q} \in \tree'.\textsf{kde}.Q_{\mu_i} 
 \mid \tree'.\textsf{kde}.H_{\mu_i, a, j, \ell}(\mathbf{z}) = \tree'.\textsf{kde}.H_{\mu_i, a, j, \ell}(\mathbf{q}) \text{ and }  \mathbf{q} \in Q_{\mu_{i'} \rightarrow \mu_i}(\tree')  \} \right| \right] \nonumber\\
     &   =  \sum_{i' \geq i} \wt{O}\lp \frac{\mu_{i}}{\mu_{i'}} \rp\cdot \mathbb{E}_{H_{\mu_i, a, j, \ell}}\left[ \left|\{   \mathbf{q} \in \tree.\textsf{kde}.Q_{\mu_{i'}} \mid \tree'.\textsf{kde}.H_{\mu_i, a, j, \ell}(\mathbf{z}) = \tree'.\textsf{kde}.H_{\mu_i, a, j, \ell}(\mathbf{q}) \} \right| \right] \label{eq:subtree_bucket_bound:line1} \\
      &   =  \sum_{i' \geq i} \wt{O}\lp \frac{\mu_{i}}{\mu_{i'}} \cdot 2^{j+1} \mu_{i'} \rp \label{eq:subtree_bucket_bound:line2} \\
        &   =   \wt{O}\lp \mu_{i} \cdot 2^{j+1} \rp, 
    \end{align}
    where \eqref{eq:subtree_bucket_bound:line1} follows by  Lemma~\ref{lem:prob_go_to_subtree}, and~\eqref{eq:subtree_bucket_bound:line2} holds by \eqref{eq:bound_full_hash_all_data}. 
\end{proof}

\begin{lemma}\label{lem:dynamic_cps_tree_update_total_paths} 
     The expected total running time for  $\tree'.\textsf{kde}.\textsc{AddDataPoint}(\mathbf{z})$ (Line~\ref{alg:update_CPS_tree:update_kde} of Algorithm~\ref{alg:update_CPS_tree}) over all internal nodes $\tree'$ along the path  from the new leaf $\textsc{Leaf}(\mathbf{z})$ to the root of $\tree$ is
    $
      n^{o(1)}\cdot \mathrm{cost}(k)$.
    Moreover, the expected number of paths $\mathcal{A}$ (Line~\ref{alg:update_CPS:paths_to_resample}, Algorithm~\ref{alg:update_CPS}) that need to be resampled satisfies that 
    $\mathbb{E}[|\mathcal{A}|] = \wt{O}(1)$.
\end{lemma}

\begin{proof}
     We first study the update time and the total number of paths that need to be updated at a single internal node $\tree'$. Notice that,  when $\tree'.\textsf{kde}.\textsc{AddDataPoint}(\mathbf{z})$ (Line~\ref{alg:update_CPS_tree:update_kde} of Algorithm~\ref{alg:update_CPS_tree}) is called, the procedure $\textsc{AddPointAndUpdateQueries}$ in Algorithm~\ref{alg:addpointandupdatequeries} is called in the dynamic \textsf{KDE} data structure $\tree'.\textsf{kde}$. Hence, we analyse the expected running time of $\textsc{AddPointAndUpdateQueries}$ in Algorithm~\ref{alg:addpointandupdatequeries}.

 First, we have that executing Lines~\ref{alg:addupdate:line:samplepoint_start}--\ref{alg:addupdate:line:samplepoint_end} of in Algorithm~\ref{alg:addpointandupdatequeries} takes
    $
     K_{2,j} \cdot |\tree'.\textsf{kde}.B^*_{H_{\mu_i,a,j,\ell}}(\mathbf{z})| \cdot {n'}^{o(1)}$
    time, where $n' = |\tree'.\textsf{data}|$ is the number of data points stored at $\tree'$, and these five lines are executed with probability at most $1/({2^{j+1} \mu_i})$.
    Since we   only consider the collisions with points in $\tree'.\textsf{kde}.Q_{\mu_i}$, it holds by Lemma~\ref{lem:expected_collisions_subtree} that \[\mathbb{E}_{H_{\mu_i,a,j,\ell}}\left[|\tree'.\textsf{kde}.B^*_{H_{\mu_i,a,j,\ell}}(\mathbf{z})|\right] = \widetilde{O}\left(2^{j+1} \mu_i  \right).\] Hence, by our choice of $K_{2,j} = O(\log(n')\cdot \mathrm{cost}(k))$,   the expected total running time over all $\mu_i$, $a$, and $j$ of Lines~\ref{alg:addupdate:line:samplepoint_start}--\ref{alg:addupdate:line:samplepoint_end} of Algorithm~\ref{alg:addpointandupdatequeries} is $ \varepsilon^{-2} \cdot \mathrm{cost}(k) \cdot n'^{o(1)} $.  The same analysis can also be applied for Lines~\ref{alg:addupdate:outside_geometric_level_start}--\ref{alg:addupdate:q_recover_outside_geometriclevel} of Algorithm~\ref{alg:addpointandupdatequeries}. Moreover, the expected number of recovered points in $S$  (Line~\ref{alg:addupdate:recovered_points} of Algorithm~\ref{alg:addpointandupdatequeries}) is $\wt{O}(1)$, as the expected number of collisions we consider is \[\mathbb{E}_{H_{\mu_i,a,j,\ell}}\left[|\tree'.\textsf{kde}.B^*_{H_{\mu_i,a,j,\ell}}(\mathbf{z})|\right] = \widetilde{O}(2^{j+1} \mu_i),\] and these points are only considered with probability at most $1/({2^{j+1} \mu_i})$.
    
    Next, we analyse the running time of Lines~\ref{alg:addupdate:update_estimates_start}--\ref{alg:addupdate:update_estimates_end} of Algorithm~\ref{alg:addpointandupdatequeries}. For every $\mathbf{q} \in S$, the total running time for Lines~\ref{alg:addupdate:update_estimates_start}--\ref{alg:addupdate:update_estimates_end} is $\wt{O}(\varepsilon^{-2} \cdot K_{2,j} + \epsilon^{-2} \cdot \mathrm{cost}(k)) = \wt{O}(\epsilon^{-2} \cdot \mathrm{cost}(k))$, due to Lines~\ref{alg:addupdate:remove_datapoint_full_hash} and~\ref{alg:addupdate:re-estimate_query}.
    
  Hence, the expected total running time for running $\tree'.\textsf{kde}.\textsc{AddDataPoint}(\mathbf{z})$ at a single $\tree'$ is $\wt{O}\left(\epsilon^{-2}\cdot n'^{o(1)}\cdot \mathrm{cost}(k)\right)$. As there are at most $\lceil \log n\rceil$  nodes $\tree'$ that are updated when $\mathbf{z}$ is added and   $n' \leq n$, the running time guarantee of the lemma follows. 

  It remains to prove that $\mathbb{E}[|\mathcal{A}|] = \wt{O}(1)$.  Notice that, the number of points $\mathbf{q} \in \tree'.\textsf{kde}.Q$ whose \textsf{KDE} estimate changes at $\tree'$ is the number of recovered points in $S$ (Line~\ref{alg:addupdate:recovered_points} of Algorithm~\ref{alg:addpointandupdatequeries}). From the \textsc{AddPointAndUpdateQueries}  procedure (Algorithm~\ref{alg:addpointandupdatequeries}), it holds for every $\mu_i$ and $a$   that $\mathbb{E}[|S|] = \wt{O}(1)$; as such  for every $\tree'$ the expected number of \textsf{KDE} estimates that change -- and therefore the number of paths that need to be resampled -- is   $\wt{O}(1)$.  As there are at most $\lceil \log  n \rceil$ trees $\tree'$ that are updated when $\mathbf{z}$ is added, it holds that  $\mathbb{E}[|\mathcal{A}|] = \wt{O}(1)$.  
\end{proof}

Next  we   bound the size of the set $\mathcal{B}_{\mathbf{x}_i}$ that keeps track of the neighbours $\mathbf{x}_j$ of $\mathbf{x}_i$ in the approximate similarity graph $G$ that have higher degree.

\begin{lemma}\label{lemma:changed_neighbors}
     It holds with high probability for all $\mathbf{x}_i \in X$ that
    $|\mathcal{B}_{\mathbf{x}_i}| \leq 14\cdot L$.
\end{lemma}
\begin{proof}
   We first notice that  
    $
    \mathcal{B}_{\mathbf{x}_i} = \left\{\mathbf{x}_j \in X \mid \tree.\textsf{kde}.\hat{\mu}_{\mathbf{x}_j} > \tree.\textsf{kde}.\hat{\mu}_{\mathbf{x}_i} \text{ and } i \in Y_{\mathbf{x}_j}\right\},
    $
    where $Y_{\mathbf{x}_j} \triangleq \{y_{j,1}, \ldots y_{j,L}\}$ are the indices corresponding to the sampled neighbours of $\mathbf{x}_j$.
 For every pair of indices $i, j$, and   every $1 \leq \ell \leq L$, we define the random variable $Z_{i, j, \ell}$ to be 1 if $j$ is the neighbour sampled from $i$ at iteration $\ell$, and $0$ otherwise, i.e.,
    \[
    Z_{i, j, \ell} \triangleq \twopartdefow{1}{y_{i, \ell} = j}{0}
    \]
    We fix an arbitrary $\mathbf{x}_i$, and notice that
    \begin{equation}\label{eq:upper_bound_high_deg_neighbors}
    |\mathcal{B}_{\mathbf{x}_i}| = \sum_{\ell=1}^{L}\sum_{\substack{j=1 \\ \hat{\mu}_{\mathbf{x}_j} > \hat{\mu}_{\mathbf{x}_i}}}^n Z_{j, i, \ell},
    \end{equation}
    where for ease of notation we set $\hat{\mu}_{\mathbf{x}_i} \triangleq \tree.\textsf{kde}.\hat{\mu}_{\mathbf{x}_i}$ and $\hat{\mu}_{\mathbf{x}_j} \triangleq \tree.\textsf{kde}.\hat{\mu}_{\mathbf{x}_j}$   to be the \textsf{KDE} estimates at the root $\tree$.
    We have that 
    \begin{align}
\mathbb{E}\left[\sum_{\ell=1}^{L}\sum_{\substack{j=1 \\ \hat{\mu}_{\mathbf{x}_j} > \hat{\mu}_{\mathbf{x}_i}}}^n Z_{j, i, \ell}\right] &= \sum_{\ell=1}^{L}\sum_{\substack{j=1 \nonumber\\ \hat{\mu}_{\mathbf{x}_j} > \hat{\mu}_{\mathbf{x}_i}}}^n \mathbb{E}[Z_{j, i, \ell}] \nonumber\\
    &= \sum_{\ell=1}^{L}\sum_{\substack{j=1 \\ \hat{\mu}_{\mathbf{x}_j} > \hat{\mu}_{\mathbf{x}_i}}}^n \mathbb{P}[y_{j,\ell} = i] \nonumber\\
    &\leq \sum_{\ell=1}^{L}\sum_{\substack{j=1 \\ \hat{\mu}_{\mathbf{x}_j} > \hat{\mu}_{\mathbf{x}_i}}}^n \frac{2k(\mathbf{x}_i, \mathbf{x}_j)}{\deg_\graphk(\mathbf{x}_j)} \nonumber\\
    &< \sum_{\ell=1}^{L}\sum_{\substack{j=1 \\ \hat{\mu}_{\mathbf{x}_j} > \hat{\mu}_{\mathbf{x}_i}}}^n \frac{4k(\mathbf{x}_i, \mathbf{x}_j)}{\deg_\graphk(\mathbf{x}_i)}. \nonumber\\
    & \leq 4 \cdot L. \label{eq:upper_bound_high_deg_neighbors_expectation}
    \end{align}
    Here, the second last   inequality holds by the fact that 
    \[
    \deg_\graphk(\mathbf{x}_j) \geq \frac{\hat{\mu}_{\mathbf{x}_j}}{1+\varepsilon} > \frac{\hat{\mu}_{\mathbf{x}_i}}{1+\varepsilon} \geq \frac{(1-\varepsilon)\deg_\graphk(\mathbf{x}_i)}{1+\varepsilon} \geq \frac{\deg_\graphk(\mathbf{x}_i)}{2},
    \]
    where the last inequality follows by our choice of  $\varepsilon \leq 1/6$. Employing the same analysis, we have   that 
    \begin{align*}
        R &= \sum_{\ell=1}^{L}\sum_{\substack{j=1 \\ \hat{\mu}_{\mathbf{x}_j} > \hat{\mu}_{\mathbf{x}_i}}}^n \mathbb{E}\left[Z_{j, i, \ell}^2\right] = \sum_{\ell=1}^{L}\sum_{\substack{j=1 \\ \hat{\mu}_{\mathbf{x}_j} > \hat{\mu}_{\mathbf{x}_i}}}^n \mathbb{P}[y_{j,\ell} = i]  \leq 4 \cdot L.
    \end{align*}
    We apply the Bernstein's inequality, and  have that
\begin{align*}
    \mathbb{P}\left[ \left|\sum_{\ell=1}^{L}\sum_{\substack{j=1 \\ \hat{\mu}_{\mathbf{x}_j} > \hat{\mu}_{\mathbf{x}_i}}}^n Z_{j, i, \ell} - \mathbb{E}\left[\sum_{\ell=1}^{L}\sum_{\substack{j=1 \\ \hat{\mu}_{\mathbf{x}_j} > \hat{\mu}_{\mathbf{x}_i}}}^n Z_{j, i, \ell}\right]\right| \geq 10 L\right] 
    & \leq 2 \exp\left(- \frac{100L^2 / 2}{4L + 10\cdot  L/3}  \right) \\
    & = 2 \exp\left( - \frac{75L}{22}  \right)   \\
    & = o(1/n).
\end{align*}
 Hence,  by   the union bound, it holds with high probability for all $\mathbf{x}_i \in X$ that 
\[
\left|
\sum_{\ell=1}^{L}\sum_{\substack{j=1 \\ \hat{\mu}_{\mathbf{x}_j} > \hat{\mu}_{\mathbf{x}_i}}}^n Z_{j, i, \ell} - \mathbb{E}\left[\sum_{\ell=1}^{L}\sum_{\substack{j=1 \\ \hat{\mu}_{\mathbf{x}_j} > \hat{\mu}_{\mathbf{x}_i}}}^n Z_{j, i, \ell}\right]\right| < 10 L;
\]
combining this with  \eqref{eq:upper_bound_high_deg_neighbors} and~\eqref{eq:upper_bound_high_deg_neighbors_expectation}, we have  with high probability that
$
\left| |\mathcal{B}_{\mathbf{x}_i}| - 4L \right| < 10L$, 
which implies that $|\mathcal{B}_{\mathbf{x}_i}| < 14L$. 
\end{proof}

We are now ready to prove the running time guarantee of the update step.

\begin{lemma}\label{lem:update_time_dynamic_CPS}
The  expected running time of $\textsc{UpdateGraph}(G,\mathcal{T},\mathbf{z})$ is $n^{o(1)}\cdot \mathrm{cost}(k)$.
\end{lemma}

\begin{proof} 
    We   analyse  the running time of $\textsc{UpdateGraph}(G,\mathcal{T},\mathbf{z})$  step by step. 
    \begin{itemize}[leftmargin=1cm]
        \item  The \textsc{AddDataPointTree} procedure is dominated by the call to the \textsc{AddDataPoint} procedure on Line~\ref{alg:update_CPS_tree:update_kde} of Algorithm~\ref{alg:update_CPS_tree}, which takes $\epsilon^{-2}\cdot n^{o(1)}\cdot \mathrm{cost}(k)$ time  by Lemma~\ref{lem:dynamic_cps_tree_update_total_paths}.  
        \item   Next, we analyse the running time of sampling $L$ new neighbours of the new data point $\mathbf{z}$ (Lines~\ref{alg:update_CPS:start_sample_new_vertex}--\ref{alg:update_CPS:end_sample_new_vertex}). The algorithm samples a neighbour $\mathbf{x}_j$ using the $\textsc{Sample}$ procedure, which takes $\epsilon^{-2}\cdot n^{o(1)}\cdot \mathrm{cost}(k)$ time (Lemma~\ref{lem:running_time_sample_procedures}). To add the edge $(\mathbf{z}, \mathbf{x}_j)$, the algorithm computes the \textsf{KDE} estimate $\tree.\textsf{kde}.\hat{\mu}_{\mathbf{z}}$, which takes $\epsilon^{-2}\cdot n^{o(1)}\cdot \mathrm{cost}(k)$ time, and the weight value $k(\mathbf{z}, \mathbf{x}_j)$ which takes $O(d) = \wt{O}(1)$ time. Since $L  = \wt{O}(1)$, the total running time of Lines~\ref{alg:update_CPS:start_sample_new_vertex}--\ref{alg:update_CPS:end_sample_new_vertex} is $\epsilon^{-2}\cdot n^{o(1)}\cdot \mathrm{cost}(k)$.
        \item  For Lines~\ref{alg:update_CPS:start_reweight_neighbors}--\ref{alg:update_CPS:end_reweight_neighbors}, first note that the expected number of paths that need to be resampled is $\mathbb{E}[|\mathcal{A}|] = \wt{O}(1)$ (Lemma~\ref{lem:dynamic_cps_tree_update_total_paths}), and the expected number of points $\mathbf{x}_i$ such that $\tree.\textsf{kde}.\hat{\mu}_{\mathbf{x}_i}$ has changed is  $\wt{O}(1)$. Since  by Lemma~\ref{lemma:changed_neighbors} it holds with high probability that $|\mathcal{B}_{\mathbf{x}_i}| \leq 4\cdot L = \wt{O}(1)$, the total expected running time of Lines~\ref{alg:update_CPS:start_reweight_neighbors}--\ref{alg:update_CPS:end_reweight_neighbors} is $\wt{O}(1)$.
        \item  Finally, we analyse the running time of Lines~\ref{alg:update_CPS:start_resample_edges}--\ref{alg:update_CPS:end_resample_edges}. The running time of removing all the stored data about the path $\mathcal{P}_{\mathbf{x}_i, \ell}$ that needs to be resampled (Lines~\ref{alg:update_CPS:start_remove_edges}--\ref{alg:update_CPS:end_remove_edges}) is dominated by the time needed for removing all the stored information about $\mathbf{x}_i$ in $\tree^*.\textsf{left}.\textsf{kde}$ and $\tree^*.\textsf{right}.\textsf{kde}$ for every $\tree^*$ (Line~\ref{alg:update_CPS:start_remove_edges}). Doing this for all $\tree^*$ takes $\epsilon^{-2}\cdot n^{o(1)}\cdot \mathrm{cost}(k)$ time, since there are $O(\log n )$ such trees $\tree^*$ and in the data structures $\tree^*.\textsf{left}.\textsf{kde}$ and $\tree^*.\textsf{right}.\textsf{kde}$, $\mathbf{x}_i$ is removed from all buckets $B^*_{H_{\mu_i, a, j, \ell}}(\mathbf{x}_i)$, and there are $\epsilon^{-2}\cdot n^{o(1)}\cdot \mathrm{cost}(k)$ such buckets.
        The running time of the rest of the loop (Lines~\ref{alg:update_CPS:start_resample_edge}--\ref{alg:update_CPS:end_resample_edges}) is dominated by the running time for resampling a path $\mathcal{P}_{\mathbf{x}_i, \ell}$, which is $\epsilon^{-2}\cdot n^{o(1)}\cdot \mathrm{cost}(k)$ (Lemma~\ref{lem:running_time_sample_procedures}). Therefore, by the fact that $\mathbb{E}[|\mathcal{A}|] = \wt{O}(1)$ (Lemma~\ref{lem:dynamic_cps_tree_update_total_paths}), the total expected running time of Lines~\ref{alg:update_CPS:start_resample_edges}--\ref{alg:update_CPS:end_resample_edges} is $\epsilon^{-2}\cdot n^{o(1)}\cdot \mathrm{cost}(k)$.
    \end{itemize}
  Combining everything together proves the lemma.
\end{proof}

\subsubsection{Proof of   Correctness}

\begin{lemma}\label{lemma:correctnes_update_CPS}
   Let $G' = (X \cup \mathbf{z}, E', w_{G'})$ be the updated graph after running $\textsc{UpdateGraph}(G, \mathcal{T}, \mathbf{z})$    for the new arriving  $\mathbf{z}$. Then,  it holds with probability at least $9/10$ that $G'$ is an approximate similarity graph on   $X \cup \mathbf{z}$.
\end{lemma}
\begin{proof}
    We   prove this statement by showing that  running $\textsc{ConstructGraph}(X)$     followed by $\textsc{UpdateGraph}(G, \mathcal{T}, \mathbf{z})$      is equivalent to running $\textsc{ConstructGraph}(X\cup\mathbf{z})$. 
    \begin{itemize}[leftmargin=1cm]
        \item  First, we prove that the structure of the tree $\tree$ is the same in both settings: when running $\textsc{ConstructGraph}(X)$, we ensure that the tree $\tree$ is a complete binary tree. Then, when inserting a data point $\mathbf{z}$ using the $\textsc{AddDataPoint}(\mathbf{z})$ procedure on Line~\ref{alg:update_CPS:paths_to_resample} of Algorithm~\ref{alg:update_CPS}, $\mathbf{z}$ is inserted appropriately (by the condition on Line~\ref{alg:update_CPS_tree:add_in_tree} of Algorithm~\ref{alg:update_CPS_tree}) such that the updated tree is also a complete binary tree. Therefore, the structure of the tree $\tree$ is identical in both settings. 
        \item  Next, on Line~\ref{alg:update_CPS_tree:update_kde}  of Algorithm~\ref{alg:update_CPS},  $\mathbf{z}$ is added to the relevant $\tree'.\textsf{kde}$ dynamic \textsf{KDE} data structures using the $\textsc{AddDataPoint}(\mathbf{z})$ procedure of Algorithm~\ref{alg:dynamic_kde}. This ensures that the stored data points $\tree'.\textsf{kde}.X$ at every internal node $\tree'$ are identical in both settings and, by the guarantees of the dynamic \textsf{KDE} data structures (Theorem~\ref{thm:incremental_dynamic_kde}), the query estimates $\tree'.\textsf{kde}.\hat{\mu}_{\mathbf{q}}$ for every internal node $\tree'$ and any $\mathbf{q} \in \tree'.\textsf{kde}.Q$ are the same in both settings.
        \item  For the new data point $\mathbf{z}$, we sample $L$ new neigbours (Lines~\ref{alg:update_CPS:start_sample_new_vertex}--\ref{alg:update_CPS:end_sample_new_vertex} of Algorithm~\ref{alg:update_CPS}). By the previous points, it holds that the tree $\tree$ is identical in both settings, and therefore the sampling procedure on Lines~\ref{alg:update_CPS:start_sample_new_vertex}--\ref{alg:update_CPS:end_sample_new_vertex} in Algorithm~\ref{alg:update_CPS} for the new data point $\mathbf{z}$ is equivalent to the sampling procedure on Lines~\ref{alg:line:construct_sparsifier:start_sampling}--\ref{alg:line:construct_sparsifier:reweight_factor} of Algorithm~\ref{alg:initialise_CPS} for the point $\mathbf{z}$ when executing $\textsc{Initialise}(X \cup \mathbf{z}, \varepsilon)$.
        \item  Then, for any data point $\mathbf{x}_i \in X$, let $(\mathbf{x}_i, \mathbf{x}_j) \in E$ be one of its sampled neighbours edge after running $\textsc{ConstructGraph}$. It holds that the scaling factor for the edge weight $w_G(\mathbf{x}_i, \mathbf{x}_j)$ is \[\widehat{w}(i,j) = \frac{L \cdot k(\mathbf{x}_i, \mathbf{x}_j) }{ \min\{\tree.\textsf{kde}.\hat{\mu}_{\mathbf{x}_i}, \tree.\textsf{kde}.\hat{\mu}_{\mathbf{x}_j}\}}.\] Notice that after running $\textsc{UpdateG}(\mathbf{z})$, the scaling factor $w_G(\mathbf{x}_i, \mathbf{x}_j)$ can change due to a change in $\min\{\tree.\textsf{kde}.\hat{\mu}_{\mathbf{x}_i}, \tree.\textsf{kde}.\hat{\mu}_{\mathbf{x}_j}\}$.
        Without loss of generality, let $\min\{\tree.\textsf{kde}.\hat{\mu}_{\mathbf{x}_i}, \tree.\textsf{kde}.\hat{\mu}_{\mathbf{x}_j}\} = \tree.\textsf{kde}.\hat{\mu}_{\mathbf{x}_j}$. By Line~\ref{alg:line:construct_sparsifier:add_min_neighbor} of Algorithm~\ref{alg:initialise_CPS},  in this case we have $\mathbf{x}_i \in \mathcal{B}_{\mathbf{x}_j}$.  We further distinguish between the  two cases:
        \begin{enumerate}
            \item  If $\tree.\textsf{kde}.\hat{\mu}_{\mathbf{x}_i}$ changes after running $\textsc{UpdateGraph}(\mathbf{z})$, then by the $\textsc{AddDataPointTree}(\tree, \mathbf{z})$ procedure  all the paths $\mathcal{P}_{\mathbf{x}_i, \ell}$ for $1 \leq \ell \leq L$ will be resampled and updated on Lines~\ref{alg:update_CPS:start_remove_edges}--\ref{alg:update_CPS:end_resample_edges}.
            \item  On the other hand, if $\tree.\textsf{kde}.\hat{\mu}_{\mathbf{x}_j}$ changes and $\tree.\textsf{kde}.\hat{\mu}_{\mathbf{x}_i}$ does not, then the paths $\mathcal{P}_{\mathbf{x}_j, \ell'}$ ending  at the leaf corresponding to  $\mathbf{x}_i$ are not necessarily resampled. In this case, the scaling factor is updated on Lines~\ref{alg:update_CPS:start_reweight_neighbors}--\ref{alg:update_CPS:end_reweight_neighbors}, and therefore $w_G(\mathbf{x}_i, \mathbf{x}_j)$ is appropriately rescaled.
        \end{enumerate}
        \item  Let ${\mathcal{P}^*_{\mathbf{x}_i, \ell}} \in \tree.\textsf{paths}$ be any sampling path that is not resampled, i.e., ${\mathcal{P}^*_{\mathbf{x}_i, \ell}} \notin \mathcal{A}$. This implies that the \textsf{KDE} estimate of $\tree'.\textsf{kde}.\hat{\mu}_{\mathbf{x}_i}$ does not change at any internal $\tree'$ where $\mathcal{P}_{\mathbf{x}_i, \ell}$ is stored, and therefore the sampling procedure for $\mathcal{P}^*_{\mathbf{x}_i, \ell}$ is identical in both settings.
        \item  Finally, let $\mathcal{P}_{\mathbf{x}_i, \ell} \in \mathcal{A}$ be a sampling path that is resampled, and   $(\mathbf{x}_i, \mathbf{x}_j)$ be the sampled edge (contribution) corresponding to $\mathcal{P}_{\mathbf{x}_i, \ell}$. Before resampling the path $\mathcal{P}_{\mathbf{x}_i, \ell}$ starting from $\tree'$, on Lines~\ref{alg:update_CPS:start_remove_edges}--\ref{alg:update_CPS:end_remove_edges} the algorithm removes the stored paths $\mathcal{P}_{\mathbf{x}_i, \ell}$ and query points $\mathbf{x}_i$ from every internal node $\tree^*$ below $\tree'$, and removes the weight contribution to $w_G(\mathbf{x}_i, \mathbf{x}_j)$ from $\mathcal{P}_{\mathbf{x}_i, \ell}$. Then, on Lines~\ref{alg:update_CPS:start_resample_edge}--\ref{alg:update_CPS:end_resample_edges}, we resample a new edge $(\mathbf{x}_i, \mathbf{x}^*_j)$, in an equivalent manner as sampling a new edge when running Lines~\ref{alg:line:construct_sparsifier:start_sampling}--\ref{alg:line:construct_sparsifier:reweight_factor} of Algorithm~\ref{alg:initialise_CPS}.  Therefore, the resampling procedure for the path $\mathcal{P}_{\mathbf{x}_i, \ell}$ is identical to the sampling procedure for $\mathcal{P}_{\mathbf{x}_i, \ell}$ when running $\textsc{Initialise}(X \cup \mathbf{z}, \varepsilon)$, because the resampling procedure uses the updated \textsf{KDE} estimates at each internal node $\tree'$, which are identical to the \textsf{KDE} estimates that would be computed in $\textsc{Initialise}(X \cup \mathbf{z}, \varepsilon)$.
    \end{itemize}
     Combining everything together proves the lemma. 
\end{proof}

Finally, we are ready to prove the second statement of Theorem~\ref{thm:dynamic_cps}.

\begin{proof}[Proof of the Second Statement of Theorem~\ref{thm:dynamic_cps}]
Lemma~\ref{lem:update_time_dynamic_CPS} shows the time complexity of $\textsc{UpdateGraph}(G,\mathcal{T},\mathbf{z})$, and Lemma~\ref{lemma:correctnes_update_CPS} shows the correctness of our updated procedures. Combining these two facts together proves the second statement of Theorem~\ref{thm:dynamic_cps}.
\end{proof}

\section{Additional Experimental Results} \label{app:experiments}

In this section, we provide some more details about our experimental setup and give some additional experimental results.
Table~\ref{tab:full-dataset-info} provides additional information about all of the datasets used in our experiments.
\begin{table}[h]
  \caption{Datasets used for experimental evaluation. $n$ is the number of data points, $d$ is the dimension, and $\sigma$ is the parameter we use in the Gaussian kernel. \newline}
  \label{tab:full-dataset-info}
  \centering
  \resizebox{\textwidth}{!}{\begin{tabular}{lllllll}
    \toprule 
    \textbf{Dataset} & $\mathbf{n}$ & $\mathbf{d}$ & $\sigma$ & \textbf{License} & \textbf{Reference} & \textbf{Description} \\
    \midrule
    blobs &  20,000 &  10 & 0.01 & BSD & \citep{scikit-learn} & \begin{tabular}{@{}c@{}}Synthetic clusters from a  mixture \\ of Gaussian distributions.\end{tabular}\\
\midrule 
    cifar10 &  50,000 &  2,048 & 0.0001 & - & \citep{resnet,cifar} & \begin{tabular}{@{}c@{}}ResNet-50 embeddings of images.\end{tabular}\\
\midrule 
    mnist &  70,000 &  728 & 0.000001 & CC BY-SA 3.0 & \citep{lecun_mnist_1998} & \begin{tabular}{@{}c@{}}Images of handwritten digits.\end{tabular}\\
\midrule 
    shuttle &  58,000 &  9 & 0.01 & CC BY 4.0 & \citep{shuttle_dataset} & \begin{tabular}{@{}c@{}}Numerical data from NASA \\ space shuttle sensors.\end{tabular}\\
\midrule 
    aloi &  108,000 &  128 & 0.01 & - & \citep{aloi_dataset} & \begin{tabular}{@{}c@{}}Images of objects under a variety \\ of lighting conditions.\end{tabular}\\
\midrule 
    msd &  515,345 &  90 & 0.000001 & CC BY 4.0 & \citep{msd-main-ref} & \begin{tabular}{@{}c@{}}Numerical and categorical \\ features of songs.\end{tabular}\\
\midrule 
    covtype &  581,012 & 54 &  0.000005 & CC BY 4.0 & \citep{covtype_dataset} & \begin{tabular}{@{}c@{}}Cartographic features used to predict \\ forest cover type.\end{tabular}\\
\midrule 
    glove &  1,193,514 &  100 & 0.1 & PDDL 1.0 & \citep{glove_dataset} & \begin{tabular}{@{}c@{}}Word embedding vectors.\end{tabular}\\
\midrule 
    census &  2,458,285 &  68 & 0.01 & CC BY 4.0 & \citep{us_census_data} & \begin{tabular}{@{}c@{}}Categorical and numerical data from \\ the 1990 US census.\end{tabular}\\
    \bottomrule
  \end{tabular}}
\end{table}

\subsection{Dynamic \textsf{KDE} Experiments}
Tables~\ref{tab:dynamic_kde_app}~and~\ref{tab:dynamic_kde_exact} show the experimental evaluation of the dynamic \textsf{KDE} algorithms on several additional datasets.
The results demonstrate that our algorithm scales better to larger datasets than the baseline algorithms. Figures~\ref{fig:kde_all_errors}~and~\ref{fig:kde_all_times} show the relative errors and running times for all iterations, datasets, and algorithms for the dynamic \textsf{KDE} experiments.

\newcommand{\kdeappcaption}{Experimental results for dynamic \textsf{KDE}. For each dataset, the shaded results correspond to the algorithm with the lowest total running time.}
\begin{table}[h]
  \caption{\kdeappcaption \newline\label{tab:dynamic_kde_app}}
  \centering
  
  \begin{tabular}{ccccccc}
    \toprule
    & \multicolumn{2}{c}{\textsc{CKNS}} & \multicolumn{2}{c}{\textsc{DynamicRS}} & \multicolumn{2}{c}{\textsc{Our Algorithm}}\\ 
    \cmidrule(lr){2-3} \cmidrule(lr){4-5} \cmidrule(lr){6-7}
    dataset & Time (s) & Err & Time (s) & Err & Time (s) & Err\\ 
    \midrule
    shuttle & $ 32.9{\scriptstyle \pm  2.1}$ & $ 0.146{\scriptstyle \pm  0.002}$ & \cellcolor{gray!25}$ 0.8{\scriptstyle \pm  0.0}$ & \cellcolor{gray!25}$ 0.078{\scriptstyle \pm  0.005}$ & $ 10.9{\scriptstyle \pm  0.3}$ & $ 0.159{\scriptstyle \pm  0.024}$ \\
    aloi & $ 619.0{\scriptstyle \pm  10.7}$ & $ 0.050{\scriptstyle \pm  0.006}$ & \cellcolor{gray!25}$ 19.7{\scriptstyle \pm  0.3}$ & \cellcolor{gray!25}$ 0.010{\scriptstyle \pm  0.003}$ & $ 46.9{\scriptstyle \pm  0.7}$ & $ 0.060{\scriptstyle \pm  0.021}$ \\
    msd & $ 14,360.0{\scriptstyle \pm  0.0}$ & $ 0.385{\scriptstyle \pm  0.000}$ & $ 1,887.8{\scriptstyle \pm  0.0}$ & $ 5.430{\scriptstyle \pm  0.000}$ & \cellcolor{gray!25}$ 306.4{\scriptstyle \pm  0.0}$ & \cellcolor{gray!25}$ 0.388{\scriptstyle \pm  0.000}$ \\
    covtype & $ 5,650.3{\scriptstyle \pm  109.0}$ & $ 0.159{\scriptstyle \pm  0.002}$ & $ 309.2{\scriptstyle \pm  2.4}$ & $ 0.018{\scriptstyle \pm  0.003}$ & \cellcolor{gray!25}$ 151.7{\scriptstyle \pm  4.5}$ & \cellcolor{gray!25}$ 0.196{\scriptstyle \pm  0.017}$ \\
    glove & $ 2,640.8{\scriptstyle \pm  1677.7}$ & $ 0.221{\scriptstyle \pm  0.229}$ & $ 1,038.6{\scriptstyle \pm  26.5}$ & $ 0.004{\scriptstyle \pm  0.005}$ & \cellcolor{gray!25}$ 445.6{\scriptstyle \pm  214.6}$ & \cellcolor{gray!25}$ 0.296{\scriptstyle \pm  0.469}$ \\
    census & $ 10,471.5{\scriptstyle \pm  160.6}$ & $ 0.080{\scriptstyle \pm  0.000}$ & $ 3,424.8{\scriptstyle \pm  5.2}$ & $ 0.005{\scriptstyle \pm  0.001}$ & \cellcolor{gray!25}$ 836.5{\scriptstyle \pm  44.6}$ & \cellcolor{gray!25}$ 0.102{\scriptstyle \pm  0.021}$ \\
    \bottomrule
  \end{tabular}
\end{table}

\newcommand{\kdeexactcaption}{Running time for dynamic \textsf{KDE} with the exact algorithm.}
\begin{table}[t]
  \caption{\kdeexactcaption \newline \label{tab:dynamic_kde_exact}}
  \centering
  \begin{tabular}{ll}\toprule
    Dataset & Running Time \\
 \midrule 
    shuttle & $ 4.1{\scriptstyle \pm  0.1}$ \\
    aloi & $ 164.5{\scriptstyle \pm  13.6}$ \\
    msd & $ 2,715.6{\scriptstyle \pm  0.0}$ \\
    covtype & $ 2,349.8{\scriptstyle \pm  101.2}$ \\
    glove & $ 5,251.7{\scriptstyle \pm  0.0}$ \\
    census & $ 16,202.6{\scriptstyle \pm  154.6}$ \\
    \bottomrule
  \end{tabular}
\end{table}

\begin{figure}[ht]
    \centering
    \begin{subfigure}[b]{0.38\textwidth}
        \begin{center}\includegraphics[width=\textwidth]{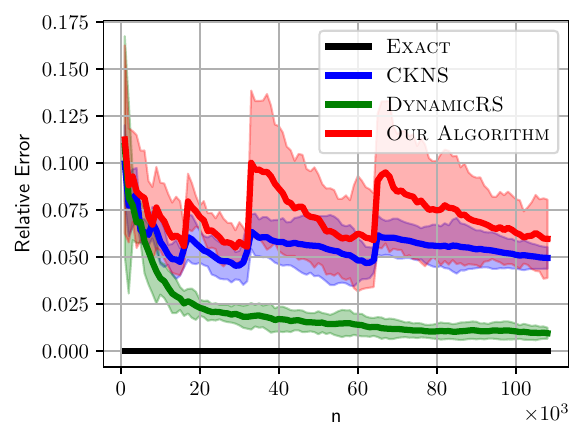}
        \caption{ALOI Relative Error}
        \end{center}
    \end{subfigure}
    \hspace{2cm}
    \begin{subfigure}[b]{0.38\textwidth}
        \begin{center}\includegraphics[width=\textwidth]{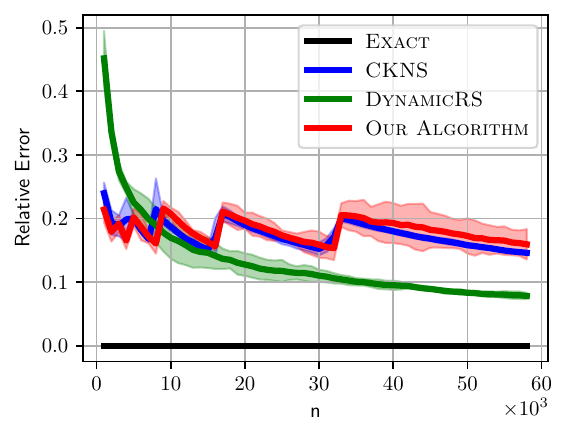}
        \caption{SHUTTLE Relative Error}
        \end{center}
    \end{subfigure}
    \begin{subfigure}[b]{0.38\textwidth}
\begin{center}\includegraphics[width=\textwidth]{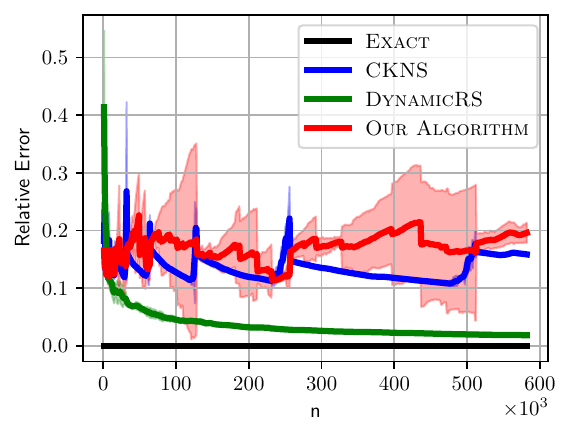}
        \caption{COVTYPE Relative Error}
        \end{center}
    \end{subfigure}
    \hspace{2cm}
    \begin{subfigure}[b]{0.38\textwidth}
\begin{center}\includegraphics[width=\textwidth]{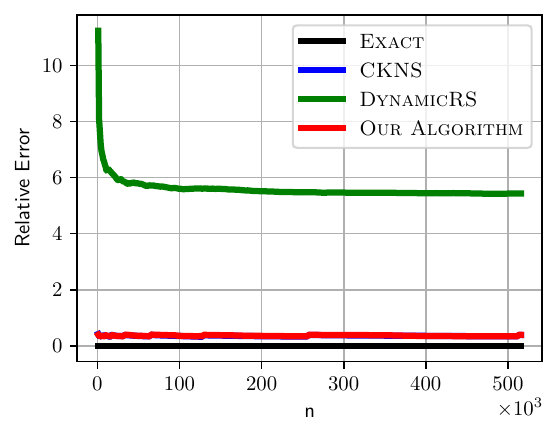}
        \caption{MSD Relative Error}
        \end{center}
    \end{subfigure}
    \begin{subfigure}[b]{0.38\textwidth}
        \begin{center}\includegraphics[width=\textwidth]{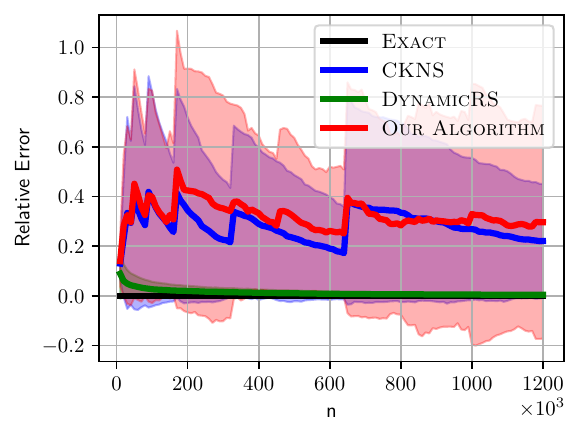}
        \caption{GLOVE Relative Error}
        \end{center}
    \end{subfigure}
    \hspace{2cm}
    \begin{subfigure}[b]{0.38\textwidth}
        \begin{center}\includegraphics[width=\textwidth]{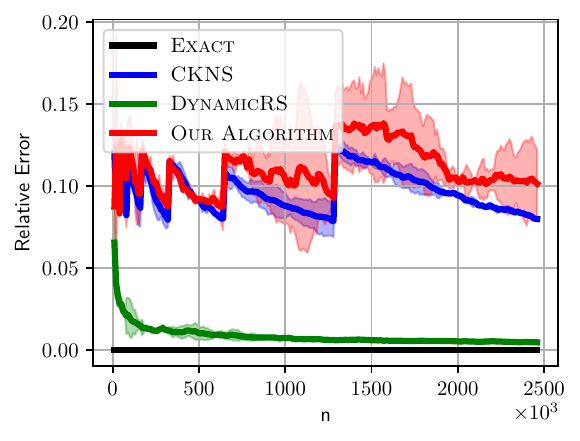}
        \caption{CENSUS Relative Error}
        \end{center}
    \end{subfigure}
    \caption{Relative errors for all datasets.
    \label{fig:kde_all_errors}}
\end{figure}

\begin{figure}[ht]
    \centering
    \begin{subfigure}[b]{0.38\textwidth}
        \begin{center}\includegraphics[width=\textwidth]{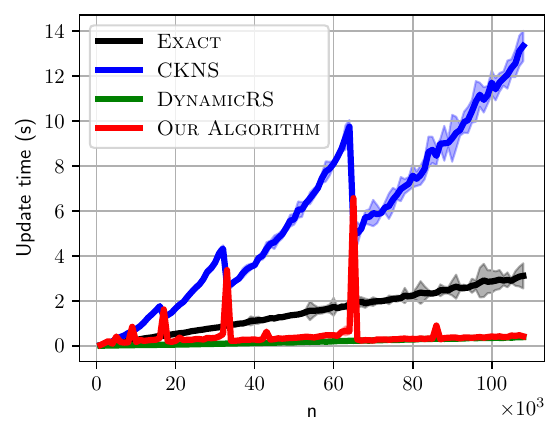}
        \caption{ALOI Update Time}
        \end{center}
    \end{subfigure}
    \hspace{2cm}
    \begin{subfigure}[b]{0.38\textwidth}
        \begin{center}\includegraphics[width=\textwidth]{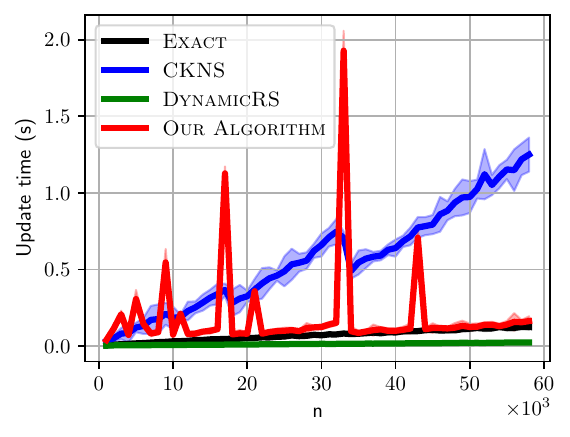}
        \caption{SHUTTLE Update Time}
        \end{center}
    \end{subfigure}
    \begin{subfigure}[b]{0.38\textwidth}
        \begin{center}\includegraphics[width=\textwidth]{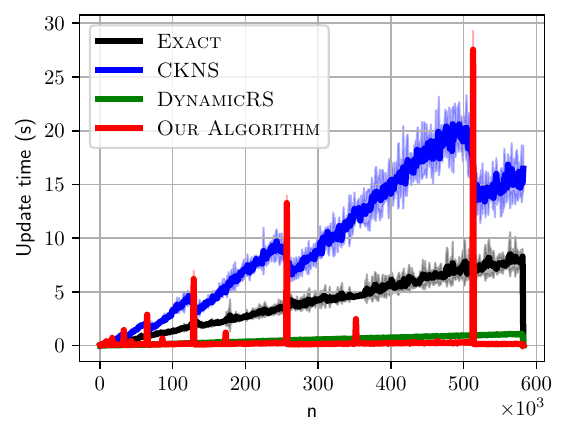}
        \caption{COVTYPE Update Time}
        \end{center}
    \end{subfigure}
    \hspace{2cm}
    \begin{subfigure}[b]{0.38\textwidth}
        \begin{center}\includegraphics[width=\textwidth]{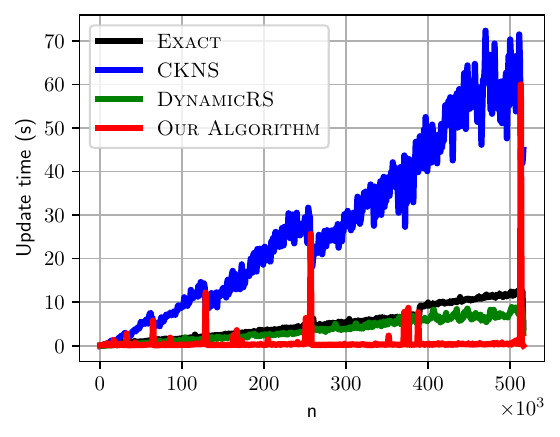}
        \caption{MSD Update Time}
        \end{center}
    \end{subfigure}
    \begin{subfigure}[b]{0.38\textwidth}
        \begin{center}\includegraphics[width=\textwidth]{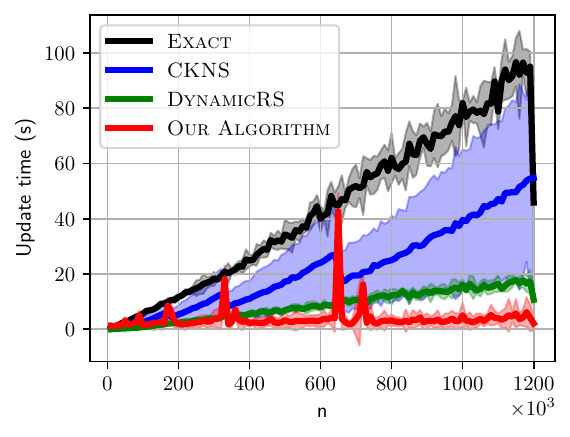}
        \caption{GLOVE Update Time}
        \end{center}
    \end{subfigure}
    \hspace{2cm}
    \begin{subfigure}[b]{0.38\textwidth}
        \begin{center}\includegraphics[width=\textwidth]{figures/census_update_time.pdf}
        \caption{CENSUS Update Time}
        \end{center}
    \end{subfigure}
    \caption{Running times for all data sets.
    \label{fig:kde_all_times}}
\end{figure}

\FloatBarrier
\subsection{Plots for Dynamic Similarity Graph Experiments}
Table~\ref{tab:dynamic_sg_app} shows the results of the experiments for the dynamic similarity graph, evaluated with the Adjusted Rand Index~(ARI)~\citep{rand1971objective}.

\newcommand{\simcaptionapp}{ARI values for the dynamic similarity graph experiments.}
\begin{table}[h]
  \caption{\simcaptionapp \newline\label{tab:dynamic_sg_app}}
  \centering
  \resizebox{0.8\textwidth}{!}{
  \begin{tabular}{lllllll}
    \toprule
    & \multicolumn{2}{c}{\textsc{FullyConnected}} & \multicolumn{2}{c}{\textsc{kNN}} & \multicolumn{2}{c}{\textsc{Our Algorithm}}\\ 
    \cmidrule(lr){2-3} \cmidrule(lr){4-5} \cmidrule(lr){6-7}
    dataset & Time (s) & ARI & Time (s) & ARI & Time (s) & ARI\\ 
    \midrule
    blobs & $ 72.8{\scriptstyle \pm  2.2}$ & $ 1.000{\scriptstyle \pm  0.000}$ & $ 383.6{\scriptstyle \pm  3.9}$ & $ 0.797{\scriptstyle \pm  0.287}$ & \cellcolor{gray!25}$ 21.2{\scriptstyle \pm  0.8}$ & \cellcolor{gray!25}$ 1.000{\scriptstyle \pm  0.000}$ \\
    cifar10 & $ 19,158.2{\scriptstyle \pm  231.6}$ & $ 0.000{\scriptstyle \pm  0.000}$ & $ 3,503.0{\scriptstyle \pm  490.6}$ & $ 0.098{\scriptstyle \pm  0.001}$ & \cellcolor{gray!25}$ 1,403.5{\scriptstyle \pm  152.4}$ & \cellcolor{gray!25}$ 0.221{\scriptstyle \pm  0.013}$ \\
    mnist & \cellcolor{gray!25}$ 1,328.3{\scriptstyle \pm  159.5}$ & \cellcolor{gray!25}$ 0.149{\scriptstyle \pm  0.000}$ & $ 5,796.3{\scriptstyle \pm  234.3}$ & $ 0.673{\scriptstyle \pm  0.001}$ & $ 1,470.3{\scriptstyle \pm  77.9}$ & $ 0.238{\scriptstyle \pm  0.011}$ \\
    \bottomrule
  \end{tabular}}
\end{table}

\end{document}